\documentclass[english]{article}
\usepackage[T1]{fontenc}
\usepackage[utf8]{inputenc}
\usepackage{geometry}
\usepackage{color}
\usepackage{babel}
\usepackage{float}
\usepackage{textcomp}
\usepackage{amsmath}
\usepackage{amsthm}
\usepackage{amssymb}
\usepackage{graphicx}
\usepackage{setspace}
\usepackage[unicode=true,pdfusetitle,
 bookmarks=true,bookmarksnumbered=false,bookmarksopen=false,
 breaklinks=false,pdfborder={0 0 1},backref=false,colorlinks=false]
 {hyperref}
\usepackage{breakurl}

\makeatletter
\theoremstyle{plain}
\newtheorem{lem}{\protect\lemmaname}
\theoremstyle{plain}
\newtheorem{prop}{\protect\propositionname}
\theoremstyle{remark}
\newtheorem*{rem*}{\protect\remarkname}
\theoremstyle{plain}
\newtheorem{cor}{\protect\corollaryname}
\theoremstyle{definition}
\newtheorem{defn}{\protect\definitionname}
\theoremstyle{plain}
\newtheorem{conjecture}{\protect\conjecturename}

\@ifundefined{showcaptionsetup}{}{%
 \PassOptionsToPackage{caption=false}{subfig}}
\usepackage{subfig}
\makeatother

\usepackage{listings}
\providecommand{\conjecturename}{Conjecture}
\providecommand{\corollaryname}{Corollary}
\providecommand{\definitionname}{Definition}
\providecommand{\lemmaname}{Lemma}
\providecommand{\propositionname}{Proposition}
\providecommand{\remarkname}{Remark}

\begin{document}
\begin{doublespace}
\begin{center}
\textbf{\textcolor{black}{\large{}The Perfect Marriage and Much More:
Combining Dimension Reduction, Distance Measures and Covariance}}{\large\par}
\par\end{center}

\begin{center}
\textbf{Ravi Kashyap }
\par\end{center}

\begin{center}
\textbf{SolBridge International School of Business / City University
of Hong Kong }
\par\end{center}

\begin{center}
\begin{center}
\today
\par\end{center}
\par\end{center}

\begin{center}
\par\end{center}

\begin{center}
\par\end{center}

\begin{center}
Mathematics Subject Codes: 60E05 Distributions: general theory; 62-07
Data analysis; 62P20 Applications to economics
\par\end{center}

\begin{center}
\textbf{\textcolor{blue}{\href{https://doi.org/10.1016/j.physa.2019.04.174}{Edited Version: Kashyap, R. (2019). The perfect marriage and much more: Combining dimension reduction, distance measures and covariance. Physica A: Statistical Mechanics and its Applications, XXX, XXX-XXX. }}}
\par\end{center}

\begin{center}
\tableofcontents{}\pagebreak{}
\par\end{center}
\end{doublespace}
\begin{doublespace}

\section{Abstract}
\end{doublespace}

\begin{doublespace}
We develop a novel methodology based on the marriage between the Bhattacharyya
distance, a measure of similarity across distributions of random variables,
and the Johnson-Lindenstrauss Lemma, a technique for dimension reduction.
The resulting technique is a simple yet powerful tool that allows
comparisons between data-sets representing any two distributions.
The degree to which different entities, (markets, universities, hospitals,
cities, groups of securities, etc.), have different distance measures
of their corresponding distributions tells us the extent to which
they are different, aiding participants looking for diversification
or looking for more of the same thing. We demonstrate a relationship
between covariance and distance measures based on a generic extension
of Stein's Lemma. We consider an asset pricing application and then
briefly discuss how this methodology lends itself to numerous market-structure
studies and even applications outside the realm of finance / social
sciences by illustrating a biological application. We provide numerical
illustrations using security prices, volumes and volatilities of both
these variables from six different countries.
\end{doublespace}
\begin{doublespace}

\section{Introduction}
\end{doublespace}

\begin{doublespace}
The varying behavior of participants in a social system, which can
also be viewed as unpredictable actions, will give rise to unintended
consequences and as long as participants are free to observe the results
and modify their actions, this effect will persist (Kashyap 2017;
appendix \ref{sec:Appendix:-Uncertainty-and} has a discussion of
uncertainty and unintended consequences). Unintended consequences
and its siamese twin, uncertainty, are fostering a trend of collecting
data to improve decision making, which is perhaps leading to more
analysis of the data and more actions, leading to a cycle of increasing
data collection and actions, giving rise to information explosion
(Dordick \& Wang 1993; Korth \& Silberschatz 1997; Sweeney 2001; Fuller
2010; Major \& Savin-Baden 2010; Beath, Becerra-Fernandez, Ross \&
Short 2012). (Kashyap 2015) consider ways to reduce the complexity
of social systems, which could be one way to mitigate the effect of
unintended outcomes. While attempts at designing less complex systems
are worthy endeavors, reduced complexity might be hard to accomplish
in certain instances and despite successfully reducing complexity,
alternate techniques at dealing with uncertainty are commendable complementary
pursuits (Kashyap 2016).

While it might be possible to observe historical trends (or other
attributes) and make comparisons across fewer number of entities;
in large systems where there are numerous components or contributing
elements, this can be a daunting task. In this present paper, we present
quantitative measures across aggregations of smaller elements that
can aid decision makers by providing simple yet powerful metrics to
compare large groups of entities.

\textbf{\textit{We consider a measure of similarity, the Bhattacharyya
distance, across distributions of variables. We develop a novel methodology
based on the marriage between the Bhattacharyya distance and the Johnson-Lindenstrauss
Lemma (JL-Lemma), a technique for dimension reduction; providing us
with a simple yet powerful tool that allows comparisons between data-sets
representing any two distributions. The degree to which different
entities, (markets, universities, hospitals, countries, cities, farms,
forests, groups of securities, etc.), have different distance measures
of their corresponding distributions tells us the extent to which
they are different, aiding participants looking for diversification
or looking for more of the same thing. To calculate the distance measure
between any two random variables (both of which could be multi-variate
with either different number of component random variables or observations),
we utilize the JL-Lemma to reduce the dimensions of one of the random
variables (either the number of component random variables or the
number of observations) such that the two random variables have the
same number of dimensions allowing us to apply distance measures.}}
\end{doublespace}

\textbf{\textit{We demonstrate a relationship between covariance and
distance measures based on a generic extension of Stein's Lemma. We
consider an asset pricing application and also briefly discuss how
this methodology lends itself to numerous market-structure and asset
pricing studies and even applications outside the realm of finance
/ social sciences by illustrating a biological application (sections
\ref{sec:Marketstructure,-Microstructure-}, \ref{subsec:Asset-Pricing-Application},
\ref{subsec:Biological-Application}). We illustrate some ways in
which our techniques can be used in practice by comparing stock prices,
trading volumes, price volatilities and volume volatilities across
six markets (sections \ref{subsec:Comparison-of-Security}; \ref{sec:Comparison-Volumes-Prices-Volatilities}).}}

\begin{doublespace}
An unintended consequence of our efforts, has become a review of the
vast literature on distance measures and related statistical techniques,
which can be useful for anyone that attempts to apply the corresponding
techniques to the problems mentioned here and also to the many unmentioned
but possibly related ones. The results and the discussion draw upon
sources from statistics, probability, economics / finance, communication
systems, pattern recognition and information theory; becoming one
example of how elements of different fields can be combined to provide
answers to the questions raised by a particular field. All the propositions
are new results and they depend on existing results which are given
as lemmas without proof. Such an approach ensures that the results
are instructive and immediately applicable to a wider audience.
\end{doublespace}
\begin{doublespace}

\section{Literature Review of Methodological Fundamentals}
\end{doublespace}
\begin{doublespace}

\subsection{Bhattacharyya Distance}
\end{doublespace}

\begin{doublespace}
We use the Bhattacharyya distance (Bhattacharyya 1943; 1946) as a
measure of similarity or dissimilarity between the probability distributions
of the two entities we are looking to compare. These entities could
be two securities, groups of securities, markets or any statistical
populations that we are interested in studying. The Bhattacharyya
distance is defined as the negative logarithm of the Bhattacharyya
coefficient. 
\begin{equation}
D_{BC}\left(p_{i},p_{i}^{\prime}\right)=-\ln\left[\rho\left(p_{i},p_{i}^{\prime}\right)\right]
\end{equation}
Here, $D_{BC}\left(p_{i},p_{i}^{\prime}\right)$, is the Bhattacharyya
Distance between two multinomial populations each consisting of $k$
categories or classes with associated probabilities $p_{1},p_{2},...,p_{k}$
and $p_{1}^{\prime},p_{2}^{\prime},...,p_{k}^{\prime}$ respectively.
The Bhattacharyya coefficient, $\rho\left(p_{i},p_{i}^{\prime}\right)$,
is calculated as shown below for discrete and continuous probability
distributions. 
\begin{equation}
\rho\left(p_{i},p_{i}^{\prime}\right)=\sum_{i}^{k}\sqrt{p_{i}p_{i}^{\prime}}
\end{equation}
\begin{equation}
\rho\left(p_{i},p_{i}^{\prime}\right)=\int\sqrt{p_{i}\left(x\right)p_{i}^{\prime}\left(x\right)}dx
\end{equation}

Bhattacharyya’s original interpretation of the measure was geometric
(Derpanis 2008). He considered two multinomial populations each consisting
of $k$ categories or classes with associated probabilities $p_{1},p_{2},...,p_{k}$
and $p_{1}^{\prime},p_{2}^{\prime},...,p_{k}^{\prime}$ respectively.
Then, as $\sum_{i}^{k}p_{i}=1$ and $\sum_{i}^{k}p_{i}^{\prime}=1$,
he noted that $(\sqrt{p_{1}},...,\sqrt{p_{k}})$ and $(\sqrt{p_{1}^{\prime}},...,\sqrt{p_{k}^{\prime}})$
could be considered as the direction cosines of two vectors in $k-$dimensional
space referred to a system of orthogonal co-ordinate axes. As a measure
of divergence between the two populations Bhattacharyya used the square
of the angle between the two position vectors. If $\theta$ is the
angle between the vectors then: 
\begin{equation}
\rho\left(p_{i},p_{i}^{\prime}\right)=cos\theta=\sum_{i}^{k}\sqrt{p_{i}p_{i}^{\prime}}
\end{equation}
Thus if the two populations are identical: $cos\theta=1$ corresponding
to $\theta=0$, hence we see the intuitive motivation behind the definition
as the vectors are co-linear. Bhattacharyya further showed that by
passing to the limiting case a measure of divergence could be obtained
between two populations defined in any way given that the two populations
have the same number of variates. The value of coefficient then lies
between $0$ and $1$. 
\begin{equation}
0\leq\rho\left(p_{i},p_{i}^{\prime}\right)=\sum_{i}^{k}\sqrt{p_{i}p_{i}^{\prime}}=\sum_{i}^{k}p_{i}\sqrt{\frac{p_{i}^{\prime}}{p_{i}}}\leq\sqrt{\sum_{i}^{k}p_{i}^{\prime}}=1
\end{equation}
\begin{equation}
0\leq D_{BC}\left(p_{i},p_{i}^{\prime}\right)\leq\infty
\end{equation}
Here, the last inequality follows from Jensen's inequality. (Comaniciu,
Ramesh \& Meer 2003) modify this as shown below and prove that this
alternate measure, $d\left(p_{i},p_{i}^{\prime}\right)$ termed the
modified Bhattacharyya Metric, follows all the metric axioms: positive,
symmetric, is zero for the same two elements and satisfies the triangle
inequality. 
\begin{equation}
d\left(p_{i},p_{i}^{\prime}\right)=\sqrt{1-\rho\left(p_{i},p_{i}^{\prime}\right)}
\end{equation}
We get the following formulae (Lee and Bretschneider 2012) for the
Bhattacharyya distance when applied to the case of two univariate
normal distributions. 
\begin{equation}
D_{BC-N}(p,q)=\frac{1}{4}\ln\left(\frac{1}{4}\left(\frac{\sigma_{p}^{2}}{\sigma_{q}^{2}}+\frac{\sigma_{q}^{2}}{\sigma_{p}^{2}}+2\right)\right)+\frac{1}{4}\left(\frac{(\mu_{p}-\mu_{q})^{2}}{\sigma_{p}^{2}+\sigma_{q}^{2}}\right)
\end{equation}

Here, $D_{BC-N}(p,q)$ is the Bhattacharyya distance between $p$
and $q$ normal distributions or classes. $\sigma_{p}^{2}$ is the
variance of the $p-$th distribution, $\mu_{p}$ is the mean of the
$p-$th distribution, and $p,q$ are two different distributions.

The original paper on the Bhattacharyya distance (Bhattacharyya 1943)
mentions a natural extension to the case of more than two populations.
For an $M$ population system, each with $k$ random variates, the
definition of the coefficient becomes, 
\begin{equation}
\rho\left(p_{1},p_{2},...,p_{M}\right)=\int\cdots\int\left[p_{1}\left(x\right)p_{2}\left(x\right)...p_{M}\left(x\right)\right]^{\frac{1}{M}}dx_{1}\cdots dx_{k}
\end{equation}

For two multivariate normal distributions, $D_{BC-MN}\left(p_{1},p_{2}\right)$
is the Bhattacharyya distance between two multivariate normal distributions,
$\boldsymbol{p_{1}},\boldsymbol{p_{2}}$ where $\boldsymbol{p_{i}}\sim\mathcal{N}(\boldsymbol{\mu}_{i},\,\boldsymbol{\Sigma}_{i})$.
\begin{equation}
D_{BC-MN}\left(p_{1},p_{2}\right)=\frac{1}{8}(\boldsymbol{\mu}_{1}-\boldsymbol{\mu}_{2})^{T}\boldsymbol{\Sigma}^{-1}(\boldsymbol{\mu}_{1}-\boldsymbol{\mu}_{2})+\frac{1}{2}\ln\,\left(\frac{\det\boldsymbol{\Sigma}}{\sqrt{\det\boldsymbol{\Sigma}_{1}\,\det\boldsymbol{\Sigma}_{2}}}\right)
\end{equation}

$\boldsymbol{\mu}_{i}$ and $\boldsymbol{\Sigma}_{i}$ are the means
and covariances of the distributions, and $\boldsymbol{\Sigma}=\frac{\boldsymbol{\Sigma}_{1}+\boldsymbol{\Sigma}_{2}}{2}$.
We need to keep in mind that a discrete sample could be stored in
matrices of the form $A$ and $B$, where, $n$ is the number of observations
and $m$ denotes the number of variables captured by the two matrices.
\begin{equation}
\boldsymbol{A_{m\times n}}\sim\mathcal{N}\left(\boldsymbol{\mu_{1}},\boldsymbol{\varSigma_{1}}\right)
\end{equation}
\begin{equation}
\boldsymbol{B_{m\times n}}\sim\mathcal{N}\left(\boldsymbol{\mu_{2}},\boldsymbol{\varSigma_{2}}\right)
\end{equation}

The Bhattacharyya measure finds numerous applications in communications,
pattern recognition and information theory (Mak and Barnard 1996;
Guorong, Peiqi \& Minhui 1996). Other measures of divergence include
the Mahalanobis distance (Mahalanobis 1936) which is a measure of
the distance between a point and a distribution, Kullback-Leibler
(KL) divergence (Kullback \& Leibler 1951) and the Hellinger or Matusita
measure, $D_{H-M}\left(p_{i},p_{i}^{\prime}\right)$, (Hellinger 1909;
Matusita 1955) which is related to the Bhattacharyya distance since
minimizing the Matusita distance is equivalent to maximizing the Bhattacharyya
distance. The KL measure is not symmetric and is discussed in (Duchi
2007; Contreras-Reyes \& Arellano-Valle 2012) when applied to normal
or skew normal distributions. In (Aherne, Thacker \& Rockett 1998)
it is shown that the Bhattacharyya coefficient approximates the $\chi^{2}-$measure
(Chi-Squared measure), while avoiding the singularity problem that
occurs when comparing instances of the distributions that are both
zero. 
\begin{equation}
D_{H-M}\left(p_{i},p_{i}^{\prime}\right)=\sum_{i}^{k}\left(\sqrt{p_{i}}-\sqrt{p_{i}^{\prime}}\right)^{2}=2-2\rho\left(p_{i},p_{i}^{\prime}\right)
\end{equation}
\begin{equation}
\chi^{2}\left(p_{i},p_{i}^{\prime}\right)=\frac{1}{2}\sum_{i}^{k}\frac{\left(p_{i}-p_{i}^{\prime}\right)^{2}}{\left(p_{i}+p_{i}^{\prime}\right)}
\end{equation}
For discrete probability distributions $P$ and $Q$, the Kullback–Leibler
divergence, $D_{\mathrm{KL}}(P\|Q)$ , of $Q$ from $P$ is defined
to be,
\begin{equation}
D_{\mathrm{KL}}(P\|Q)=\sum_{i}P(i)\,\log\frac{P(i)}{Q(i)}
\end{equation}
For distributions $P$ and $Q$, of continuous random variables, the
Kullback–Leibler divergence is defined (Bishop 2006) to be the integral
below, where $p$ and $q$ denote the densities of $P$ and $Q$.
(Huzurbazar 1955) proves a remarkable property that for all distributions
admitting sufficient statistics the exact forms of the KL divergence
come out as explicit functions of the parameters of the distribution.
\begin{equation}
D_{\mathrm{KL}}(P\|Q)=\int_{-\infty}^{\infty}p(x)\,\log\frac{p(x)}{q(x)}\,{\rm d}x\!
\end{equation}
(Schweppe 1967b) develops new expressions for the KL divergence and
the Bhattacharyya distance, in terms of the effects of various conditional
expectation filters (physically realizable linear systems) acting
on Gaussian processes. In particular, the distances are given by time
integrals of the variances and mean values of the outputs of filters
designed to generate the conditional expectations of certain processes.
Defining a matched filter to be one whose output is the conditional
expectation of the signal contained in the input, then the performance
of a matched filter is much easier to analyze than the performance
of a mismatched filter. The divergence involves mismatched filters
while the Bhattacharyya distance uses only matched filters. Hence
the Bhattacharyya distance is easier to analyze. In (Kailath 1967)
it is shown that the two measures give similar results for Gaussian
processes with unequal mean value functions and that the Bhattacharyya
distance yields better results when the processes have unequal covariance
functions, where in fact the divergence measure fails. In (Schweppe
1967a) the Bhattacharyya distance is specialized to Markov-Gaussian
processes. (Cha 2007) is a comprehensive survey on distance/similarity
measures between probability density functions. 

For a discrete sample in a two class scenario, (Jain 1976) shows that
any estimate of the Bhattacharyya coefficient is biased and consistent
using a Taylor series expansion around the neighborhood of the points.
We can write the Bhattacharyya coefficient as, 
\begin{equation}
\rho\left(\alpha,\beta\right)=\sum_{i=1}^{N}\sqrt{\alpha_{i}\beta_{i}}
\end{equation}
Here, $\alpha$ and $\beta$ denote the sets of parameters $\left\{ \alpha_{1},...,\alpha_{N}\right\} $
and $\left\{ \beta_{1},...,\beta_{N}\right\} $ respectively. The
class conditional densities are given by,
\begin{equation}
p_{i}=\sum_{i=1}^{N}\alpha_{i}\delta\left(x-i\right)
\end{equation}
\begin{equation}
p_{i}^{\prime}=\sum_{i=1}^{N}\beta_{i}\delta\left(x-i\right)
\end{equation}
\[
\text{such that }\qquad\sum_{i=1}^{N}\alpha_{i}=\sum_{i=1}^{N}\beta_{i}=1
\]
\begin{equation}
\delta\left(i-j\right)=\delta_{ij}=\begin{cases}
0 & \text{if }i\neq j,\\
1 & \text{if }i=j.
\end{cases}
\end{equation}
Let $\hat{\alpha_{i}}$ and $\hat{\beta_{i}}$ denote the maximum
likelihood estimates of $\alpha$and $\beta$, respectively, $i=1,...,N$,
based on $m$ samples available from each of the two classes: 
\begin{equation}
\hat{\alpha_{i}}=\frac{m_{i1}}{m}\qquad\hat{\beta_{i}}=\frac{m_{i2}}{m}
\end{equation}
where $m_{i1}$ and $m_{i2}$ are the numbers of samples for which
$x$ takes the value $i$ from class $c_{1}$ and $c_{2}$, respectively.
We define a sample estimate of the Bhattacharyya coefficient as,
\begin{equation}
\hat{\rho}=\rho\left(\hat{\alpha},\hat{\beta}\right)=\sum_{i=1}^{N}\sqrt{\hat{\alpha}_{i}\hat{\beta}_{i}}
\end{equation}
where, $\hat{\alpha}=\left\{ \hat{\alpha}_{1},...,\hat{\alpha}_{N}\right\} $
and $\hat{\beta}=\left\{ \hat{\beta}_{1},...,\hat{\beta}_{N}\right\} $
respectively. (Djouadi, Snorrason \& Garber 1990) derive closed-form
expressions for the bias and variance of Bhattacharyya coefficient
estimates based on a certain number of training samples from two classes,
described by multivariate Gaussian distributions. Numerical examples
are used to show the relationship between the true parameters of the
densities, the number of training samples, the class variances, and
the dimensionality of the observation space.
\end{doublespace}
\begin{doublespace}

\subsection{Dimension Reduction}
\end{doublespace}

\begin{doublespace}
A key requirement to apply the Bhattacharyya distance in practice
is to have data-sets with the same number of dimensions. (Fodor 2002;
Burges 2009; Sorzano, Vargas \& Montano 2014) are comprehensive collections
of methodologies aimed at reducing the dimensions of a data-set using
Principal Component Analysis or Singular Value Decomposition and related
techniques. 

(Johnson \& Lindenstrauss 1984) proved a fundamental result (JL Lemma)
that says that any $n$ point subset of Euclidean space can be embedded
in $k=O(log\frac{n}{\epsilon^{2}})$ dimensions without distorting
the distances between any pair of points by more than a factor of
$\left(1\pm\epsilon\right)$, for any $0<\epsilon<1$. Whereas principal
component analysis is only useful when the original data points are
inherently low dimensional, the JL Lemma requires absolutely no assumption
on the original data. Also, note that the final data points have no
dependence on $d$, the dimensions of the original data which could
live in an arbitrarily high dimension. 

Simplified versions of the original proof are in (Frankl and Maehara
1988; 1990). We use the version of the bounds for the dimensions of
the transformed subspace given in (Frankl \& Maehara 1990; Dasgupta
\& Gupta 1999). (Nelson 2010) gives a proof using the Hanson and Wright
inequality (Hanson \& Wright 1971). (Achlioptas 2003) gives a proof
of the JL Lemma using randomized algorithms. (Venkatasubramanian \&
Wang 2011) present a study of the empirical behavior of algorithms
for dimensionality reduction based on the JL Lemma. We point out the
wonderful references for the theory of matrix algebra and numerical
applications (Gentle 2007; 2012).
\end{doublespace}
\begin{lem}
\begin{doublespace}
\label{Prop:Johnson and Lindenstrauss --- Dasgupta and Gupta}For
any $0<\epsilon<1$ and any integer $n$, let $k<d$ be a positive
integer such that 
\[
k\geq4\left(\frac{\epsilon^{2}}{2}-\frac{\epsilon^{3}}{3}\right)^{-1}\ln n
\]
Then for any set $V$ of $n$ points in $\boldsymbol{R}^{d}$, there
is a map $f:\boldsymbol{R}^{d}\rightarrow\boldsymbol{R}^{k}$ such
that for all $u,v\in V$, 
\[
\left(1-\epsilon\right)\Vert u-v\Vert^{2}\leq\Vert f\left(u\right)-f\left(v\right)\Vert^{2}\leq\left(1+\epsilon\right)\Vert u-v\Vert^{2}
\]
Furthermore, this map can be found in randomized polynomial time and
one such map is $f\left(x\right)=\frac{1}{\sqrt{k}}Ax$ where, $x\in\boldsymbol{R}^{d}$
and $A$ is a $k\times d$ matrix in which each entry is sampled i.i.d
from a Gaussian $N\left(0,1\right)$ distribution.
\end{doublespace}
\end{lem}
\begin{doublespace}

\section{\label{sec:Intuition-for-Dimension}Intuition for Dimension Reduction}
\end{doublespace}

\begin{doublespace}
The above discussions of distance measures and dimension reduction
are extensively used in many areas, but the combination of the two
is bound to create confusions in the minds of the uninitiated. Hence,
we provide examples from daily life, (both ours and from creatures
in higher dimensions), to provide a convincing argument as to why
the two used together can be a powerful tool for the study of complex
systems governed by uncertainty.
\end{doublespace}
\begin{doublespace}

\subsection{Game of Darts }
\end{doublespace}

\begin{doublespace}
If we consider a cloud of points in multi-dimensional space. It would
be reasonable to expect that the distance between the points, or how
the points are distributed, gives a measure of the randomness inherent
in the process generating them. When dimension reduction moves the
points to lower dimensions, and the change in the distance between
them stays bounded, the randomness properties of the original process
are retained, to the extent as dictated, by the bounds established
by the procedure performing the dimension transformation, which in
our case is given by the JL Lemma. 

Taking an example, from our own real lives, the marks on a dartboard
in a game of darts are representative of the skills of the people
throwing the darts. For simplicity, we could assume that there are
three types of dart throwers: novice, intermediate and advanced. Identifying
the category of the person making the marks, would be similar to identifying
the type of distribution of a stochastic process. If we map the marks,
on the board, to a line using a transformation, that keeps the distances
between the marks bounded, a flavor of the skill level would be retained
and we would be able to identify the category of the person making
the marks. Dimension reduction, using a transformation that keeps
the distances bounded, is in essence the same undertaking.
\end{doublespace}
\begin{doublespace}

\subsection{The Merits and Limits of Four Physical Dimensions}
\end{doublespace}

\begin{doublespace}
Another example of dimension transformation is from the physical world
we live in. We are four dimensional creatures: latitude, longitude,
height and time are our dimensions since we need to know these four
co-ordinates to fully specify the position of any object in our universe.
This is perhaps, made clear to lay audiences, (with regards to physics,
such as many of us), by the movie Interstellar (Thorne 2014). Also,
(Sagan 2006) has a mesmerizing account of many physical aspects including
how objects, or, beings can transform from higher to lower dimensions
and change shapes; but they would need to obey the laws of the lower
dimension. The last dimension, time, is the one we cannot control,
or, move around in. But we can change the other three co-ordinates
and hence we have three degrees of freedom.

(Appendix \ref{sec:Appendix:-Example-Four-Physical-Dimensions}) has
a detailed example to build the intuition related to dimension transformation
from our four dimensional physical world. This example should also
make it clear to us that it is better to work with the highest dimension
that we can afford to work with, since each higher dimension retains
some flavors of the object we are trying to understand that might
be not discernible in lower dimensions. This should also tell us that
objects we observe in our universe might have many interesting properties
that we are unable to observe due to the restrictions of our physical
dimensions.
\end{doublespace}
\begin{doublespace}

\section{Methodological Innovations}
\end{doublespace}

\begin{doublespace}
In this section, we collect the new results developed in this paper.
These can be broadly categorized into two buckets. The first is the
types of distributions that we would obtain when a particular distribution
type is transformed to a different dimension (different number of
random variables) using the JL-Lemma (sections \ref{subsec:Normal-Log-Normal-Mixture},
\ref{subsec:Normal-Normal-Product}, \ref{subsec:Truncated-Normal-},
\ref{subsec:Discrete-Multivariate-Distributi}). The second group
of results considers the relationships between covariance and distance
measures (section \ref{subsec:Covariance-and-Distance}).
\end{doublespace}
\begin{doublespace}

\subsection{\label{subsec:Normal-Log-Normal-Mixture}Normal Log-Normal Mixture}
\end{doublespace}

\begin{doublespace}
The normality or otherwise of stock price changes is discussed extensively
in the literature: (Osborne\textbf{ }1959; 1962; Fama 1965; 1995;
Mandelbrot \& Taylor 1967; Kon 1984; Richardson \& Smith 1993). Starting
with a geometric Brownian motion for the stock price process, it can
be established that stock prices are distributed log-normally, (Hull
2006). If we are looking at variables that are only positive, such
as prices, quantities traded, or volatilities, then it is a reasonable
initial assumption that they are distributed log normally (we relax
this assumption to incorporate more generic settings in later sections). 

Transforming log-normal multi-variate variables into a lower dimension
by multiplication with an independent normal distribution (see: Lemma
\ref{Prop:Johnson and Lindenstrauss --- Dasgupta and Gupta}) results
in the sum of variables with a normal log-normal mixture, (Clark 1973;
Tauchen \& Pitts 1983; Yang 2008), evaluation of which requires numerical
techniques (Miranda \& Fackler 2002). A random variable, $U$, would
be termed a normal log-normal mixture if it is of the form,
\begin{equation}
U=Xe^{Y}
\end{equation}
where, $X$ and $Y$ are random variables with correlation coefficient,
$\rho$ satisfying the below, 
\begin{equation}
\left[\begin{array}{c}
X\\
Y
\end{array}\right]\sim N\left(\left[\begin{array}{c}
\mu_{X}\\
\mu_{Y}
\end{array}\right],\left[\begin{array}{cc}
\sigma_{X}^{2} & \rho\sigma_{X}\sigma_{Y}\\
\rho\sigma_{X}\sigma_{Y} & \sigma_{Y}^{2}
\end{array}\right]\right)
\end{equation}
We note that for $\sigma_{Y}=0$ when $Y$ degenerates to a constant,
this is just the distribution of $X$ and $\rho$ is unidentified. 

To transform a column vector with $d$ observations of a random variable
into a lower dimension of order, $k<d$, we can multiply the column
vector with a matrix, $A\sim N(0;\frac{1}{k})$ of dimension $k\times d$.
\end{doublespace}
\begin{prop}
\begin{doublespace}
\label{prop:Normal_Lognormal_The-density-function} A dimension transformation
of $d$ observations of a log-normal variable into a lower dimension,
$k$, using Lemma \ref{Prop:Johnson and Lindenstrauss --- Dasgupta and Gupta},
yields a probability density function which is the sum of random variables
with a normal log-normal mixture, given by the convolution,
\[
f_{S}\left(s\right)=f_{U_{1}}\left(u_{1}\right)*f_{U_{2}}\left(u_{2}\right)*...*f_{U_{k}}\left(u_{k}\right)
\]
\[
\text{Here, }f_{U_{i}}\left(u_{i}\right)=\frac{\sqrt{k}}{2\pi\sigma_{Y_{i}}}\int_{-\infty}^{\infty}\;e^{-y-\frac{ku_{i}^{2}}{2e^{2y}}-\frac{\left[y-\mu_{Y_{i}}\right]^{2}}{2\sigma_{Y_{i}}^{2}}}dy
\]
\[
U_{i}=X_{i}e^{Y_{i}}
\]
\[
\left[\begin{array}{c}
X_{i}\\
Y_{i}
\end{array}\right]\sim N\left(\left[\begin{array}{c}
0\\
\mu_{Y_{i}}
\end{array}\right],\left[\begin{array}{cc}
\frac{1}{k} & 0\\
0 & \sigma_{Y_{i}}^{2}
\end{array}\right]\right)
\]
The convolution of two probability densities arises when we have the
sum of two independent random variables, $Z=X+Y$. The density of
$Z,\;h_{Z}\left(z\right)$ is given by,
\[
{\displaystyle h_{Z}\left(z\right)=\left(f_{X}\text{\textasteriskcentered}f_{Y}\right)\left(z\right)=f_{X}\left(x\right)*f_{Y}\left(y\right)=\int_{-\infty}^{\infty}f_{X}\left(z-y\right)*f_{Y}\left(y\right)dy=\int_{-\infty}^{\infty}f_{X}\left(x\right)*f_{Y}\left(z-x\right)dx}
\]
When the number of independent random variables being added is more
than two, or the reduced dimension after the Lemma \ref{Prop:Johnson and Lindenstrauss --- Dasgupta and Gupta}
transformation is more than two, $k>2$, then we can take the convolution
of the density resulting after the convolution of the first two random
variables, with the density of the third variable and so on in a pair
wise manner, till we have the final density.
\end{doublespace}
\end{prop}
\begin{proof}
\begin{doublespace}
Appendix \ref{subsec:Proof-of-Proposition: Normal_Lognormal} sketches
a general proof and then tailors it to our case where the normal distribution
has zero mean and the two variables are uncorrelated.
\end{doublespace}
\end{proof}
\begin{doublespace}
Methods of estimating parameters and comparing the resulting distributions
when normal (log-normal) distributions are mixed, are studied in (Fowlkes
1979; Vernic, Teodorescu \& Pelican 2009). As noted earlier, the normal
log-normal mixture tends to the normal distribution when the log-normal
distribution has low variance and this property can be helpful for
deciding when this approach is suitable.
\end{doublespace}
\begin{doublespace}

\subsection{\label{subsec:Normal-Normal-Product}Normal Normal Product}
\end{doublespace}

\begin{doublespace}
For completeness, we illustrate how dimension reduction would work
on a data-set containing random variables that have normal distributions.
This can also serve as a useful benchmark given the wide usage of
the normal distribution and can be an independently useful result.

(Craig 1936) was one of the earlier attempts to investigate the probable
error of the product of the two quantities, each of known probable
error; becoming the first work to determine the algebraic expression
for the moment-generating function of the product, but without being
able to determine the distribution of the product. It was found that
the distribution of $Z=XY$ is a function of the coefficient of correlation
of both variables and of two parameters that are proportional to the
inverse of the coefficient of variation of each variable. When the
product of the means of the two random variables is nonzero, the distribution
is skewed as well as having excess kurtosis, although (Aroian 1947;
Aroian, Taneja \& Cornwell 1978) showed that the product approaches
the normal distribution as one or both of the ratios of the means
to standard errors (the inverse of the coefficients of variation)
of each random variable get large in absolute value.

(Aroian 1947) showed that the gamma distribution (standardized Pearson
type III) can provide an approximation in some situations. Instead,
the analytical solution for this product distribution is a Bessel
function of the second kind with a purely imaginary argument (Aroian
1947; Craig 1936). The four moments of the product of two correlated
normal variables are given in (Craig 1936; Aroian, Taneja \& Cornwell
1978).
\end{doublespace}
\begin{prop}
\begin{doublespace}
\label{prop:Normal_Normal_The-density-function} A dimension transformation
of $d$ observations of a normal variable into a lower dimension,
$k$, using Lemma \ref{Prop:Johnson and Lindenstrauss --- Dasgupta and Gupta},
yields a probability density function which is the sum of random variables
with a normal normal product distribution, given by the convolution,
\[
f_{S}\left(s\right)=f_{U_{1}}\left(u_{1}\right)*f_{U_{2}}\left(u_{2}\right)*...*f_{U_{k}}\left(u_{k}\right)
\]
\[
\text{Here, }f_{U_{i}}\left(u_{i}\right)=\int_{-\infty}^{\infty}\left(\frac{1}{\left|x\right|}\right)\frac{1}{\sigma_{Y_{i}}\sqrt{2\pi}}\;e^{-\frac{\left(x-\mu_{Y_{i}}\right)^{2}}{2\sigma_{Y_{i}}^{2}}}\sqrt{\frac{k}{2\pi}}\;e^{-\frac{k\left(\frac{u_{i}}{x}\right)^{2}}{2}}dx
\]
\[
U_{i}=X_{i}Y_{i}
\]
\[
\left[\begin{array}{c}
X_{i}\\
Y_{i}
\end{array}\right]\sim N\left(\left[\begin{array}{c}
0\\
\mu_{Y_{i}}
\end{array}\right],\left[\begin{array}{cc}
\frac{1}{k} & 0\\
0 & \sigma_{Y_{i}}^{2}
\end{array}\right]\right)
\]
\end{doublespace}
\end{prop}
\begin{proof}
\begin{doublespace}
Appendix \ref{subsec:Proof-of-Proposition: Normal_Normal} sketches
a general proof and then tailors it to our case where the normal distribution
has zero mean and the two variables are uncorrelated.
\end{doublespace}
\end{proof}
\begin{rem*}
\begin{doublespace}
We note the following two useful results.
\end{doublespace}
\end{rem*}
\begin{enumerate}
\begin{doublespace}
\item By writing the product as the difference of two squared variables,
it is easy to see that the product distribution is a linear combination
of two Chi-Square random variables,
\begin{equation}
U_{i}=X_{i}Y_{i}=\frac{1}{4}\left\{ \left[X_{i}+Y_{i}\right]^{2}-\left[X_{i}-Y_{i}\right]^{2}\right\} 
\end{equation}
If $X_{i}\sim N\left(0,\sigma_{X_{i}}^{2}\right);Y_{i}\sim N\left(0,\sigma_{Y_{i}}^{2}\right)$
with a coefficient of correlation, $\rho_{X_{i}Y_{i}}$ then, 
\begin{equation}
U_{i}=X_{i}Y_{i}\sim\frac{1}{4}\left\{ \left(\sigma_{X_{i}}^{2}+\sigma_{Y_{i}}^{2}+2\sigma_{X_{i}}\sigma_{Y_{i}}\rho_{X_{i}Y_{i}}\right)P-\left(\sigma_{X_{i}}^{2}+\sigma_{Y_{i}}^{2}-2\sigma_{X_{i}}\sigma_{Y_{i}}\rho_{X_{i}Y_{i}}\right)Q\right\} 
\end{equation}
Here, $P,Q\sim\chi_{1}^{2}$ or central chi-squared random variables
with one degree of freedom. $P,Q$ are independent only if $\sigma_{X_{i}}^{2}=\sigma_{Y_{i}}^{2}$.
Hence, in general, $P,Q$ are dependent non central chi-squared variables.
\item The result we use in (appendix \ref{subsec:Proof-of-Proposition: Normal_Normal})
to derive the above convolution, can also be arrived at, by writing
the density for $W=XY$ using the Dirac Delta function, $\delta\left(x\right)$
as,
\begin{equation}
f_{W}\left(w\right)=f_{W}\left(x,y\right)=\int_{-\infty}^{\infty}f_{X}\left(\left.x\right|y\right)f_{Y}\left(y\right)dy=\int_{-\infty}^{\infty}f_{X}\left(x\right)dx\int_{-\infty}^{\infty}f_{Y}\left(y\right)\delta\left(w-xy\right)dy
\end{equation}
\begin{equation}
=\int_{-\infty}^{\infty}f_{X}\left(x\right)dx\int_{-\infty}^{\infty}f_{Y}\left(y\right)\frac{\delta\left(y-\frac{w}{x}\right)}{\left|x\right|}dy
\end{equation}
\begin{equation}
=\int_{-\infty}^{\infty}f_{X}\left(x\right)\frac{1}{\left|x\right|}f_{Y}\left(\frac{w}{x}\right)dx
\end{equation}
\end{doublespace}
\end{enumerate}
\begin{doublespace}
(Glen, Leemis \& Drew 2004) present an algorithm for computing the
probability density function of the product of two independent random
variables. (Springer \& Thompson 1966) use the Mellin integral transform
(Epstein 1948; End-note \ref{enu:The-Mellin-transform}) to develop
fundamental methods for the derivation of the probability distributions
and density functions of the products of $n$ independent random variables;
(Springer \& Thompson 1970) use these methods to show that the products
of independent beta, gamma and central Gaussian random variables are
Meijer G-functions (Mathai \& Saxena 1973; End-note \ref{enu:The-G-function-was}).

(Ware \& Lad 2003) has results very closely aligned to our requirements.
They attempt to calculate the probability that the sum of the product
of variables with a Normal distribution is negative. They first assess
the distribution of the product of two independent normally distributed
variables by comparing three different methods: 1) a numerical method
approximation, which involves implementing a numerical integration
procedure on MATLAB; 2) a Monte Carlo construction and; 3) an approximation
to the analytic result using the Normal distribution under certain
conditions, by calculating the first two moments of the product, and
then finding a distribution whose parameters match the moments. Second,
they consider the sum of the products of two Normally distributed
variables by applying the Convolution Formula. Lastly, they combine
the two steps to arrive at the main results while also showing that
a direct Monte Carlo approximation approach could be used. (Seijas-Macías
\& Oliveira 2012) is a recent work that has several comparisons using
Newton-Cotes numerical integration (Weisstein 2004; End-note \ref{enu:The-Newton-Cotes-formulas,}). 

In addition, while it useful to keep these results at the back of
our mind, it is worth finding simpler distributions instead of having
to numerically calculate the distance based on the normal product
or the normal log-normal mixture sum. An alternate is discussed in
the next section, where we set both distributions to be truncated
normal.
\end{doublespace}
\begin{doublespace}

\subsection{\label{subsec:Truncated-Normal-}Truncated Normal / Multivariate
Normal}
\end{doublespace}

\begin{doublespace}
A truncated normal distribution is the probability distribution of
a normally distributed random variable whose value is either bounded
below, above or both (Horrace 2005; Burkardt 2014) and hence seems
like a natural candidate to fit the distributions we are dealing with.
\textcolor{black}{We can estimate the parameters of the distributions
(both the unchanged one and the one with the reduced dimensions) by
setting them as truncated multivariate normals and calculate the distance
based on these estimated distributions. }

To assess the suitability of normal distributions to fit the observed
sample, there are a variety of tests. The univariate sample measures
of skewness $S$ and kurtosis $K$ may be used for testing univariate
normality. Under normality, the theoretical values of $S$ and $K$
are $0$ and $3$, respectively. One of the most famous tests for
normality of regression residuals is the test of (Jarque \& Bera 1980;
1987; End-notes \ref{enu:Skewness}; \ref{enu:Kurtosis}). The test
statistic JB (Jarque-Bera) is a function of the measures of skewness
and kurtosis computed from the sample and is based on the Lagrange
multiplier test or score test. 

The simplest testing problem assumes that the data $y$ are generated
by a joint density function $f\left(y,\theta^{0}\right)$ under the
null hypothesis and by $f\left(y,\theta\right)$ under the alternative,
with $\theta^{0},\theta\in\boldsymbol{R}^{k}$. The Lagrange Multiplier
test is derived from a constrained maximization principle (Engle 1984;
End-note \ref{enu:Rao's-score-LM-test,}). Maximizing the log-likelihood
subject to the constraint that $\theta^{0}=\theta$ yields a set of
Lagrange Multipliers which measure the shadow price of the constraint.
If the price is high, the constraint should be rejected as inconsistent
with the data. 

Another test frequently used is the sum of squares of the standardized
sample skewness and kurtosis which is asymptotically distributed as
a $\chi^{2}$ variate (Doornik and Hansen 2008) along with transformations
of skewness and kurtosis to facilitate easier implementation and also
to contend with small sample issues. (Malkovich \& Afifi 1973) generalize
these statistics to test a hypothesis of multivariate normality including
a discussion of transformations to simplify the computational methods.

(Székely \& Rizzo 2005; End-note \ref{enu:The-Euclidean-distance})
propose a new class of consistent tests for comparing multivariate
distributions based on Euclidean distance between sample elements.
Applications include one-sample goodness-of-fit tests for discrete
or continuous multivariate distributions in arbitrary dimension $d\geq1$.
The new tests can be applied to assess distributional assumptions
for many classical procedures in multivariate analysis. \textbf{(}Wald
\& Wolfowitz 1946) consider the problem of setting tolerance limits
for normal distributions with unknown mean and variance. For a univariate
distribution with a sample of $N$ independent observations, two functions
$L_{1}$ and $L_{2}$ of the sample need to be constructed such that
the probability that the limits $L_{1}$ and $L_{2}$ will include
at-least a given proportion $\gamma$ of the population is equal to
a preassigned value $\beta$.

Suppose $X\sim N(\mu,\sigma^{2})$ has a normal distribution and lies
within the interval $X\in(a,b),\;-\infty\leq a<b\leq\infty$. Then
$X$ conditional on $a<X<b$ has a truncated normal distribution.
Its probability density function, $f_{X}$, for $a\leq x\leq b$ ,
is given by
\begin{equation}
f_{X}\left(x\mid\mu,\sigma^{2},a,b\right)=\begin{cases}
\frac{\frac{1}{\sigma}\phi\left(\frac{x-\mu}{\sigma}\right)}{\Phi\left(\frac{b-\mu}{\sigma}\right)-\Phi\left(\frac{a-\mu}{\sigma}\right)} & \quad;a\leq x\leq b\\
0 & ;\text{otherwise}
\end{cases}
\end{equation}

Here, ${\scriptstyle {\phi(\xi)=\frac{1}{\sqrt{2\pi}}\exp{(-\frac{1}{2}\xi^{2}})}}$
is the probability density function of the standard normal distribution
and ${\scriptstyle {\Phi(\cdot)}}$ is its cumulative distribution
function. There is an understanding that if ${\scriptstyle {b=\infty}\ }$,
then ${\scriptstyle {\Phi(\frac{b-\mu}{\sigma})=1}}$, and similarly,
if ${\scriptstyle {a=-\infty}}$ , then ${\scriptstyle {\Phi(\frac{a-\mu}{\sigma})=0}}$. 
\end{doublespace}
\begin{prop}
\begin{doublespace}
\label{prop:Truncated_N_The-Bhattacharyya-coefficient}The Bhattacharyya
distance when we have truncated normal distributions, $p,q$, that
do not overlap is zero and when they overlap it is given by 
\begin{eqnarray*}
D_{BC-TN}(p,q) & = & \frac{1}{4}\left(\frac{(\mu_{p}-\mu_{q})^{2}}{\sigma_{p}^{2}+\sigma_{q}^{2}}\right)+\frac{1}{4}\ln\left(\frac{1}{4}\left(\frac{\sigma_{p}^{2}}{\sigma_{q}^{2}}+\frac{\sigma_{q}^{2}}{\sigma_{p}^{2}}+2\right)\right)\\
 &  & +\frac{1}{2}\ln\left[\Phi\left(\frac{b-\mu_{p}}{\sigma_{p}}\right)-\Phi\left(\frac{a-\mu_{p}}{\sigma_{p}}\right)\right]+\frac{1}{2}\ln\left[\Phi\left(\frac{d-\mu_{q}}{\sigma_{q}}\right)-\Phi\left(\frac{c-\mu_{q}}{\sigma_{q}}\right)\right]\\
 &  & -\ln\left\{ \Phi\left[\frac{u-\nu}{\varsigma}\right]-\Phi\left[\frac{l-\nu}{\varsigma}\right]\right\} 
\end{eqnarray*}
Here,
\[
p\sim N\left(\mu_{p},\sigma_{p}^{2},a,b\right)\;;\;q\sim N\left(\mu_{q},\sigma_{q}^{2},c,d\right)
\]
\[
l=\min\left(a,c\right)\;;\;u=\min\left(b,d\right)
\]
\[
\nu=\frac{\left(\mu_{p}\sigma_{q}^{2}+\mu_{q}\sigma_{p}^{2}\right)}{\left(\sigma_{p}^{2}+\sigma_{q}^{2}\right)}\;;\;\varsigma=\sqrt{\frac{2\sigma_{p}^{2}\sigma_{q}^{2}}{\left(\sigma_{p}^{2}+\sigma_{q}^{2}\right)}}
\]
\end{doublespace}
\end{prop}
\begin{proof}
\begin{doublespace}
Appendix \ref{subsec:Proof-of-Proposition: Truncated Normal}.
\end{doublespace}
\end{proof}
\begin{doublespace}
It is easily seen that 
\begin{equation}
\underset{\left(a,c\right)\rightarrow-\infty;\left(b,d\right)\rightarrow\infty}{\lim}D_{BC-TN}(p,q)=D_{BC-N}(p,q)
\end{equation}
Looking at conditions when $D_{BC-TN}(p,q)\geq D_{BC-N}(p,q)$ gives
the below. This shows that the distance measure increases with greater
truncation and then decreases as the extent of overlap between the
distributions decreases.
\begin{equation}
\sqrt{\left[\Phi\left(\frac{b-\mu_{p}}{\sigma_{p}}\right)-\Phi\left(\frac{a-\mu_{p}}{\sigma_{p}}\right)\right]\left[\Phi\left(\frac{d-\mu_{q}}{\sigma_{q}}\right)-\Phi\left(\frac{c-\mu_{q}}{\sigma_{q}}\right)\right]}\geq\left\{ \Phi\left[\frac{u-\nu}{\varsigma}\right]-\Phi\left[\frac{l-\nu}{\varsigma}\right]\right\} 
\end{equation}
\textbf{(}Kiani, Panaretos, Psarakis \& Saleem 2008\textbf{; }Zogheib
\& Hlynka 2009; Soranzo \& Epure 2014) list some of the numerous techniques
to calculate the normal cumulative distribution. Approximations to
the error function are also feasible options (Cody 1969; Chiani, Dardari
\& Simon 2003).

Similarly, a truncated multivariate normal distribution $\boldsymbol{X}$
has the density function, 
\begin{equation}
f_{\mathbf{X}}\left(x_{1},\ldots,x_{k}\mid\boldsymbol{\mu_{p}},\,\boldsymbol{\Sigma_{p}},\,\boldsymbol{a},\,\boldsymbol{b}\right)=\frac{\exp\left(-\frac{1}{2}\left({\mathbf{x}}-{\boldsymbol{\mu_{p}}}\right)^{\mathrm{T}}{\boldsymbol{\Sigma_{p}}}^{-1}\left({\mathbf{x}}-{\boldsymbol{\mu_{p}}}\right)\right)}{\int_{\boldsymbol{a}}^{\boldsymbol{b}}\exp\left(-\frac{1}{2}\left({\mathbf{x}}-{\boldsymbol{\mu_{p}}}\right)^{\mathrm{T}}{\boldsymbol{\Sigma_{p}}}^{-1}\left({\mathbf{x}}-{\boldsymbol{\mu_{p}}}\right)\right)d\boldsymbol{x};\;\boldsymbol{x}\in\boldsymbol{R}_{\boldsymbol{a}\leq\boldsymbol{x}\leq\boldsymbol{b}}^{k}}
\end{equation}
Here, $\boldsymbol{\mu_{p}}$ is the mean vector and $\boldsymbol{\Sigma_{p}}$
is the symmetric positive definite covariance matrix of the $\boldsymbol{p}$
distribution and the integral is a $k$ dimensional integral with
lower and upper bounds given by the vectors $\left(\boldsymbol{a},\boldsymbol{b}\right)$
and $\boldsymbol{\boldsymbol{x}\in\boldsymbol{R}_{\boldsymbol{a}\leq\boldsymbol{x}\leq\boldsymbol{b}}^{k}}$.
\end{doublespace}
\begin{prop}
\begin{doublespace}
\label{prop:Truncated_MN_The-Bhattacharyya-distance}The Bhattacharyya
distance when we have truncated multivariate normal distributions
$\boldsymbol{p},\boldsymbol{q}$ and all the $k$ dimensions have
some overlap is given by 
\begin{eqnarray*}
D_{BC-TMN}\left(\boldsymbol{p},\boldsymbol{q}\right) & = & \frac{1}{8}(\boldsymbol{\mu_{p}}-\boldsymbol{\mu_{q}})^{T}\boldsymbol{\Sigma}^{-1}(\boldsymbol{\mu_{p}}-\boldsymbol{\mu_{q}})+\frac{1}{2}\ln\,\left(\frac{\det\boldsymbol{\Sigma}}{\sqrt{\det\boldsymbol{\Sigma_{p}}\,\det\boldsymbol{\Sigma_{q}}}}\right)\\
 &  & +\frac{1}{2}\ln\left[\frac{1}{\sqrt{(2\pi)^{k}\left(|\boldsymbol{\Sigma_{p}}|\right)}}\int_{\boldsymbol{a}}^{\boldsymbol{b}}\exp\left(-\frac{1}{2}\left({\mathbf{x}}-{\boldsymbol{\mu_{p}}}\right)^{\mathrm{T}}{\boldsymbol{\Sigma_{p}}}^{-1}\left({\mathbf{x}}-{\boldsymbol{\mu_{p}}}\right)\right)d\boldsymbol{x};\;\boldsymbol{x}\in\boldsymbol{R}_{\boldsymbol{a}\leq\boldsymbol{x}\leq\boldsymbol{b}}^{k}\right]\\
 &  & +\frac{1}{2}\ln\left[\frac{1}{\sqrt{(2\pi)^{k}\left(|\boldsymbol{\Sigma_{q}}|\right)}}\int_{\boldsymbol{c}}^{\boldsymbol{d}}\exp\left(-\frac{1}{2}\left({\mathbf{x}}-{\boldsymbol{\mu_{p}}}\right)^{\mathrm{T}}{\boldsymbol{\Sigma_{\boldsymbol{q}}}}^{-1}\left({\mathbf{x}}-{\boldsymbol{\mu_{\boldsymbol{q}}}}\right)\right)d\boldsymbol{x};\;\boldsymbol{x}\in\boldsymbol{R}_{\boldsymbol{\boldsymbol{c}}\leq\boldsymbol{x}\leq\boldsymbol{d}}^{k}\right]\\
 &  & -\ln\left[\frac{1}{\sqrt{(2\pi)^{k}\det\left({\boldsymbol{\Sigma_{p}}}{\boldsymbol{\Sigma}}^{-1}{\boldsymbol{\Sigma_{q}}}\right)}}\right.\\
 &  & \left.\int_{\boldsymbol{l}}^{\boldsymbol{u}}\exp\left(-\frac{1}{2}\left\{ \left({\mathbf{x}-\mathbf{m}}\right)^{\mathrm{T}}\left({\boldsymbol{\Sigma_{q}}}^{-1}\left[{\boldsymbol{\Sigma}}\right]{\boldsymbol{\Sigma_{p}}}^{-1}\right)\left({\mathbf{x}-\mathbf{m}}\right)\right\} \right)d\boldsymbol{x};\;\boldsymbol{x}\in\boldsymbol{R}_{\boldsymbol{\min\left(a,\boldsymbol{c}\right)}\leq\boldsymbol{x}\leq\boldsymbol{\min\left(b,d\right)}}^{k}\vphantom{\frac{\sqrt{(2\pi)^{k}\det\left({\boldsymbol{\Sigma_{p}}}^{-1}\right)}}{\sqrt{(2\pi)^{k}\det\left({\boldsymbol{\Sigma_{p}}}^{-1}\right)}}}\right]
\end{eqnarray*}
Here,
\[
\boldsymbol{p}\sim N\left(\boldsymbol{\mu_{p}},\,\boldsymbol{\Sigma_{p}},\,\boldsymbol{a},\,\boldsymbol{b}\right)
\]
\[
\boldsymbol{q}\sim N\left(\boldsymbol{\mu_{\boldsymbol{q}}},\,\boldsymbol{\Sigma_{\boldsymbol{q}}},\,\boldsymbol{c},\,\boldsymbol{d}\right)
\]
\[
\boldsymbol{u=\min\left(b,d\right)}\;;\;\boldsymbol{l=\min\left(a,c\right)}
\]
\[
{\mathbf{m}}=\left[\left({\boldsymbol{\mu_{p}}}^{\mathrm{T}}{\boldsymbol{\Sigma_{p}}}^{-1}+{\boldsymbol{\mu_{q}}}^{\mathrm{T}}{\boldsymbol{\Sigma_{q}}}^{-1}\right)\left({\boldsymbol{\Sigma_{p}}}^{-1}+{\boldsymbol{\Sigma_{q}}}^{-1}\right)^{-1}\right]^{\mathrm{T}}
\]
\[
\boldsymbol{\Sigma}=\frac{\boldsymbol{\Sigma_{p}}+\boldsymbol{\Sigma_{q}}}{2}
\]
\end{doublespace}
\end{prop}
\begin{proof}
\begin{doublespace}
Appendix \ref{subsec:Proof-of-Proposition: Truncated Multivariate Normal}
\end{doublespace}
\end{proof}
\begin{doublespace}
Again we see that,

\begin{equation}
\underset{\left(\boldsymbol{a},\boldsymbol{c}\right)\rightarrow-\infty;\left(\boldsymbol{b},\boldsymbol{d}\right)\rightarrow\infty}{\lim}D_{BC-TMN}\left(\boldsymbol{p},\boldsymbol{q}\right)=D_{BC-MN}\left(\boldsymbol{p},\boldsymbol{q}\right)
\end{equation}
Looking at conditions when $D_{BC-TMN}(\boldsymbol{p},\boldsymbol{q})\geq D_{BC-MN}(\boldsymbol{p},\boldsymbol{q})$
gives a similar condition as the univariate case,
\begin{align}
 & \frac{1}{2}\ln\left[\frac{1}{\sqrt{(2\pi)^{k}\left(|\boldsymbol{\Sigma_{p}}|\right)}}\int_{\boldsymbol{a}}^{\boldsymbol{b}}\exp\left(-\frac{1}{2}\left({\mathbf{x}}-{\boldsymbol{\mu_{p}}}\right)^{\mathrm{T}}{\boldsymbol{\Sigma_{p}}}^{-1}\left({\mathbf{x}}-{\boldsymbol{\mu_{p}}}\right)\right)d\boldsymbol{x};\;\boldsymbol{x}\in\boldsymbol{R}_{\boldsymbol{a}\leq\boldsymbol{x}\leq\boldsymbol{b}}^{k}\right]\\
 & +\frac{1}{2}\ln\left[\frac{1}{\sqrt{(2\pi)^{k}\left(|\boldsymbol{\Sigma_{q}}|\right)}}\int_{\boldsymbol{c}}^{\boldsymbol{d}}\exp\left(-\frac{1}{2}\left({\mathbf{x}}-{\boldsymbol{\mu_{p}}}\right)^{\mathrm{T}}{\boldsymbol{\Sigma_{\boldsymbol{q}}}}^{-1}\left({\mathbf{x}}-{\boldsymbol{\mu_{\boldsymbol{q}}}}\right)\right)d\boldsymbol{x};\;\boldsymbol{x}\in\boldsymbol{R}_{\boldsymbol{\boldsymbol{c}}\leq\boldsymbol{x}\leq\boldsymbol{d}}^{k}\right]\\
\geq & -\ln\left[\frac{1}{\sqrt{(2\pi)^{k}\det\left({\boldsymbol{\Sigma_{p}}}{\boldsymbol{\Sigma}}^{-1}{\boldsymbol{\Sigma_{q}}}\right)}}\right.\\
 & \left.\int_{\boldsymbol{l}}^{\boldsymbol{u}}\exp\left(-\frac{1}{2}\left\{ \left({\mathbf{x}-\mathbf{m}}\right)^{\mathrm{T}}\left({\boldsymbol{\Sigma_{q}}}^{-1}\left[{\boldsymbol{\Sigma}}\right]{\boldsymbol{\Sigma_{p}}}^{-1}\right)\left({\mathbf{x}-\mathbf{m}}\right)\right\} \right)d\boldsymbol{x};\;\boldsymbol{x}\in\boldsymbol{R}_{\boldsymbol{\min\left(a,\boldsymbol{c}\right)}\leq\boldsymbol{x}\leq\boldsymbol{\min\left(b,d\right)}}^{k}\vphantom{\frac{\sqrt{(2\pi)^{k}\det\left({\boldsymbol{\Sigma_{p}}}^{-1}\right)}}{\sqrt{(2\pi)^{k}\det\left({\boldsymbol{\Sigma_{p}}}^{-1}\right)}}}\right]
\end{align}

\end{doublespace}
\begin{doublespace}

\subsection{\label{subsec:Discrete-Multivariate-Distributi}Discrete Multivariate
Distribution}
\end{doublespace}

\begin{doublespace}
A practical approach would be to use discrete approximations for the
probability distributions. This is typically done by matching the
moments of the original and approximate distributions. Discrete approximations
of probability distributions typically are determined by dividing
the range of possible values or the range of cumulative probabilities
into a set of collectively exhaustive and mutually exclusive intervals.
Each interval is approximated by a value equal to its mean or median
and a probability equal to the chance that the true value will be
in the interval. 

A criterion for judging the accuracy of the discrete approximation
is that it preserve as many of the moments of the original distribution
as possible. (Keefer \& Bodily 1983; Smith 1993) provide comparisons
of commonly used discretization methods. (Miller \& Rice 1983) look
at numerical integration using gaussian quadrature as a means of improving
the approximation. This approach approximates the integral of the
product of a function $g\left(x\right)$ and a weighting function
$w\left(x\right)$ over the internal $\left(a,b\right)$ by evaluating
$g\left(x\right)$ at several values $x$ and computing a weighted
sum of the results.
\begin{equation}
\int_{a}^{b}g\left(x\right)w\left(x\right)dx=\sum_{i=1}^{N}w_{i}g\left(x_{i}\right)
\end{equation}

For the case of a discrete approximation of a probability distribution,
the density function $f\left(x\right)$ with support $\left(a,b\right)$
is associated with the weighting function and the probabilities $p_{i}$
with the weights $w_{i}$. We approximate $g(x)$ by a polynomial,
and choose $x_{i}$ and $p_{i}$ (or $w_{i}$) to provide an adequate
approximation for each term of the polynomial. This translates to
finding a set of values and probabilities such that,
\begin{equation}
\int_{a}^{b}x^{k}f\left(x\right)dx=\sum_{i=1}^{N}p_{i}x_{i}^{k}\qquad\text{for }k=0,1,2,\ldots
\end{equation}
A discrete approximation with $N$ probability-value pairs can match
the first $2N-1$ moments exactly by finding $p_{i}$ and $x_{i}$
that satisfy the below equations, which are solved using well known
techniques.
\begin{align}
p_{1}+p_{2}+\ldots+p_{N} & =\int_{a}^{b}x^{0}f\left(x\right)dx=1\\
p_{1}x_{1}+p_{2}x_{2}+\ldots+p_{N}x_{N} & =\int_{a}^{b}xf\left(x\right)dx\\
p_{1}x_{1}^{2}+p_{2}x_{2}^{2}+\ldots+p_{N}x_{N}^{2} & =\int_{a}^{b}x^{2}f\left(x\right)dx\\
\vdots & \vdots\\
p_{1}x_{1}^{2N-1}+p_{2}x_{2}^{2N-1}+\ldots+p_{N}x_{N}^{2N-1} & =\int_{a}^{b}x^{2N-1}f\left(x\right)dx
\end{align}

Among several alternatives to distribution approximations, (Wallace
1958),\textcolor{black}{{} restricts attention to finding asymptotic
expansions in which the errors of approximation approach zero as some
parameter, typically sample size, tends to infinity. Estimating the
parameters of the distributions (both the unchanged one and the one
with the reduced dimensions) by setting them as discrete multivariate
distributions and calculating the distance based on these estimated
distributions is a practical alternative to using the truncated normal
distributions.}
\end{doublespace}
\begin{doublespace}

\subsection{\label{subsec:Covariance-and-Distance}Covariance and Distance}
\end{doublespace}

\begin{doublespace}
We compare the covariance between two random variables and the Bhattacharyya
coefficient, assuming support over relevant portions of the real line
and that the density function is differentiable. Consider the covariance
and the distance of $Y=X$ and $Y=-X$. It is easily seen that the
distance between the distributions is $zero$ in both cases, but the
covariance is $one$ when they are exactly the same and $minus\;one$
in the other case. This shows that despite knowing the distance we
do not have full information about the co-movements of the two distributions.
Hence distance is not a substitute for covariance but rather a very
useful complement, since knowing the two will tell us the magnitudes
of how one variable will change if we know the other and the likelihood
of observing certain values from one distribution if have observed
values from the other one.

An interesting result concerning the covariance is given by Stein's
lemma (Stein 1973; 1981; Rubinstein 1973; 1976) which connects the
covariance between two random variables that are joint normally distributed
and the covariance between any function of one of the random variables
and the other. If $X$ and $Y$ have a joint distribution that is
bivariate normal and if $c\left(\cdotp\right)$ is a continuously
differentiable function of $X$ then, 
\begin{equation}
\text{Cov}\left[c\left(X\right),Y\right]=E\left[c'\left(X\right)\right]\text{Cov}\left[X,Y\right]
\end{equation}
It easily follows that, if $X$ is a normal random variable with mean
$\mu$ and variance $\sigma^{2}$ and if $c$ is a differentiable
function with derivative $c'$ such that $E\left[c'\left(X\right)\right]<\infty$,
we then have, 
\begin{equation}
E\left[\left(X-\mu\right)c\left(X\right)\right]=\sigma^{2}E[c'\left(X\right)]
\end{equation}

(Losq \& Chateau 1982) extend the result to the case when $c\left(\cdotp\right)$
is a function of $n$ random variables. (Wei \& Lee 1988) extend this
further to the case where both variables are functions of multivariate
normal random variables. (Siegel 1993) derives a remarkable formula
for the covariance of $X_{1}$ with the minimum of the multivariate
normal vector $\left(X_{1},\ldots,X_{n}\right)$. 
\begin{equation}
\text{Cov}[X_{1},\min\left(X_{1},\ldots,X_{n}\right)]=\sum_{i=1}^{n}\text{Cov}\left(X_{1},X_{i}\right)\text{Pr}\left[X_{i}=\min\left(X_{1},\ldots,X_{n}\right)\right]
\end{equation}
(Liu 1994) uses a multivariate version of Stein's identity to devise
the below more generalized formula, where $X_{\left(i\right)}$ is
the $i-$th largest among $X_{1},\ldots,X_{n}$,
\begin{equation}
\text{Cov}[X_{1},X_{\left(i\right)}]=\sum_{j=1}^{n}\text{Cov}\left(X_{1},X_{j}\right)\text{Pr}\left[X_{j}=X_{(i)}\right]
\end{equation}
(Kattumannil 2009) extends the Stein lemma by relaxing the requirement
of normality. If the continuous random variable $X$ has support over
the interval $\left(a,b\right)$, that is $-\infty\leq a<X<b\leq\infty$,
with mean $E(X)=\mu$, variance $Var(X)=\sigma^{2}$, density function
$f_{X}\left(t\right)$ and cumulative distribution $F_{X}\left(t\right)$.
Let $h$ be a real valued function such that $E\left[h\left(X\right)\right]=\mu$
and $E\left[h^{2}\left(X\right)\right]<\infty$. Suppose $f_{X}\left(t\right)$
is a differentiable function with derivative $f_{X}'\left(t\right)$
and there exists a non-vanishing (non-zero over the support) function
$g\left(t\right)$ such that,
\begin{equation}
\frac{f_{X}'\left(t\right)}{f_{X}\left(t\right)}=-\frac{g'\left(t\right)}{g\left(t\right)}+\frac{\left[\mu-h\left(t\right)\right]}{g\left(t\right)}\quad,\quad t\in(a,b)
\end{equation}
\begin{equation}
\Rightarrow\quad f_{X}'\left(t\right)g\left(t\right)+g'\left(t\right)f_{X}\left(t\right)=\left[\mu-h\left(t\right)\right]f_{X}\left(t\right)
\end{equation}
\begin{equation}
\Rightarrow\quad\frac{\partial f_{X}\left(t\right)g\left(t\right)}{\partial t}=\left[\mu-h\left(t\right)\right]f_{X}\left(t\right)
\end{equation}
Integrating with respect to $t$ from $r$ to $b$ and assuming $\underset{t\rightarrow b}{\lim}\;g\left(t\right)f_{X}\left(t\right)=0$
shows that for a given $h\left(t\right)$ the value of $g\left(t\right)$
uniquely determines the distribution of $X$.
\begin{equation}
\left|f_{X}\left(t\right)g\left(t\right)\right|_{r}^{b}=\int_{r}^{b}\left[\mu-h\left(t\right)\right]f_{X}\left(t\right)dt
\end{equation}
\begin{equation}
\Rightarrow\quad f_{X}\left(r\right)=\frac{1}{g\left(r\right)}\left\{ \int_{r}^{b}\left[h\left(t\right)-\mu\right]f_{X}\left(t\right)dt\right\} 
\end{equation}
Similarly, integrating with respect to $t$ from $a$ to $r$ and
assuming $\underset{r\rightarrow a}{\lim}\;g\left(t\right)f_{X}\left(t\right)=0$,
\begin{equation}
f_{X}\left(r\right)=\frac{1}{g\left(r\right)}\left\{ \int_{r}^{a}\left[\mu-h\left(t\right)\right]f_{X}\left(t\right)dt\right\} 
\end{equation}

For any absolutely continuous function, $c\left(t\right)$, with derivative
$c'\left(t\right)$ satisfying $E\left[c\left(X\right)h\left(X\right)\right]<\infty,E\left[c^{2}\left(X\right)\right]<\infty,E\left[g\left(X\right)c'\left(X\right)\right]<\infty$,
and provided $\underset{t\rightarrow b}{\lim}\;g\left(t\right)f_{X}\left(t\right)=0$
we have the following identity,
\begin{equation}
E\left[c\left(X\right)\left\{ h\left(X\right)\text{\textminus}\mu\right\} \right]=E\left[c'\left(X\right)g\left(X\right)\right]
\end{equation}
 (Teerapabolarn 2013) is a further extension of this normality relaxation
to discrete distributions. Another useful reference, (Kimeldorf \&
Sampson 1973), provides a class of inequalities between the covariance
of two random variables and the variance of a function of the two
random variables. Let $\boldsymbol{A}$ be an open convex subset of
the plane $\boldsymbol{R}^{2}$ and $\boldsymbol{\Upsilon}$ be the
class of all pairs $\left(X,Y\right)$ of real random variables having
finite variances and $P\left\{ \left(X,Y\right)\in\boldsymbol{A}\right\} =1$.
Assume that $g$ is a function with continuous partial first derivatives
in the domain $\boldsymbol{A}$ characterized by the following functional
inequality and an equivalent partial differential inequality,
\begin{equation}
\forall\left(X,Y\right)\in\boldsymbol{\Upsilon};\quad\left(x_{1},y_{1}\right),\left(x_{2},y_{2}\right)\in\boldsymbol{A}\Longleftrightarrow\left(x_{2}-x_{1}\right)\left(y_{2}-y_{1}\right)\leq\left[g\left(x_{2},y_{2}\right)-g\left(x_{1},y_{1}\right)\right]^{2}
\end{equation}
\begin{equation}
\left(\frac{\partial g}{\partial y}\right)\left(\frac{\partial g}{\partial y}\right)\geq\frac{1}{4}
\end{equation}
If the above equivalent necessary and sufficient conditions are satisfied
we have,
\begin{equation}
\text{cov}\left(X,Y\right)\leq\text{var }g\left(X,Y\right)
\end{equation}
We now derive the following general extension to Stein's lemma, that
does not require normality, involving the covariance between a random
variable and a function of another random variable. 
\end{doublespace}
\begin{prop}
\begin{doublespace}
\label{prop:Stein_Lemma_Generic}If the continuous random variables
$X$ and $Y$ have support over the interval $\left(a,b\right)$,
that is $-\infty\leq a<X,Y<b\leq\infty$, with means $E(X)=\mu_{X};E(Y)=\mu_{Y}$,
variances $Var(X)=\sigma_{X}^{2};Var(Y)=\sigma_{Y}^{2}$, density
functions $f_{X}\left(t\right);f_{Y}\left(u\right)$, cumulative distributions
$F_{X}\left(t\right);F_{Y}\left(u\right)$ and joint density function,
$f_{XY}\left(t,u\right)$. Let $h\left(u\right)$ be a real valued
function such that $E\left[h\left(Y\right)\right]=\mu_{Y}$ and $E\left[h^{2}\left(Y\right)\right]<\infty$.
Suppose $f_{XY}\left(t,u\right)$ is a differentiable function with
derivative $f_{XY}'\left(t,u\right)$ with respect to $t$ and there
exists a non-vanishing (non-zero over the support) function $g\left(t,u\right)$
such that,

\[
\frac{f_{XY}'\left(t,u\right)}{f_{XY}\left(t,u\right)}=-\frac{g'\left(t,u\right)}{g\left(t,u\right)}+\frac{\left[\mu_{Y}-h\left(u\right)\right]}{g\left(t,u\right)}\quad,\quad t,u\in(a,b)
\]
 Assuming \textup{$\underset{t\rightarrow b}{\lim}\;g\left(t,u\right)f_{XY}\left(t,u\right)=0$}
shows that for a given $h\left(u\right)$ the value of $g\left(t,u\right)$
uniquely determines the joint distribution of $X$ and $Y$.
\[
f_{XY}\left(r,u\right)\:g\left(r,u\right)=\int_{r}^{b}\left[h\left(u\right)-\mu_{Y}\right]\:f_{XY}\left(t,u\right)\:dt
\]
\[
\int_{a}^{b}f_{XY}\left(r,u\right)\:g\left(r,u\right)\:du=\int_{a}^{b}\int_{r}^{b}\left[h\left(u\right)-\mu_{Y}\right]\:f_{XY}\left(t,u\right)\:dt\:du
\]
Similarly, assuming \textup{$\underset{t\rightarrow a}{\lim}\;g\left(t,u\right)f_{XY}\left(t,u\right)=0$
gives,}
\[
f_{XY}\left(r,u\right)\:g\left(r,u\right)=\int_{a}^{r}\left[\mu_{Y}-h\left(u\right)\right]\:f_{XY}\left(t,u\right)\:dt
\]
\[
\int_{a}^{b}f_{XY}\left(r,u\right)\:g\left(r,u\right)\:du=\int_{a}^{b}\int_{a}^{r}\left[\mu_{Y}-h\left(u\right)\right]\:f_{XY}\left(t,u\right)\:dt\:du
\]
For any absolutely continuous function, $c\left(t\right)$, with derivative
$c'\left(t\right)$ satisfying $E\left[c\left(X\right)h\left(Y\right)\right]<\infty,E\left[c^{2}\left(X\right)\right]<\infty,E\left[g\left(X,Y\right)c'\left(X\right)\right]<\infty$,
and provided \textup{$\underset{t\rightarrow b}{\lim}\;g\left(t,u\right)f_{XY}\left(t,u\right)=0$,
}we have the following identity,
\[
\text{Cov}\left[c\left(X\right),h\left(Y\right)\right]=E\left[c'\left(X\right)\:g\left(X,Y\right)\right]
\]
\end{doublespace}
\end{prop}
\begin{proof}
\begin{doublespace}
Appendix \ref{subsec:Proof-of-Proposition: Stein Lemma Generic}.
\end{doublespace}
\end{proof}
\begin{cor}
\begin{doublespace}
\label{cor:Stein_Lemma_Generic}It is easily seen that , 
\[
\text{Cov}\left[c\left(X\right),Y\right]=E\left[c'\left(X\right)\:g\left(X,Y\right)\right]
\]
\end{doublespace}
\end{cor}
\begin{proof}
\begin{doublespace}
Substituting $h\left(u\right)=u$ in the proof of the proposition
\ref{prop:Stein_Lemma_Generic} and simplifying gives this result.
\end{doublespace}
\end{proof}
\begin{doublespace}
Building upon some of the above results, primarily the extension to
Stein's lemma in corollary \ref{cor:Stein_Lemma_Generic}, we formulate
the following connection between distance (Bhattacharyya coefficient)
and covariance.
\end{doublespace}
\begin{prop}
\begin{doublespace}
\label{prop:Distance-Covariance-Relationship}The following equations
govern the relationship between the Bhattacharyya coefficient, $\rho\left(f_{X},f_{Y}\right)$,
and the covariance between any two distributions with joint density
function, $f_{XY}\left(t,u\right)$, means, $\mu_{X}\text{ and }\mu_{Y}$
and density functions $f_{X}\left(t\right)\text{ and }f_{Y}\left(t\right)$,
\[
\text{Cov}\left[c\left(X\right),Y\right]=\text{Cov}\left(X,Y\right)-E\left[\sqrt{\frac{f_{Y}\left(t\right)}{f_{X}\left(t\right)}}Y\right]+\mu_{Y}\:\rho\left(f_{X},f_{Y}\right)
\]
\[
\text{Cov}\left(X,Y\right)+\mu_{Y}\:\rho\left(f_{X},f_{Y}\right)=E\left[c'\left(X\right)\:g\left(X,Y\right)\right]+E\left[\sqrt{\frac{f_{Y}\left(t\right)}{f_{X}\left(t\right)}}Y\right]
\]
Here, 
\[
c\left(t\right)=t-\sqrt{\frac{f_{Y}\left(t\right)}{f_{X}\left(t\right)}}
\]
and $g\left(t,u\right)$ is a non-vanishing function such that,

\[
\frac{f_{XY}'\left(t,u\right)}{f_{XY}\left(t,u\right)}=-\frac{g'\left(t,u\right)}{g\left(t,u\right)}+\frac{\left[\mu_{Y}-u\right]}{g\left(t,u\right)}\quad,\quad t,u\in(a,b)
\]
\end{doublespace}
\end{prop}
\begin{proof}
\begin{doublespace}
Appendix \ref{subsec:Proof-of-Proposition: Distance_Covariance_Relationship}.
\end{doublespace}
\end{proof}
\begin{cor}
\begin{doublespace}
\label{cor:Distance_Covariance_Relationship}It is easily seen that
, 
\[
\text{Cov}\left[c\left(X\right),h\left(Y\right)\right]=\text{Cov}\left(X,h\left(Y\right)\right)-E\left[\sqrt{\frac{f_{Y}\left(t\right)}{f_{X}\left(t\right)}}h\left(Y\right)\right]+\mu_{Y}\rho\left(f_{X},f_{Y}\right)
\]
\[
\text{Cov}\left(X,h\left(Y\right)\right)+\mu_{Y}\rho\left(f_{X},f_{Y}\right)=E\left[c'\left(X\right),g\left(X,Y\right)\right]+E\left[\sqrt{\frac{f_{Y}\left(t\right)}{f_{X}\left(t\right)}}h\left(Y\right)\right]
\]
\end{doublespace}
\end{cor}
\begin{proof}
\begin{doublespace}
Using $h\left(u\right)$ instead of $u$ in the proof of the proposition
\ref{prop:Distance-Covariance-Relationship} and simplifying gives
this result.
\end{doublespace}
\end{proof}
\begin{doublespace}

\section{\label{sec:Marketstructure,-Microstructure-}Market-structure, Microstructure
and Other Applications}
\end{doublespace}

\begin{doublespace}
Our methodologies can help in the comparison of economic systems that
generate prices, quantities and aid in the analysis of shopping patterns
and understanding consumer behavior. The systems could be goods transacted
at different shopping malls or trade flows across entire countries.
Study of traffic congestion, road accidents and other fatalities across
two regions could be performed to get an idea of similarities and
seek common solutions where such simplifications might be applicable.
The diversity of applications is limited only by our imagination since
these techniques can be used in any system where data is generated
and is studied to obtain a better understanding of the system.

We point below one asset pricing and one biological application to
show the limitless possibilities such a comparison affords. The empirical
illustration below in section \ref{subsec:Comparison-of-Security}
is another example of how the techniques developed here can be immediately
applicable for microstructure studies. We provide a definition of
microstructure that justifies why information collection and dimension
reduction can be helpful procedures in this space.
\end{doublespace}
\begin{defn}
\begin{doublespace}
\textbf{\textit{Market Microstructure is the investigation of the
process and protocols that govern the exchange of assets with the
objective of reducing frictions that can impede the transfer.}} 

In financial markets, where there is an abundance of recorded information,
this translates to the study of the dynamic relationships between
observed variables, such as price, volume and spread, and hidden constituents,
such as transaction costs and volatility, that hold sway over the
efficient functioning of the system. The degree to which different
markets or sub groups of securities have different measures of their
corresponding distributions tells us the extent to which they are
different. This can aid investors looking for diversification or looking
for more of the same thing. 
\end{doublespace}
\end{defn}
\begin{doublespace}

\subsection{\label{subsec:Asset-Pricing-Application}Asset Pricing Application}
\end{doublespace}

\begin{doublespace}
The price $p_{t}$ of an asset today that gives a payoff $x_{t+1}$
in the future (next time period) is given by the fundamental asset
pricing equation $p_{t}=E\left(m_{t+1}x_{t+1}\right)$ that makes
use of a stochastic discount factor $m_{t+1}$. Multi-period versions
are easily derived and hence we drop the time subscripts below. (Lucas
1978; Hansen \& Richard 1987; Hansen \& Jagannathan 1991; Cochrane
2009) are classic references. (Kashyap 2018) reviews the asset pricing
literature with emphasis on the equity risk premium puzzle (Mehra
\& Prescott 1985; Mehra 1988), which is that the return on equities
has far exceeded the average return on short-term risk-free debt and
cannot be explained by conventional representative-agent consumption
based equilibrium models. Some solution attempts and related debates
are (Rietz 1988; Mehra \& Prescott 1988; Weil 1989; Constantinides
\& Duffie 1996; Campbell \& Cochrane 1999; Bansal \& Yaron 2004; Barro
2006; Weitzman 2007). The puzzle is based on the assumption that consumption
growth is log-normal. We can relax this assumption and derive the
below relationships that can aid the many empirical attempts on asset
pricing that seek to find proxies for the discount factor or to marginal
utility growth, employing variables $f$ of the form, $m=c\left(f\right)$.
\end{doublespace}
\begin{prop}
\begin{doublespace}
\label{prop:The-asset-pricing}The asset pricing equation $p=E\left(mx\right)$
can be written as,
\end{doublespace}
\end{prop}
\begin{doublespace}
\[
p=E\left(mx\right)\equiv E\left[c\left(f\right)x\right]\quad\text{Setting},\;m=c\left(f\right)
\]
\[
p=E\left[c\left(f\right)\right]E\left(x\right)+\text{Cov}\left[c\left(f\right),x\right]
\]
\[
p=E\left[c\left(f\right)\right]E\left(x\right)+E\left[c'\left(f\right)\:g\left(f,x\right)\right]
\]

\end{doublespace}
\begin{proof}
\begin{doublespace}
Using proposition \ref{prop:Stein_Lemma_Generic} and the definitions
of $c\left(t\right)$ and $g\left(t,u\right)$ therein, gives this
result.
\end{doublespace}
\end{proof}
\begin{cor}
\begin{doublespace}
\label{cor:Asset Pricing}With a further restriction on $m\equiv c\left(f\right)=f-\sqrt{\frac{f_{x}\left(f\right)}{f_{f}\left(f\right)}}$,
where, expected payoff is $\mu_{x}$ with density function $f_{x}\left(f\right)$
and the density function of the factor $f$ is $f_{f}\left(f\right)$,
the asset pricing equation becomes, 
\[
p=E\left[c\left(f\right)\right]E\left(x\right)+\text{Cov}\left(f,x\right)+\mu_{x}\:\rho\left(f_{f},f_{x}\right)-E\left[\sqrt{\frac{f_{x}\left(f\right)}{f_{f}\left(f\right)}}x\right]
\]
\end{doublespace}
\end{cor}
\begin{proof}
\begin{doublespace}
Follows immediately from proposition \ref{prop:Distance-Covariance-Relationship}.
\end{doublespace}
\end{proof}
\begin{doublespace}

\subsection{\label{subsec:Biological-Application}Biological Application}
\end{doublespace}

\begin{doublespace}
The field of medical sciences is not immune from the increasing collection
of information (Berner \& Moss 2005) and the side effects of growing
research and literature (Huth 1989). The challenges of managing growing
research output are significantly more complex than expanding amounts
of data. This might require advances in artificial intelligence and
a better understanding of the giant leap the human mind makes at times
from sorting through information, summarizing it as knowledge and
transcending to a state of condensed wisdom. We leave this problem
aside for now; but tackle the slightly more surmountable problem of
voluminous data sources.

Suppose we have a series of observations of different properties,
(such as heart rate or blood sugar level and others as well), across
different days from different people from different geographical regions.
The goal would then be to identify people from which regions have
a more similar constitution. The number of days on which the observations
are made is the same for all regions; but there could be different
number of participants in each region. 

The data gathering could be the result of tests conducted at different
testing facilities on a certain number of days and different number
of people might show up at each facility on the different days the
tests are done. The result of the data collection will be as follows:
We have a matrix, $M_{i,j}$, for each region $i$ and property $j$,
with $T$ observations across time (rows) and $N_{i,j}$ (columns)
for number of participants from region $i$ and property $j$. Here,
we consider all properties together as representing a region. The
dimension of the data matrix is given by $\text{Dim}(M_{i,j})=(T,N_{i,j})$.
We could have two cases depending on the data collection results.
\end{doublespace}
\begin{enumerate}
\begin{doublespace}
\item $\ensuremath{N_{i,j}=N_{i,k}}\ensuremath{,\forall j,k}$ That is the
number of participants in a region are the same for all properties
being measured. This would be the simple case.
\item $\ensuremath{N_{i,j}\neq N_{i,k}}\ensuremath{,\forall j,k}$ That
is the number of participants in a region could be different for some
of the properties being measured.
\end{doublespace}
\end{enumerate}
\begin{doublespace}
The simple scenario corresponding to case 1) of the data collection
would be when we have a matrix, $M_{i,j}$, for each region $i$ and
property $j$, with $T$ observations across time (rows) and $N_{i}$
(columns) for number of participants from that region. If we consider
each property separately, we can compute the Bhattacharyya Distance
across the various matrices $M_{i,j}$ (in pairs) separately for each
property $j$ across the various regions or $i$'s. The multinomial
approach outlined earlier can be used to calculate one coefficient
for each region by combining all the properties. For the second case,
we can compute the Bhattacharyya Distance across the various matrices
$M_{i,j}$ for all the properties combined (or separately for each
property) for a region using the multinomial approach after suitable
dimension transformations. 

To illustrate this, suppose we have two matrices, $A$ and $B$, representing
two different probability distributions, with dimensions, $m\times n$
and $k\times n$, respectively. Here, $m$ and $k$ denote the number
of variables captured by the two matrices $A$ and $B$. In general,
$m$ and $k$ are not equal. $n$ is the number of observations, which
is the same across the two distributions. We can calculate the Bhattacharyya
distance or another measure of similarity or dissimilarity between
$A$ and $B$ by reducing the dimensions of the larger matrix, (say
$m>k$) to the smaller one so that we have $\underline{A}$ and $B$,
with the same dimension, $k\times n$. Here, $\underline{A}=CA,\;\text{Dim}\left(C\right)=k\times m$.
Each entry in $C$ is sampled i.i.d from a Gaussian $N\left(0,\frac{1}{k}\right)$
distribution.
\end{doublespace}
\begin{doublespace}

\section{\label{sec:Empirical-Illustrations}Empirical Illustrations}
\end{doublespace}
\begin{doublespace}

\subsection{\label{subsec:Comparison-of-Security}Comparison of Security Prices
across Markets}
\end{doublespace}

\begin{doublespace}
We start with a simple example illustrating how this measure could
be used to compare different countries based on the prices of all
equity securities traded in that market. Our data sample contains
closing prices for most of the securities from six different markets
from Jan 01, 2014 to May 28, 2014 (Figure \ref{fig:Markets-and-Tickers}).
Singapore with 566 securities is the market with the least number
of traded securities. Even if we reduce the dimension of all the other
markets with more number of securities, for a proper comparison of
these markets we would need more than two years worth of data. Hence
as a simplification, we first reduce the dimension of the matrix holding
the close prices for each market using PCA reduction, so that the
number of tickers retained would be comparable to the number of days
for which we have data. The results of such a comparison are shown
in Figure \ref{fig:PCA-Results}. Some implementation pointers using
statistical packages such as R are given in section \ref{subsec:Implementation-Pointers}.

We report the full matrix and not just the upper or lower matrix since
the PCA reduction we do takes the first country reduces the dimensions
up-to a certain number of significant digits and then reduces the
dimension of the second country to match the number of dimensions
of the first country. For example, this would mean that comparing
AUS and SGP is not exactly the same as comparing SGP and AUS. As a
safety step before calculating the distance, which requires the same
dimensions for the structures holding data for the two entities being
compared, we could perform dimension reduction using JL Lemma if the
dimensions of the two countries differs after the PCA reduction. 
\end{doublespace}

We repeat the calculations for different number of significant digits
of the PCA reduction. \textbf{\textit{This shows the fine granularity
of the results that our distance comparison produces and highlights
the issue that with PCA reduction there is loss of information, since
with different number of significant digits employed in the PCA reduction,
we get the result that different markets are similar.}}\textit{ }For
example, in Figure \ref{fig:PCA-Results}, AUS - SGP are the most
similar markets when two significant digits are used and AUS - HKG
are the most similar with six significant digits.

\begin{doublespace}
We illustrate another example, where we compare a randomly selected
sub universe of securities in each market, so that the number of tickers
retained would be comparable to the number of days for which we have
data. The results are shown in Figure \ref{fig:Results-with-Randomly}.
The left table (Figure \ref{fig:PCA-Dimension-Reduction}) is for
PCA reduction on a randomly chosen sub universe and the right table
(Figure \ref{fig:JL-Lemma-Dimenion}) is for dimension reduction using
JL Lemma for the same sub universe. We report the full matrix for
the same reason as explained earlier and perform multiple iterations
when reducing the dimension using the JL Lemma. A key observation
is that the magnitude of the distances are very different when using
PCA reduction and when using dimension reduction, due to the loss
of information that comes with the PCA technique. It is apparent that
using dimension reduction via the JL Lemma produces consistent results,
since the same pairs of markets are seen to be similar in different
iterations (in Figure \ref{fig:JL-Lemma-Dimenion}, AUS - IND are
the most similar in iteration one and also in iteration five). 
\end{doublespace}

It is worth remembering that in each iteration of the JL Lemma dimension
transformation we multiply by a different random matrix and hence
the distance is slightly different in each iteration but within the
bound established by JL Lemma. When the distance is somewhat close
between two pairs of entities, we could observe an inconsistency due
to the JL Lemma transformation in successive iterations. We discuss
this as one area that would require further theoretical investigation
in section \ref{sec:Possibilities-for-Future}.

\begin{doublespace}
Our approach could also be used when groups of securities are being
compared within the same market, a very common scenario when deciding
on the group of securities to invest in a market as opposed to deciding
which markets to invest in. Such an approach would be highly useful
for index construction or comparison across sectors within a market.
(Kashyap 2017) summarizes the theoretical results from this paper
and expands the example illustrations to open, high, low prices, volumes
and price volatilities to provide a complete market microstructure
study; relevant bits of this study to illustrate how our techniques
can be applied to different kinds of variables are reproduced in section
\ref{sec:Comparison-Volumes-Prices-Volatilities}. 

\begin{figure}[H]
\subfloat[Markets and Tickers Count\label{fig:Markets-and-Tickers}]{\includegraphics{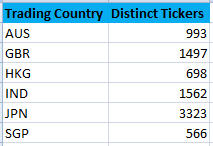}

}\hfill{}\subfloat[PCA Dimension Reduction\label{fig:PCA-Results}]{\includegraphics[width=7cm,height=10cm]{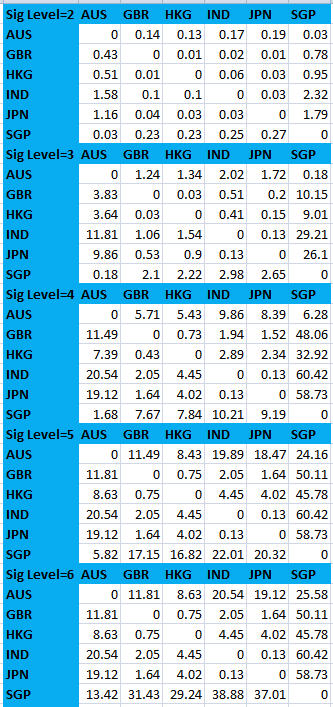}

}

\caption{Distance Measures over Full Sample}

\end{figure}
\begin{figure}[H]
\subfloat[PCA Dimension Reduction\label{fig:PCA-Dimension-Reduction}]{\includegraphics[width=7cm,height=10cm]{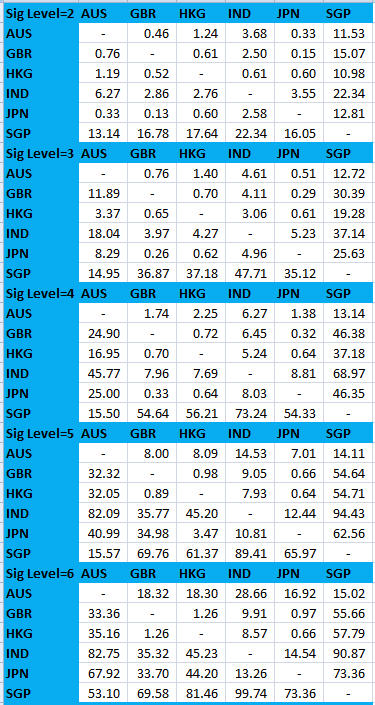}

}\hfill{}\subfloat[JL Lemma Dimension Reduction\label{fig:JL-Lemma-Dimenion}]{\includegraphics[width=8cm,height=10cm]{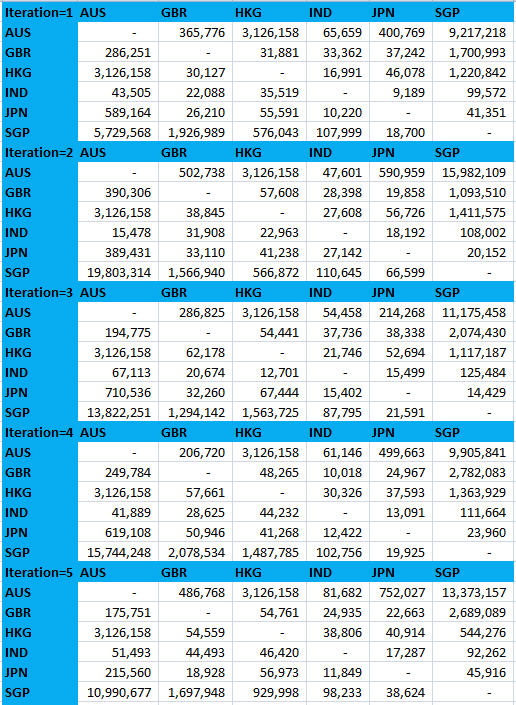}

}\caption{Distance Measures over Randomly Chosen Sub Universe \label{fig:Results-with-Randomly}}

\end{figure}

\end{doublespace}
\begin{doublespace}

\subsection{\label{sec:Comparison-Volumes-Prices-Volatilities}Comparison of
Security Trading Volumes, High-Low-Open-Close Prices and Volume-Price
Volatilities}
\end{doublespace}

\begin{doublespace}
We provide several examples in (sections \ref{subsec:Speaking-Volumes-Of:},
\ref{subsec:A-Pricey-Prescription:}, \ref{subsec:Taming-the-(Volatility)})
complementary to (section \ref{subsec:Comparison-of-Security}) of
how our techniques could be used to compare different countries based
on the time series variables across all equity securities traded in
that market. These illustrations are discussed in much greater detail
as part of a complete study on market micro-structure in (Kashyap
2017). The data sample contains prices (open, close, high and low)
and trading volumes for most of the securities from six different
markets from Jan 01, 2014 to May 28, 2014 (Figure \ref{fig:Markets-and-Tickers-1}).
The time period under consideration, the number of securities and
the bulk of the numerical analysis is similar to the comparisons done
in section \ref{subsec:Comparison-of-Security}. To compare volatilities
of prices and volumes, we calculate sixty day moving volatilities
on the close price and trading volume and calculate the distance measure
over the full sample and also across each of the randomly selected
sub-samples.
\end{doublespace}
\begin{doublespace}

\subsection{\label{subsec:Speaking-Volumes-Of:}Speaking Volumes Of: Comparison
of Trading Volumes}
\end{doublespace}

\begin{doublespace}
The results of the volume comparison over the full sample are shown
in Figure \ref{fig:Volume-PCA-Results}. For example, in Figure \ref{fig:Volume-PCA-Results},
AUS - GBR are the most similar markets when two significant digits
are used and AUS - GBR are the most similar with six significant digits.
In this case the PCA and JL Lemma dimension reduction give similar
results.

The random sample results are shown in Figure \ref{fig:Volume-Results-with-Randomly}.
The left table (Figure \ref{fig:Volume-PCA-Dimension-Reduction})
is for PCA reduction on a randomly chosen sub universe and the right
table (Figure \ref{fig:Volume-JL-Lemma-Dimenion}) is for dimension
reduction using JL Lemma for the same sub universe.

\begin{figure}[H]
\subfloat[Markets and Tickers Count\label{fig:Markets-and-Tickers-1}]{\includegraphics{Country_Ticker_Count}

}\hfill{}\subfloat[Volume PCA Dimension Reduction\label{fig:Volume-PCA-Results}]{\includegraphics[width=7cm,height=10cm]{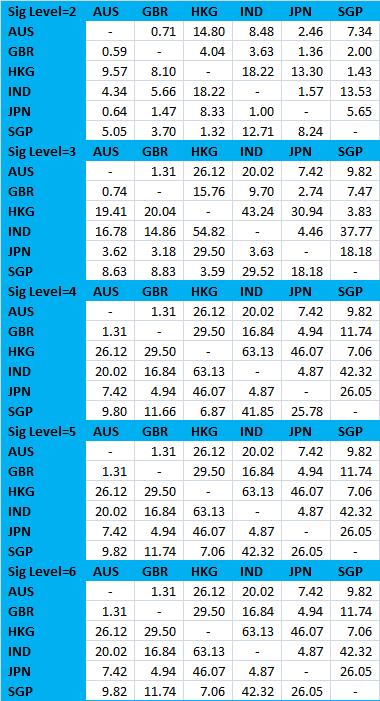}

}

\caption{Security Count by Market / Volume Distance Measures over Full Sample\label{fig:Security-Count-Volume-Distance}}
\end{figure}
\begin{figure}[H]
\subfloat[Volume PCA Dimension Reduction\label{fig:Volume-PCA-Dimension-Reduction}]{\includegraphics[width=7cm,height=10cm]{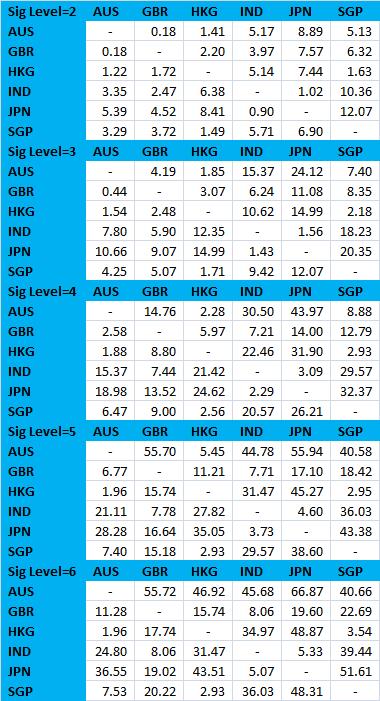}

}\hfill{}\subfloat[Volume JL Lemma Dimension Reduction\label{fig:Volume-JL-Lemma-Dimenion}]{\includegraphics[width=8cm,height=10cm]{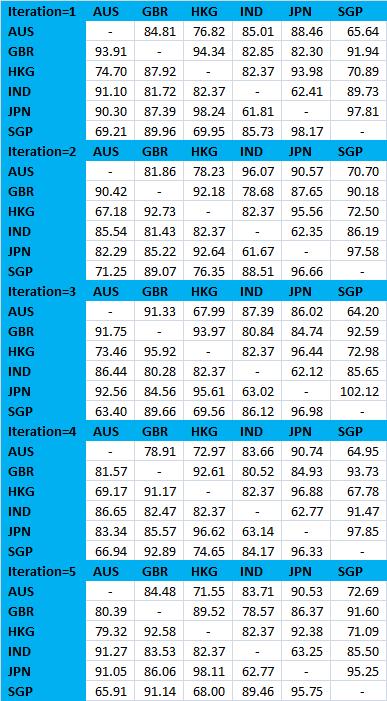}

}\caption{Volume Distance Measures over Randomly Chosen Sub Universe \label{fig:Volume-Results-with-Randomly}}
\end{figure}

\end{doublespace}
\begin{doublespace}

\subsection{\label{subsec:A-Pricey-Prescription:}A Pricey Prescription: Comparison
of Prices (Open, Close, High and Low)}
\end{doublespace}
\begin{doublespace}

\subsubsection{Open Close}
\end{doublespace}

\begin{doublespace}
The results of a comparison between open and close prices over the
full sample are shown in Figures \ref{fig:Open-Close-Distance-Measures-Full Sample},
\ref{fig:Open-PCA-Dimension-Reduction}, \ref{fig:Close-PCA-Dimension-Reduction}.
For example, in Figure \ref{fig:Close-PCA-Dimension-Reduction}, AUS
- SGP are the most similar markets when two significant digits are
used and AUS - HKG are the most similar with six significant digits.
The similarities between open and close prices, in terms of the distance
measure, are also easily observed.

The random sample results are shown in Figures \ref{fig:Open-Results-with-Randomly},
\ref{fig:Close-Results-with-Randomly}. The left table (Figures \ref{fig:Open-PCA-Dimension-Reduction-Random},
\ref{fig:Close-PCA-Dimension-Reduction-Random}) is for PCA reduction
on a randomly chosen sub universe and the right table (Figures \ref{fig:Open-JL-Lemma-Dimenion},
\ref{fig:Close-JL-Lemma-Dimenion}) is for dimension reduction using
JL Lemma for the same sub universe. In Figure \ref{fig:Close-JL-Lemma-Dimenion},
AUS - IND are the most similar in iteration one and also in iteration
five.

\begin{figure}[H]
\subfloat[Open PCA Dimension Reduction\label{fig:Open-PCA-Dimension-Reduction}]{\includegraphics[width=7.5cm,height=10cm]{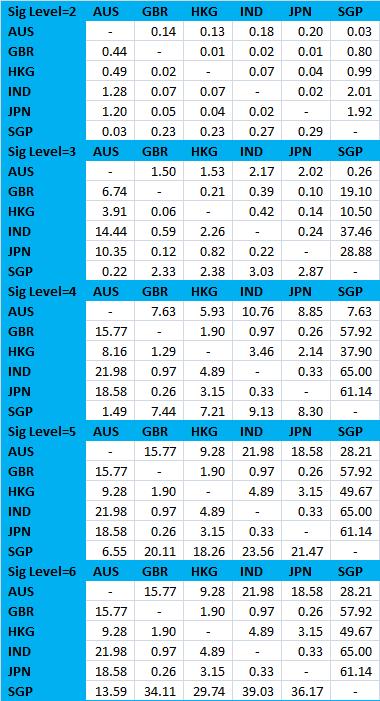}

}\hfill{}\subfloat[Close PCA Dimension Reduction\label{fig:Close-PCA-Dimension-Reduction}]{\includegraphics[width=7.5cm,height=10cm]{Close_Distance_PCA_Reduction}

}\caption{Open / Close Distance Measures over Full Sample\label{fig:Open-Close-Distance-Measures-Full Sample}}
\end{figure}

\begin{figure}[H]
\subfloat[Open PCA Dimension Reduction\label{fig:Open-PCA-Dimension-Reduction-Random}]{\includegraphics[width=7cm,height=10cm]{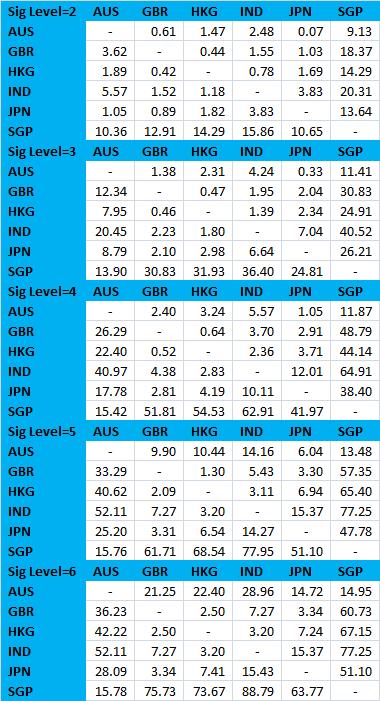}

}\hfill{}\subfloat[Open JL Lemma Dimension Reduction\label{fig:Open-JL-Lemma-Dimenion}]{\includegraphics[width=8cm,height=10cm]{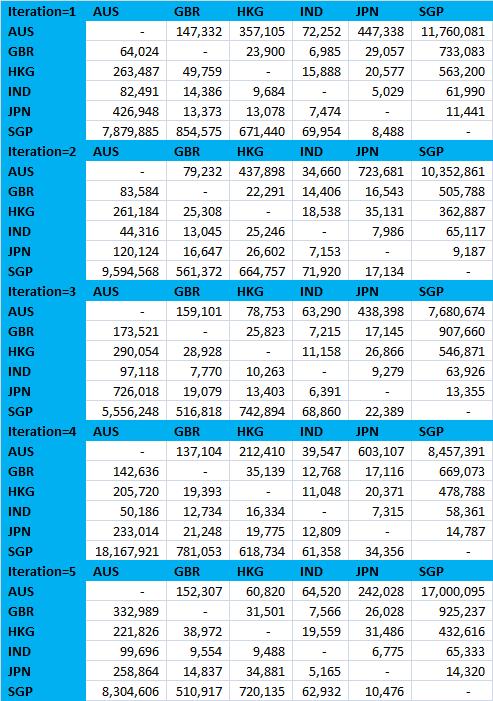}

}\caption{Open Distance Measures over Randomly Chosen Sub Universe \label{fig:Open-Results-with-Randomly}}
\end{figure}

\begin{figure}[H]
\subfloat[Close PCA Dimension Reduction\label{fig:Close-PCA-Dimension-Reduction-Random}]{\includegraphics[width=7cm,height=10cm]{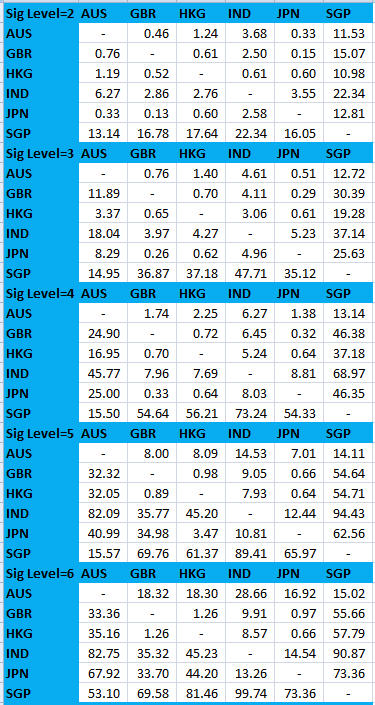}

}\hfill{}\subfloat[Close JL Lemma Dimension Reduction\label{fig:Close-JL-Lemma-Dimenion}]{\includegraphics[width=8cm,height=10cm]{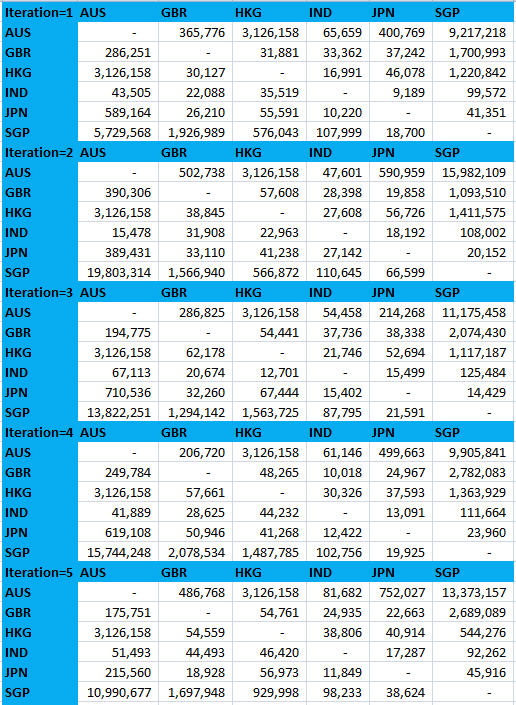}

}\caption{Close Distance Measures over Randomly Chosen Sub Universe \label{fig:Close-Results-with-Randomly}}
\end{figure}

\end{doublespace}
\begin{doublespace}

\subsubsection{High Low}
\end{doublespace}

\begin{doublespace}
The results of a comparison between high and low prices over the full
sample are shown in Figures \ref{fig:High-Low-Distance-Measures-Full Sample},
\ref{fig:High-PCA-Dimension-Reduction}, \ref{fig:Low-PCA-Dimension-Reduction}.
For example, in Figure \ref{fig:High-PCA-Dimension-Reduction}, AUS
- SGP are the most similar markets when two significant digits are
used and AUS - HKG are the most similar with six significant digits.
The similarities between high and low prices are also easily observed.

The random sample results are shown in Figures \ref{fig:High-Results-with-Randomly},
\ref{fig:Low-Results-with-Randomly}. The left table (Figures \ref{fig:High-PCA-Dimension-Reduction-Random},
\ref{fig:Low-PCA-Dimension-Reduction-Random}) is for PCA reduction
on a randomly chosen sub universe and the right table (Figures \ref{fig:High-JL-Lemma-Dimenion},
\ref{fig:Low-JL-Lemma-Dimension}) is for dimension reduction using
JL Lemma for the same sub universe. In Figures \ref{fig:High-JL-Lemma-Dimenion}
and \ref{fig:Low-JL-Lemma-Dimension}, AUS - IND are the most similar
in iteration one and also in iteration five.

\begin{figure}[H]
\subfloat[High PCA Dimension Reduction\label{fig:High-PCA-Dimension-Reduction}]{\includegraphics[width=7.5cm,height=10cm]{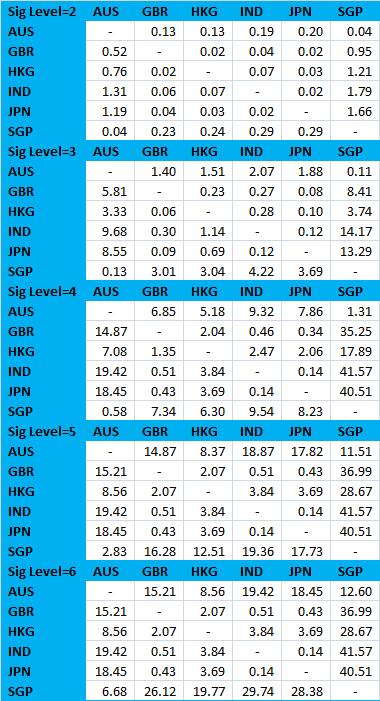}

}\hfill{}\subfloat[Low PCA Dimension Reduction\label{fig:Low-PCA-Dimension-Reduction}]{\includegraphics[width=7.5cm,height=10cm]{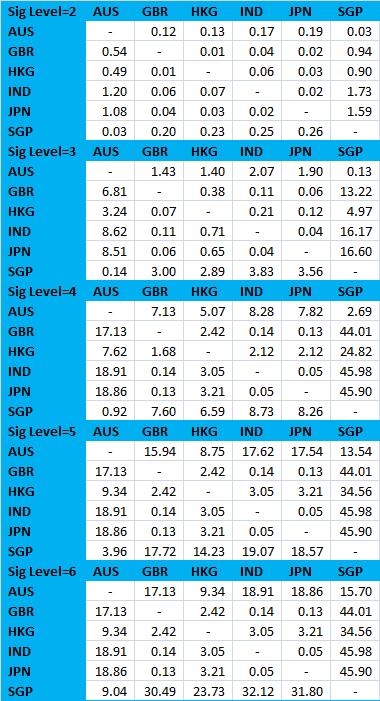}

}\caption{High / Low Distance Measures over Full Sample\label{fig:High-Low-Distance-Measures-Full Sample}}
\end{figure}

\begin{figure}[H]
\subfloat[High PCA Dimension Reduction\label{fig:High-PCA-Dimension-Reduction-Random}]{\includegraphics[width=7cm,height=10cm]{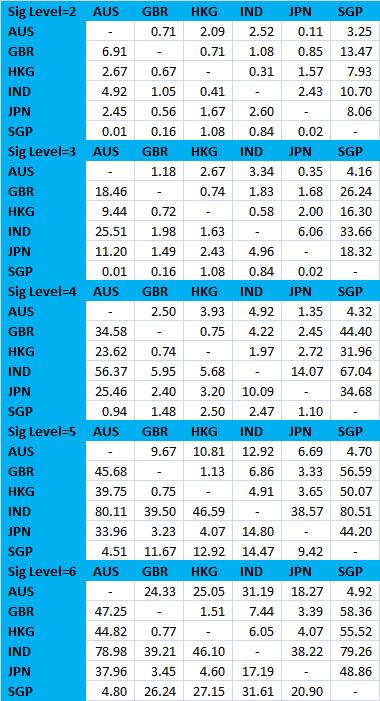}

}\hfill{}\subfloat[High JL Lemma Dimension Reduction\label{fig:High-JL-Lemma-Dimenion}]{\includegraphics[width=8cm,height=10cm]{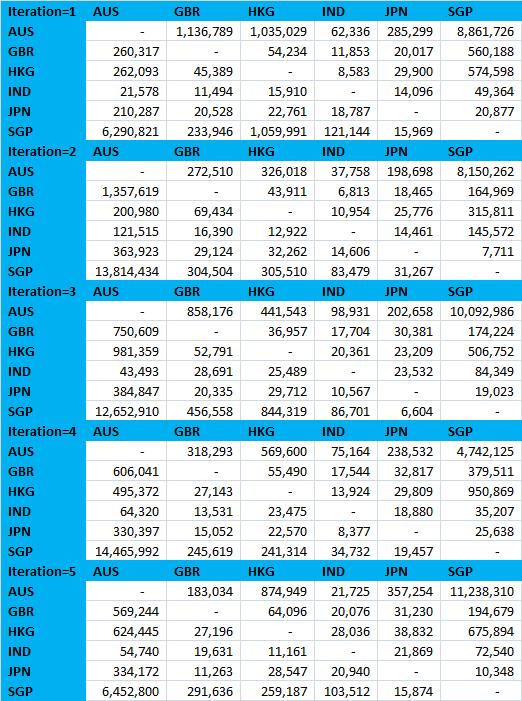}

}\caption{High Distance Measures over Randomly Chosen Sub Universe \label{fig:High-Results-with-Randomly}}
\end{figure}

\begin{figure}[H]
\subfloat[Low PCA Dimension Reduction\label{fig:Low-PCA-Dimension-Reduction-Random}]{\includegraphics[width=7cm,height=10cm]{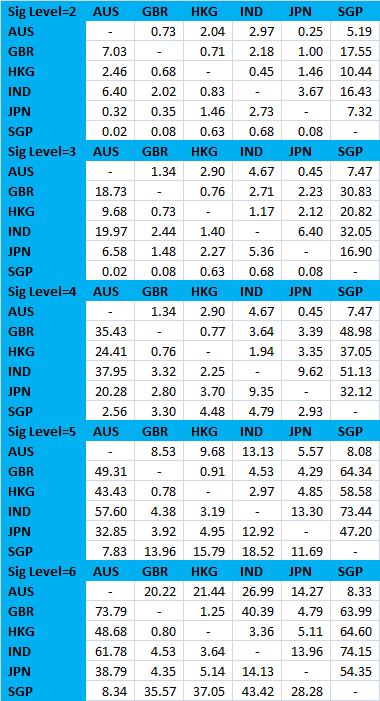}

}\hfill{}\subfloat[Low JL Lemma Dimension Reduction\label{fig:Low-JL-Lemma-Dimension}]{\includegraphics[width=8cm,height=10cm]{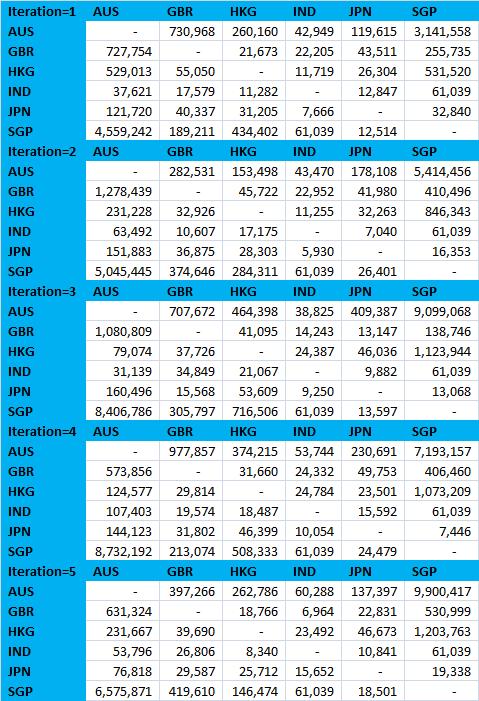}

}\caption{Low Distance Measures over Randomly Chosen Sub Universe \label{fig:Low-Results-with-Randomly}}
\end{figure}

\end{doublespace}
\begin{doublespace}

\subsection{\label{subsec:Taming-the-(Volatility)}Taming the (Volatility) Skew:
Comparison of Close Price / Volume Volatilities}
\end{doublespace}

\begin{doublespace}
The results of a comparison between close price volatilities and volume
volatilities over the full sample are shown in Figures \ref{fig:Close-Volume-Volatility-Distance-Measures-Full Sample},
\ref{fig:Close-Volatility-PCA-Dimension-Reduction}, \ref{fig:Volume-Volatility-PCA-Dimension-Reduction}.
For example, in Figure \ref{fig:Close-Volatility-PCA-Dimension-Reduction},
AUS - GBR are the most similar markets when two significant digits
are used and AUS - HKG are the most similar with six significant digits.
In Figure \ref{fig:Volume-Volatility-PCA-Dimension-Reduction}, AUS
- GBR - IND are equally similar markets when two significant digits
are used and AUS - GBR are the most similar with six significant digits.
The difference in magnitudes of the distance measures for prices,
volumes and volatilities is also easily observed. What this indicates
is that, prices are from the most dissimilar or distant distributions,
volatilities are less similar and volumes are from the most similar
or overlapping distributions. As also observed in the volume comparisons,
volume volatility comparisons give seemingly similar results when
PCA or JL Lemma dimension reductions are used. By considering the
price volatilities and creating portfolios of instruments that have
dissimilar volatility distributions, we could reduce the overall risk
or variance of the portfolio returns, becoming one potential way of
mitigating the effects of wild volatility swings.

The random sample results are shown in Figures \ref{fig:Close-Volatility-Results-with-Randomly},
\ref{fig:Volume-Volatility-Results-with-Randomly}. The left table
(Figures \ref{fig:Close-Volatility-PCA-Dimension-Reduction-Random},
\ref{fig:Volume-Volatility-PCA-Dimension-Reduction-Random}) is for
PCA reduction on a randomly chosen sub universe and the right table
(Figures \ref{fig:Close-Volatility-JL-Lemma-Dimension}, \ref{fig:Volume-Volatility-JL-Lemma-Dimension})
is for dimension reduction using JL Lemma for the same sub universe.
In Figure \ref{fig:Close-Volatility-JL-Lemma-Dimension} AUS - SGP
are the most similar in iteration one and also in iteration five.
In Figure \ref{fig:Volume-Volatility-JL-Lemma-Dimension}, AUS - SGP
are the most similar in iteration one and AUS- GBR in iteration five.

\begin{figure}[H]
\subfloat[Close Volatility PCA Dimension Reduction\label{fig:Close-Volatility-PCA-Dimension-Reduction}]{\includegraphics[width=7.5cm,height=10cm]{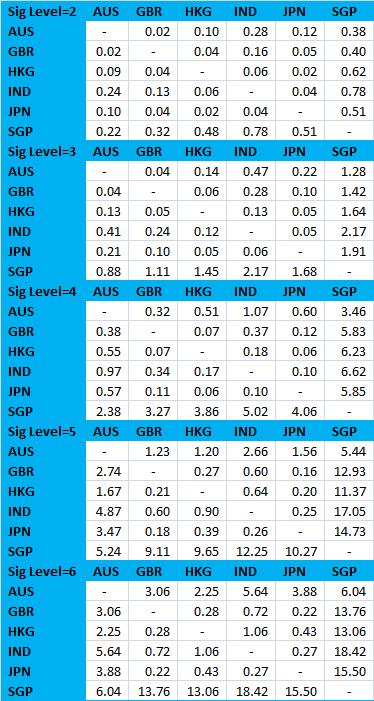}

}\hfill{}\subfloat[Volume Volatility PCA Dimension Reduction\label{fig:Volume-Volatility-PCA-Dimension-Reduction}]{\includegraphics[width=7.5cm,height=10cm]{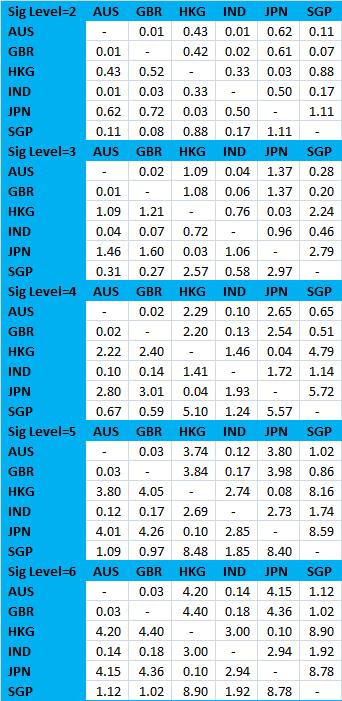}

}\caption{Close / Volume Volatility Distance Measures over Full Sample\label{fig:Close-Volume-Volatility-Distance-Measures-Full Sample}}
\end{figure}

\begin{figure}[H]
\subfloat[Close Volatility PCA Dimension Reduction\label{fig:Close-Volatility-PCA-Dimension-Reduction-Random}]{\includegraphics[width=7cm,height=10cm]{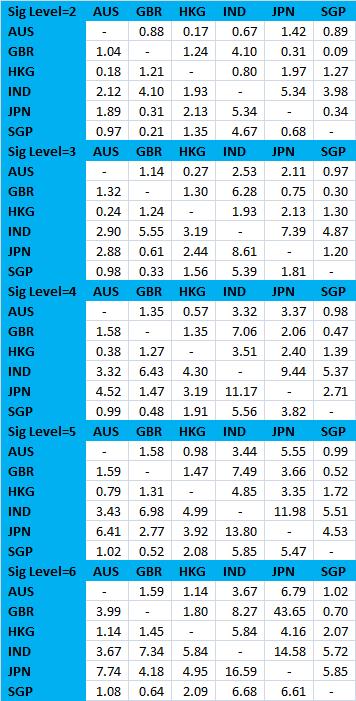}

}\hfill{}\subfloat[Close Volatility JL Lemma Dimension Reduction\label{fig:Close-Volatility-JL-Lemma-Dimension}]{\includegraphics[width=8cm,height=10cm]{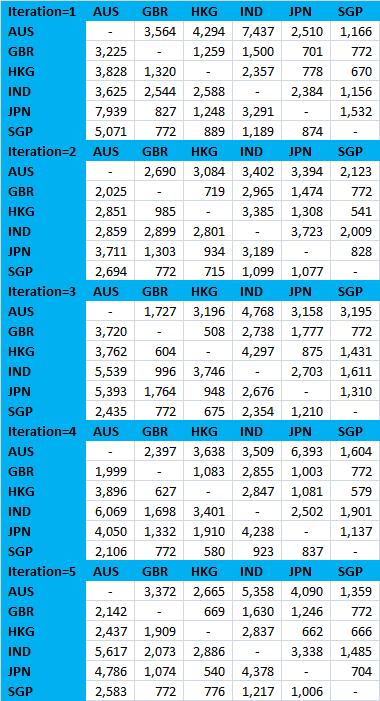}

}\caption{Close Volatility Distance Measures over Randomly Chosen Sub Universe
\label{fig:Close-Volatility-Results-with-Randomly}}
\end{figure}

\begin{figure}[H]
\subfloat[Volume Volatility PCA Dimension Reduction\label{fig:Volume-Volatility-PCA-Dimension-Reduction-Random}]{\includegraphics[width=7cm,height=10cm]{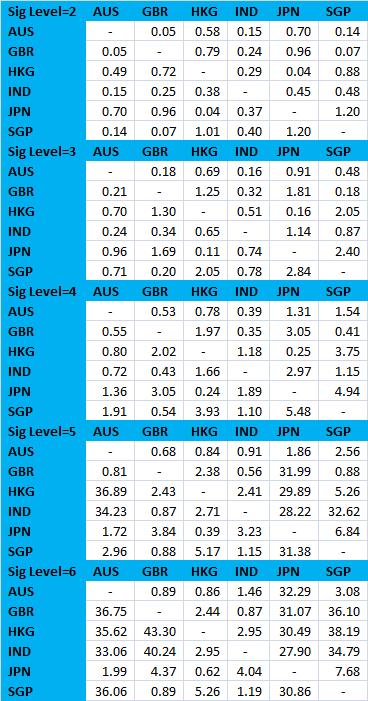}

}\hfill{}\subfloat[Volume Volatility JL Lemma Dimension Reduction\label{fig:Volume-Volatility-JL-Lemma-Dimension}]{\includegraphics[width=8cm,height=10cm]{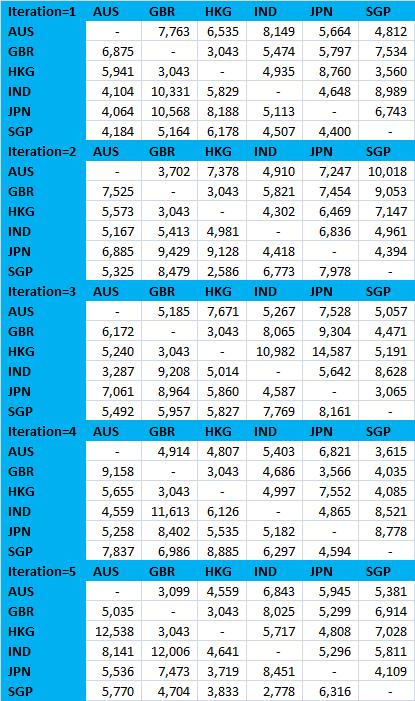}

}\caption{Volume Volatility Distance Measures over Randomly Chosen Sub Universe
\label{fig:Volume-Volatility-Results-with-Randomly}}
\end{figure}

\end{doublespace}

\subsection{\label{subsec:Other-Complementary-Case}Other Complementary Case
Studies}

Below we mention examples of studies we are planning to undertake
that can act as applications of the methodologies discussed here.

\subsubsection{\label{subsec:United-Kingdom-Crime}United Kingdom Crime Analysis}

\begin{doublespace}
We are developing a separate application which applies the techniques
in this article and also from network theory (Euler 1953; Wasserman
\& Faust 1994; Watts 1999; Aldous \& Wilson 2003; Gribkovskaia, Halskau
Sr. \& Laporte 2007) to a crime data-set from the United Kingdom.
The goal of this project is to identify regions with similar levels
and types of crimes so that policy and enforcement decisions can be
made based on lessons learnt from earlier episodes or from regions
where certain legal measures have had a decent level of success. In
addition, network interactions can help us to understand how crime
waves might be spreading from one region to another and help with
crime prevention. The crime data-set is publicly available at: \href{https://data.police.uk/data}{UK Police Data}.
\end{doublespace}

\subsubsection{\label{subsec:Hong-Kong-Shanghai}Hong Kong and Shanghai Securities
Comparison across Industrial Sectors}

\begin{doublespace}
We are creating another numerical example on financial market data
to find out which sectors in China and Hong Kong are more similar
based on the stock prices and trading volumes of all securities listed
on the Shanghai and Hong Kong stock exchange (Kashyap 2017 considers
this question around the time when the Hong Kong and Shanghai stock
exchanges were first electronically connected; though the main focus
of this study is on how trading costs have changed due to this linking
of the two financial markets). To perform this study, we could use
any publicly available data-set downloadable from various market data
providers.
\end{doublespace}

\subsection{\label{subsec:Implementation-Pointers}Implementation Pointers}

\begin{doublespace}
(Chaussé 2010) is a good reference for estimating the parameters of
a normal distribution through the Generalized Method of Moments, GMM,
(Cochrane 2009) using R package \textbf{gmm}. The numerical computation
of a multivariate normal is often a difficult problem. (Genz 1992;
Genz \& Bretz 2009) describe a transformation that simplifies the
problem and places it into a form that allows efficient calculation
using standard numerical multiple integration algorithms. (Hothorn,
Bretz \& Genz 2001) give pointers for the numerical computation of
multivariate normal probabilities based on the R package \textbf{mvtnorm}.
(Manjunath \& Wilhelm 2012) derive an explicit expression for the
truncated mean and variance for the multivariate normal distribution
with arbitrary rectangular double truncation. (Wilhelm \& Manjunath
2010; Wilhelm 2015) have details on the R package \textbf{tmvtnorm},
for the truncated multivariate normal distribution including routines
for the GMM Estimation for the Truncated Multivariate Normal Distribution.
In (appendix \ref{sec:R-Code-Snippets}), we list some R code snippets,
that includes the Johnson-Lindenstrauss matrix transformation and
a modification to the routine to calculate the Bhattacharyya distance,
available currently in the R package \textbf{fps}. This modification
allows much larger numbers and dimensions to be handled, by utilizing
the properties of logarithms and the Eigen values of a matrix.
\end{doublespace}
\begin{doublespace}

\section{\label{sec:Possibilities-for-Future}Possibilities for Future Research }
\end{doublespace}

\begin{doublespace}
There is a lot of research being done in developing distance measures,
dimension reduction techniques and understanding the properties of
distributions whose dimensions have been transformed. Such studies
are aimed at developed better theoretical foundations as well faster
algorithms for implementing newer techniques. We point out alternative
approaches that might hold potential starting with some not so recent
methods and then some newer techniques. (Chow \& Liu 1968) present
a method to approximate an $n$ dimensional discrete distribution
by a product of several of its component distributions of lower order
in such a way that the product is a probability extension of these
distributions of lower order. A product approximation is a valid probability
distribution. Only the class of second order distributions are used.
There are $n\left(n-1\right)/2$ second order approximations of which
at most $n-1$ can be used in the approximation. 

The uncertainty associated with a state estimate can be represented
most generally by a probability distribution incorporating all knowledge
about the state. The Kalman filter (Kalman 1960; End-note \ref{enu:Kalman-Filter})
exploits the fact that 1) given only the mean and variance (or covariance
in multiple dimensions) of a distribution, the most conservative assumption
that can be made about the distribution is that it is a Gaussian having
the given mean and variance and 2) the fact that the application of
a linear operator to a Gaussian distribution always yields another
Gaussian distribution. Given the assumptions of 1) and 2) it is straightforward
to show that the Kalman filter will yield the best possible estimate
of the mean and variance of the state. The requirement that the mean
and variance of the state is measurable represents little practical
difficulty but the requirement that all observation and process models
be linear is rarely satisfied in nontrivial applications. (Julier
\& Uhlmann 1996) examine an alternative generalization of the Kalman
filter that accommodates nonlinear transformations of probability
distributions through the use of a new representation of the mean
and variance information about the current state. 

(Székely, Rizzo \& Bakirov 2007; Székely \& Rizzo 2009; Lyons 2013)
develop a new measure of dependence between random vectors called
the Distance correlation. Distance covariance and distance correlation
are analogous to product-moment covariance and correlation, but unlike
the classical definition of correlation, distance correlation is zero
only if the random vectors are independent. This is applicable to
random vectors of arbitrary and not necessarily equal dimension and
to any distributions with finite first moments. (Székely \& Rizzo
2013) come up with the Energy distance, which is a statistical distance
between the distributions of random vectors characterizing equality
of distributions. The name energy derives from Newton’s gravitational
potential energy (which is a function of the distance between two
bodies), and there is an elegant relation to the notion of potential
energy between statistical observations. Energy statistics (E-statistics)
are functions of distances between statistical observations. The idea
of energy statistics is to consider statistical observations as heavenly
bodies governed by a statistical potential energy, which is zero if
and only if an underlying statistical null hypothesis is true.
\end{doublespace}

We chronicle some key simplifications we have used in our illustrative
examples:
\begin{enumerate}
\item A key limitation of our study is that we have reduced dimensions using
PCA or randomly sampled a sub matrix from the overall data-set so
that the length of time series available is in the range of the number
of securities that could be compared. Using a longer time series for
the variables would be a useful extension and a real application would
benefit immensely from more history.
\item We have used the simple formula for the Bhattacharyya distance applicable
to multivariate normal distributions. The formulae we have developed
over a truncated multivariate normal distribution or using a Normal
Log-Normal Mixture could give more accurate results. Again, later
works should look into tests that can establish which of the distributions
would apply depending on the data-set under consideration.
\item For the examples in section \ref{sec:Comparison-Volumes-Prices-Volatilities},
for each market we have looked at seven variables, open, close, low,
high, trading volume, close volatility and volume volatility. These
variables can be combined using the expression for the multinomial
distance to get a complete representation of which markets are more
similar than others. We aim to develop this methodology and illustrate
these techniques further in later works.
\end{enumerate}
Once we have the similarity measures across groups of entities, a
separate set of analysis can be performed to prudently select which
policies or procedures can be transferred across similar entities.
For example, across groups of securities, portfolios could be constructed
to see how sensitive they are to different explanatory factors and
then performance benchmarks could be used to gauge the risk return
relationship.

When the distance is somewhat close between two pairs of entities,
we could observe an inconsistency due to the JL Lemma transformation
in successive iterations. It is worth remembering that if we perform
multiple iterations of the JL Lemma dimension transformation (multiplying
by a different random matrix); the distance could be slightly different
in each iteration but within the bound established by JL Lemma. Though,
in one iteration one pair could seem to be more similar than the other
and the result could be reversed in the other iteration. When compared
to other dimension reduction techniques such as PCA, due to loss of
data, it is hard to know if the results are accurate at all when the
distance between entities is too small. Theoretical bounds could be
established with regards to what are the limits on the distance measure
where such an inconsistency could result. Further to this, a huge
branch of theoretical investigations can be undertaken regarding the
interval over which the distance measure will fall when one of the
random variables is multiplied by different random matrices governed
by the JL-Lemma to effect the same dimension transformation. 

A key point to remember is that the findings with respect to the similarity
of different entities are based on the data available for the particular
time period under consideration. There could be fundamental changes
in the true data generating processes, which are generally unknown
and the same entities for a different time period could give entirely
different results. Hence, any real application that continues to collect
the relevant data that is generated, should calculate the metrics
on a moving time period basis (rolling) and use the latest results
that are computed. Needless to say, this brings up the question of
how much of the recent history we should use as we make our decisions.
While our present econometric tools fall short of providing concrete
answers, many practical aids and suggestions are available in any
econometric guide (Hamilton 1994). This is of course true for most
(all?) empirical data investigations.

\begin{doublespace}
Another improvement could be to normalize, or, standardize distance
measures. This is perhaps applicable not just to the Bhattacharyya
distance, but to other types of distance metrics as well. For example,
the covariance is a useful number, but its use is magnified and the
intuitive lessons compounded by the use of correlations derived from
covariance measures. Again, we need to be wary of this path since
anything that can be measured is magical; that is, there is always
something bigger or something smaller, which can be expressed as a
need to contend with the concepts of the infinite, or, the infinitesimal.
So, by standardizing something we might reduce the range of the values,
but it might be necessary to retain the number of significant digits
that give us meaningful comparisons. Also, since we wish to compare
more than two entities, it might be helpful to come up with metrics
that can combine multiple distance measures at the same time. Though,
it would suffice to be able to combine two of the distance measures
at a time, since that is how we would compare any two probabilistic
entities; just as in the covariance case other possibilities such
as combining more than two at a time need not be ruled out.
\end{doublespace}
\begin{doublespace}

\section{Conclusions}
\end{doublespace}

\begin{doublespace}
We have discussed various measures of similarity between probability
distributions. We have shown how the marriage between the Bhattacharyya
distance and the Johnson-Lindenstrauss Lemma provides us with a novel
and practical methodology that allows comparisons between any two
probability distributions.\textbf{\textit{ }}This combination becomes,
to our limited knowledge, an example (the example?) of perfectly harmonious
matrimony. We demonstrated a relationship between covariance and distance
measures based on a generic extension of Stein's Lemma. We have provided
numerical examples based on security prices across six countries where
we have compared trading volumes, open-close-high-low prices, price
and volume volatilities; all of which act as how to guides of the
techniques developed here for real life applications. We have also
discussed how this methodology lends itself to numerous applications
outside the realm of finance and economics.
\end{doublespace}
\begin{doublespace}

\section{Acknowledgements and End-notes}
\end{doublespace}
\begin{enumerate}
\begin{doublespace}
\item My colleagues at SolBridge International School of Business are always
full of great suggestions to improve all my papers; in particular:
Dr. Chia Hsing Huang, Dr. Sangoo Shin, Dr. Xingcai Meng, Dr. Yun Jie
Joseph, Dr. Narasimha Rao Kowtha, Dr. Aye Mengistu Alemu, Dr. Ben
Agyei-Mensah, Dr. Taylan Urkmez, Dr. Alejandra Marin, Dr. Jerman Rose,
Dr. Awan Mahmood, Dr. Jay Won Lee and Dr. KyunHwa Kim raised some
interestingly tough questions at our monthly Brown Bag seminars that
resulted in significant improvements in section \ref{sec:Intuition-for-Dimension}
and section \ref{sec:Possibilities-for-Future}. The students of SolBridge
International School of Business continue to be the inspiration for
many of our papers by bringing us their concerns. In this case their
worries about how their daily life is being bombarded by numerous
sources of information showed us that information reduction would
have many potential benefits.
\end{doublespace}
\item Numerous seminar participants, particularly at a few meetings of the
econometric society and various finance organizations, and the faculty
members of City University of Hong Kong suggested ways to improve
the core results in the manuscript and prompted us to find ways to
express the intuition for why we need the main results. The views
and opinions expressed in this article, along with any mistakes, are
mine alone and do not necessarily reflect the official policy or position
of either of my affiliations or any other agency.
\begin{doublespace}
\item \label{enu:The-Mellin-transform}In mathematics, the Mellin transform
is an integral transform that may be regarded as the multiplicative
version of the two-sided Laplace transform (End-note \ref{enu:The-two-sided-Laplace}).
The Mellin transform of a function $f$ is
\[
{\displaystyle \left\{ {\mathcal{M}}f\right\} (s)=\varphi(s)=\int_{0}^{\infty}x^{s-1}f(x)\,dx}
\]

\end{doublespace}

The inverse transform is
\[
\left\{ \mathcal{M}^{-1}\varphi\right\} (x)=f(x)=\frac{1}{2\pi i}\int_{c-i\infty}^{c+i\infty}x^{-s}\varphi(s)\,ds
\]
This integral transform is closely connected to the theory of Dirichlet
series (End-note \ref{enu:A-Dirichlet-series}), and is often used
in number theory, mathematical statistics, and the theory of asymptotic
expansions; it is closely related to the Laplace transform (End-note
\ref{enu:The-Laplace-transform}) and the Fourier transform (End-note
\ref{enu:The-Fourier-transform}), and the theory of the gamma function
and allied special functions. To avoid extended discussions, we only
provide some details here. The following link has more details: \href{https://en.wikipedia.org/wiki/Mellin_transform}{Mellin Transform, Wikipedia Link}
\item \label{enu:The-G-function-was}In mathematics, the G-function was
introduced by Cornelis Simon Meijer (Meijer 1936) as a very general
function intended to include most of the known special functions as
particular cases. This was not the only attempt of its kind: the generalized
hypergeometric function and the MacRobert E-function had the same
aim, but Meijer's G-function was able to include those as particular
cases as well. The first definition was made by Meijer using a series;
nowadays the accepted and more general definition is via a line integral
in the complex plane, introduced in its full generality by Arthur
Erdelyi in 1953. \href{https://en.wikipedia.org/wiki/Meijer_G-function}{Meijer G-Function, Wikipedia Link}
\begin{doublespace}
\item \label{enu:The-Newton-Cotes-formulas,}In numerical analysis, the
Newton-Cotes formulas, also called the Newton-Cotes quadrature rules
or simply Newton-Cotes rules, are a group of formulas for numerical
integration (also called quadrature) based on evaluating the integrand
at equally spaced points. \href{https://en.wikipedia.org/wiki/Newton–Cotes_formulas}{Newton-Cotes formulas, Wikipedia Link}
\item \label{enu:Rao's-score-LM-test,}In statistics, Rao's score test,
also known as the score test or the Lagrange multiplier test (LM test)
in econometrics, is a statistical test of a simple null hypothesis
that a parameter of interest $\theta$ is equal to some particular
value $\theta_{0}$. It is the most powerful test when the true value
of $\theta$ is close to $\theta_{0}$. The main advantage of the
score test is that it does not require an estimate of the information
under the alternative hypothesis or unconstrained maximum likelihood.
\href{https://en.wikipedia.org/wiki/Score_test}{Score or Lagrange Multiplier Test, Wikipedia Link}
\end{doublespace}
\item \label{enu:Skewness}In probability theory and statistics, skewness
is a measure of the asymmetry of the probability distribution of a
real-valued random variable about its mean. The skewness value can
be positive or negative, or undefined. \href{https://en.wikipedia.org/wiki/Skewness}{Skewness, Wikipedia Link}
\item \label{enu:Kurtosis}In probability theory and statistics, kurtosis
(from Greek: kyrtos or kurtos, meaning ``curved, arching'') is a
measure of the ``tailedness'' of the probability distribution of
a real-valued random variable. In a similar way to the concept of
skewness, kurtosis is a descriptor of the shape of a probability distribution
and, just as for skewness, there are different ways of quantifying
it for a theoretical distribution and corresponding ways of estimating
it from a sample from a population. Depending on the particular measure
of kurtosis that is used, there are various interpretations of kurtosis,
and of how particular measures should be interpreted. \href{https://en.wikipedia.org/wiki/Kurtosis}{Kurtosis, Wikipedia Link}
\begin{doublespace}
\item \label{enu:The-Euclidean-distance}In mathematics, the Euclidean distance
or Euclidean metric is the ``ordinary'' straight-line distance between
two points in Euclidean space. With this distance, Euclidean space
becomes a metric space. The associated norm is called the Euclidean
norm. Older literature refers to the metric as the Pythagorean metric.
A generalized term for the Euclidean norm is the $L^{2}$ norm or
$L^{2}$ distance. \href{https://en.wikipedia.org/wiki/Euclidean_distance}{Euclidean Distance, Wikipedia Link}
\end{doublespace}
\item \label{enu:Kalman-Filter}In statistics and control theory, Kalman
filtering, also known as linear quadratic estimation (LQE), is an
algorithm that uses a series of measurements observed over time, containing
statistical noise and other inaccuracies, and produces estimates of
unknown variables that tend to be more accurate than those based on
a single measurement alone, by estimating a joint probability distribution
over the variables for each timeframe. The Kalman filter has numerous
applications for guidance, navigation, and control of vehicles, particularly
aircraft and spacecraft. Furthermore, the Kalman filter is a widely
applied concept in time series analysis used in fields such as signal
processing and econometrics. Kalman filters also are one of the main
topics in the field of robotic motion planning and control, and they
are sometimes included in trajectory optimization. \href{https://en.wikipedia.org/wiki/Kalman_filter}{Kalman Filter, Wikipedia Link}
\item \label{enu:A-Dirichlet-series}In mathematics, a Dirichlet series
is any series of the form 
\[
{\displaystyle \sum_{n=1}^{\infty}{\frac{a_{n}}{n^{s}}}}
\]
where $s$ is complex, and $a_{n}$ is a complex sequence. \href{https://en.wikipedia.org/wiki/Dirichlet_series}{Dirichlet Series, Wikipedia Link}
\begin{doublespace}
\item \label{enu:The-two-sided-Laplace}In mathematics, the two-sided Laplace
transform or bilateral Laplace transform is an integral transform
equivalent to probability's moment generating function. Two-sided
Laplace transforms are closely related to the Fourier transform, the
Mellin transform, and the ordinary or one-sided Laplace transform.
If $f\left(t\right)$ is a real or complex valued function of the
real variable $t$ defined for all real numbers, then the two-sided
Laplace transform is defined by the integral,
\[
{\displaystyle {\mathcal{B}}\{f\}(s)=F(s)=\int_{-\infty}^{\infty}e^{-st}f(t)\,dt}
\]
\href{https://en.wikipedia.org/wiki/Two-sided_Laplace_transform}{Two-Sided Laplace Transform, Wikipedia Link}
\end{doublespace}
\item \label{enu:The-Laplace-transform}In mathematics, the Laplace transform
of a function $f\left(t\right)$, defined for all real numbers $t\geq0$,
is the function $F(s)$, which is a unilateral transform defined by
\[
{\displaystyle F(s)=\int_{0}^{\infty}f(t)e^{-st}\,dt}
\]
where $s$ is a complex number frequency parameter. $s=\sigma+i\omega$
, with real numbers $\sigma$ and $\omega$. An alternate notation
for the Laplace transform is ${\mathcal{L}}\{f\}$ instead of $F$.
\href{https://en.wikipedia.org/wiki/Laplace_transform}{Laplace Transform, Wikipedia Link}
\item \label{enu:The-Fourier-transform}The Fourier transform (FT) decomposes
a function of time (a signal) into the frequencies that make it up,
in a way similar to how a musical chord can be expressed as the frequencies
(or pitches) of its constituent notes. The Fourier transform of a
function of time is itself a complex-valued function of frequency,
whose absolute value represents the amount of that frequency present
in the original function, and whose complex argument is the phase
offset of the basic sinusoid in that frequency. The Fourier transform
is called the frequency domain representation of the original signal.
The term Fourier transform refers to both the frequency domain representation
and the mathematical operation that associates the frequency domain
representation to a function of time. The Fourier transform is not
limited to functions of time, but in order to have a unified language,
the domain of the original function is commonly referred to as the
time domain. The Fourier transform of the function $f$ is traditionally
denoted by adding a circumflex: $\hat{f}$. There are several common
conventions for defining the Fourier transform of an integrable function
${\displaystyle f:\mathbb{R}\to\mathbb{C}}$. One definition that
is commonly used for any real number $\xi$ is, 
\[
{\displaystyle {\hat{f}}(\xi)=\int_{-\infty}^{\infty}f(x)\ e^{-2\pi ix\xi}\,dx}
\]
When the independent variable $x$ represents time, the transform
variable $\xi$ represents frequency (e.g. if time is measured in
seconds, then the frequency is in hertz). Under suitable conditions,
$f$ is determined by $\hat{f}$ for any real number $x$, via the
inverse transform:
\[
{\displaystyle f(x)=\int_{-\infty}^{\infty}{\hat{f}}(\xi)\ e^{2\pi ix\xi}\,d\xi}
\]
\href{https://en.wikipedia.org/wiki/Fourier_transform}{Fourier Transform, Wikipedia Link}
\item (Lawson 1985; Keynes 1937; 1971; 1973; McManus \& Hastings 2005; Simon
1962; Kashyap 2017; Bertsekas 2002; Henderson \& Searle 1981) are
cited in the appendix of supplementary material. They discuss additional
concepts related to uncertainty, unintended consequences, probability
distributions and linear algebra. Some of these references are also
used in the proofs of the mathematical results.
\end{enumerate}
\begin{doublespace}

\section{References }
\end{doublespace}
\begin{enumerate}
\begin{doublespace}
\item Achlioptas, D. (2003). Database-friendly random projections: Johnson-Lindenstrauss
with binary coins. Journal of computer and System Sciences, 66(4),
671-687.
\item Aherne, F. J., Thacker, N. A., \& Rockett, P. I. (1998). The Bhattacharyya
metric as an absolute similarity measure for frequency coded data.
Kybernetika, 34(4), 363-368.
\end{doublespace}
\item Aldous, J. M., \& Wilson, R. J. (2003). Graphs and applications: an
introductory approach (Vol. 1). Springer Science \& Business Media.
\begin{doublespace}
\item Aroian, L. A. (1947). The probability function of the product of two
normally distributed variables. The Annals of Mathematical Statistics,
265-271.
\item Aroian, L. A., Taneja, V. S., \& Cornwell, L. W. (1978). Mathematical
forms of the distribution of the product of two normal variables.
Communications in Statistics-Theory and Methods, 7(2), 165-172.
\item Bansal, R., \& Yaron, A. (2004). Risks for the long run: A potential
resolution of asset pricing puzzles. The Journal of Finance, 59(4),
1481-1509.
\item Barro, R. J. (2006). Rare disasters and asset markets in the twentieth
century. The Quarterly Journal of Economics, 823-866.
\item Beath, C., Becerra-Fernandez, I., Ross, J., \& Short, J. (2012). Finding
value in the information explosion. MIT Sloan Management Review, 53(4),
18.
\item Berner, E. S., \& Moss, J. (2005). Informatics challenges for the
impending patient information explosion. Journal of the American Medical
Informatics Association, 12(6), 614-617.
\item Bertsekas, D. P. (2002). Introduction to Probability: Dimitri P. Bertsekas
and John N. Tsitsiklis. Athena Scientific.
\item Bhattacharyya, A. (1943). On a Measure of Divergence Between Two Statistical
Populations Defined by their Probability Distributions, Bull. Calcutta
Math. Soc., 35, pp. 99-110.
\item Bhattacharyya, A. (1946). On a measure of divergence between two multinomial
populations. Sankhyā: The Indian Journal of Statistics, 401-406.
\item Bishop, C. M. (2006). Pattern Recognition. Machine Learning.
\item Burges, C. J. (2009). Dimension reduction: A guided tour. Machine
Learning, 2(4), 275-365.
\item Burkardt, J. (2014). The Truncated Normal Distribution. Department
of Scientific Computing Website, Florida State University.
\item Campbell, J. Y., \& Cochrane, J. H. (1999). By Force of Habit: A Consumption-Based
Explanation of Aggregate Stock Market Behavior. The Journal of Political
Economy, 107(2), 205-251.
\item Cha, S. -H. (2007). Comprehensive survey on distance/similarity measures
between probability density functions. International Journal of Mathematical
Models and Methods in Applied Sciences, 4(1), 300-307.
\item Chaussé, P. (2010). Computing generalized method of moments and generalized
empirical likelihood with R. Journal of Statistical Software, 34(11),
1-35.
\item Chiani, M., Dardari, D., \& Simon, M. K. (2003). New exponential bounds
and approximations for the computation of error probability in fading
channels. Wireless Communications, IEEE Transactions on, 2(4), 840-845.
\item Chow, C. K., \& Liu, C. N. (1968). Approximating discrete probability
distributions with dependence trees. Information Theory, IEEE Transactions
on, 14(3), 462-467.
\item Clark, P. K. (1973). A subordinated stochastic process model with
finite variance for speculative prices. Econometrica: journal of the
Econometric Society, 135-155.
\item Cochrane, J. H. (2009). Asset Pricing:(Revised Edition). Princeton
university press.
\item Cody, W. J. (1969). Rational Chebyshev approximations for the error
function. Mathematics of Computation, 23(107), 631-637.
\item Comaniciu, D., Ramesh, V., \& Meer, P. (2003). Kernel-based object
tracking. Pattern Analysis and Machine Intelligence, IEEE Transactions
on, 25(5), 564-577.
\item Constantinides, G. M., \& Duffie, D. (1996). Asset pricing with heterogeneous
consumers. Journal of Political economy, 219-240.
\item Contreras-Reyes, J. E., \& Arellano-Valle, R. B. (2012). Kullback–Leibler
divergence measure for multivariate skew-normal distributions. Entropy,
14(9), 1606-1626.
\item Craig, C. C. (1936). On the frequency function of xy. The Annals of
Mathematical Statistics, 7(1), 1-15.
\item Dasgupta, S., \& Gupta, A. (1999). An elementary proof of the Johnson-Lindenstrauss
lemma. International Computer Science Institute, Technical Report,
99-006.
\item Derpanis, K. G. (2008). The Bhattacharyya Measure. Mendeley Computer,
1(4), 1990-1992.
\item Djouadi, A., Snorrason, O. Ö., \& Garber, F. D. (1990). The quality
of training sample estimates of the Bhattacharyya coefficient. IEEE
Transactions on Pattern Analysis and Machine Intelligence, 12(1),
92-97.
\item Doornik, J. A., \& Hansen, H. (2008). An omnibus test for univariate
and multivariate normality. Oxford Bulletin of Economics and Statistics,
70(s1), 927-939.
\item Dordick, H. S., \& Wang, G. (1993). Information society: A retrospective
view. Sage Publications, Inc..
\item Duchi, J. (2007). Derivations for linear algebra and optimization.
Berkeley, California.
\item Engle, R. F. (1984). Wald, likelihood ratio, and Lagrange multiplier
tests in econometrics. Handbook of econometrics, 2, 775-826.
\item Epstein, B. (1948). Some applications of the Mellin transform in statistics.
The Annals of Mathematical Statistics, 370-379.
\end{doublespace}
\item Euler, L. (1953). Leonhard Euler and the Königsberg bridges. Scientific
American, 189(1), 66-72.
\begin{doublespace}
\item Fama, E. F. (1965). The behavior of stock-market prices. The journal
of Business, 38(1), 34-105.
\item Fama, E. F. (1995). Random walks in stock market prices. Financial
analysts journal, 51(1), 75-80.
\item Fodor, I. K. (2002). A survey of dimension reduction techniques.
\item Fowlkes, E. B. (1979). Some methods for studying the mixture of two
normal (lognormal) distributions. Journal of the American Statistical
Association, 74(367), 561-575.
\item Frankl, P., \& Maehara, H. (1988). The Johnson-Lindenstrauss lemma
and the sphericity of some graphs. Journal of Combinatorial Theory,
Series B, 44(3), 355-362.
\item Frankl, P., \& Maehara, H. (1990). Some geometric applications of
the beta distribution. Annals of the Institute of Statistical Mathematics,
42(3), 463-474.
\item Fuller, J. (2010). What is happening to news: The information explosion
and the crisis in journalism. University of Chicago Press.
\item Gentle, J. E. (2007). Matrix algebra: theory, computations, and applications
in statistics. Springer Science \& Business Media.
\item Gentle, J. E. (2012). Numerical linear algebra for applications in
statistics. Springer Science \& Business Media.
\item Genz, A. (1992). Numerical computation of multivariate normal probabilities.
Journal of computational and graphical statistics, 1(2), 141-149.
\item Genz, A., \& Bretz, F. (2009). Computation of multivariate normal
and t probabilities (Vol. 195). Springer Science \& Business Media.
\item Glen, A. G., Leemis, L. M., \& Drew, J. H. (2004). Computing the distribution
of the product of two continuous random variables. Computational statistics
\& data analysis, 44(3), 451-464.
\end{doublespace}
\item Gribkovskaia, I., Halskau Sr, Ø., \& Laporte, G. (2007). The bridges
of Königsberg—a historical perspective. Networks: An International
Journal, 49(3), 199-203.
\begin{doublespace}
\item Guorong, X., Peiqi, C., \& Minhui, W. (1996, August). Bhattacharyya
distance feature selection. In Pattern Recognition, 1996., Proceedings
of the 13th International Conference on (Vol. 2, pp. 195-199). IEEE.
\end{doublespace}
\item Hamilton, J. D. (1994). Time series analysis (Vol. 2, pp. 690-696).
Princeton, NJ: Princeton university press.
\begin{doublespace}
\item Hansen, L. P., \& Jagannathan, R. (1991). Implications of Security
Market Data for Models of Dynamic Economies. The Journal of Political
Economy, 99(2), 225-262.
\item Hansen, L. P., \& Richard, S. F. (1987). The Role of Conditioning
Information in Deducing Testable Restrictions Implied by Dynamic Asset
Pricing Models. Econometrica, 55(3), 587-613.
\item Hanson, D. L., \& Wright, F. T. (1971). A bound on tail probabilities
for quadratic forms in independent random variables. The Annals of
Mathematical Statistics, 42(3), 1079-1083.
\item Hellinger, E. (1909). Neaue begr undung der theorie der quadratischen
formen von unendlichen vielen ver anderlichen. 136, 210–271.
\item Henderson, H. V., \& Searle, S. R. (1981). On deriving the inverse
of a sum of matrices. Siam Review, 23(1), 53-60.
\item Horrace, W. C. (2005). Some results on the multivariate truncated
normal distribution. Journal of Multivariate Analysis, 94(1), 209-221.
\item Hothorn, T., Bretz, F., \& Genz, A. (2001). On multivariate t and
Gauss probabilities in R. sigma, 1000, 3.
\item Hull, J. C. (2006). Options, futures, and other derivatives. Pearson
Education India.
\item Huth, E. J. (1989). The information explosion. Bulletin of the New
York Academy of Medicine, 65(6), 647.
\item Huzurbazar, V. S. (1955). Exact forms of some invariants for distributions
admitting sufficient statistics. Biometrika, 42(3/4), 533-537.
\item Jain, A. K. (1976). Estimate of Bhattacharyya Distance. IEEE Transactions
on Systems Man and Cybernetics, 6(11), 763-766.
\item Jarque, C. M., \& Bera, A. K. (1980). Efficient tests for normality,
homoscedasticity and serial independence of regression residuals.
Economics letters, 6(3), 255-259.
\item Jarque, C. M., \& Bera, A. K. (1987). A test for normality of observations
and regression residuals. International statistical review, 55(2),
163-172.
\item Johnson, W. B., \& Lindenstrauss, J. (1984). Extensions of Lipschitz
mappings into a Hilbert space. Contemporary mathematics, 26(189-206),
1.
\item Julier, S. J., \& Uhlmann, J. K. (1996). A general method for approximating
nonlinear transformations of probability distributions. Technical
report, Robotics Research Group, Department of Engineering Science,
University of Oxford.
\item Kailath, T. (1967). The divergence and Bhattacharyya distance measures
in signal selection. Communication Technology, IEEE Transactions on,
15(1), 52-60.
\item Kalman, R. E. (1960). A new approach to linear filtering and prediction
problems. Journal of basic Engineering, 82(1), 35-45.
\item Kashyap, R. (2015). Financial Services, Economic Growth and Well-Being:
A Four Pronged Study. Indian Journal of Finance, 9(1), 9-22.
\item Kashyap, R. (2016). Fighting Uncertainty with Uncertainty. Working
Paper.
\item Kashyap, R. (2017). Notes on Uncertainty, Unintended Consequences
and Everything Else. Working Paper.
\end{doublespace}
\item Kashyap, R. (2018). Solving the Equity Risk Premium Puzzle and Inching
Towards a Theory of Everything. Institutional Investor Journals, Journal
of Private Equity, 21(2), 45-63.
\begin{doublespace}
\item Kattumannil, S. K. (2009). On Stein’s identity and its applications.
Statistics \& Probability Letters, 79(12), 1444-1449.
\item Keefer, D. L., \& Bodily, S. E. (1983). Three-point approximations
for continuous random variables. Management Science, 29(5), 595-609.
\item Keynes, J. M. (1937). The General Theory of Employment. The Quarterly
Journal of Economics, 51(2), 209-223.
\item Keynes, J. M. (1971). The Collected Writings of John Maynard Keynes:
In 2 Volumes. A Treatise on Money. The Applied Theory of Money. Macmillan
for the Royal Economic Society. 
\item Keynes, J. M. (1973). A treatise on probability, the collected writings
of John Maynard Keynes, vol. VIII.
\item Kiani, M., Panaretos, J., Psarakis, S., \& Saleem, M. (2008). Approximations
to the normal distribution function and an extended table for the
mean range of the normal variables.
\item Kimeldorf, G., \& Sampson, A. (1973). A class of covariance inequalities.
Journal of the American Statistical Association, 68(341), 228-230.
\item Kon, S. J. (1984). Models of stock returns—a comparison. The Journal
of Finance, 39(1), 147-165.
\item Korth, H. F., \& Silberschatz, A. (1997). Database research faces
the information explosion. Communications of the ACM, 40(2), 139-142.
\item Kullback, S., \& Leibler, R. A. (1951). On information and sufficiency.
The annals of mathematical statistics, 22(1), 79-86.
\item Lawson, T. (1985). Uncertainty and economic analysis. The Economic
Journal, 95(380), 909-927.
\item Lee, K. Y., \& Bretschneider, T. R. (2012). Separability Measures
of Target Classes for Polarimetric Synthetic Aperture Radar Imagery.
Asian Journal of Geoinformatics, 12(2).
\item Liu, J. S. (1994). Siegel's formula via Stein's identities. Statistics
\& Probability Letters, 21(3), 247-251.
\item Losq, E., \& Chateau, J. P. D. (1982). A Generalization of the CAPM
Based on a Property of the Covariance Operator. Journal of Financial
and Quantitative Analysis, 17(05), 783-797.
\item Lucas Jr, R. E. (1978). Asset prices in an exchange economy. Econometrica:
Journal of the Econometric Society, 1429-1445.
\item Lyons, R. (2013). Distance covariance in metric spaces. The Annals
of Probability, 41(5), 3284-3305.
\item Mahalanobis, P. C. (1936). On the generalized distance in statistics.
Proceedings of the National Institute of Sciences (Calcutta), 2, 49-55.
\item Major, C. H., \& Savin-Baden, M. (2010). An introduction to qualitative
research synthesis: Managing the information explosion in social science
research. Routledge.
\item Malkovich, J. F., \& Afifi, A. A. (1973). On tests for multivariate
normality. Journal of the American statistical association, 68(341),
176-179.
\item Mak, B., \& Barnard, E. (1996, October). Phone clustering using the
Bhattacharyya distance. In Spoken Language, 1996. ICSLP 96. Proceedings.,
Fourth International Conference on (Vol. 4, pp. 2005-2008). IEEE.
\item Mandelbrot, B., \& Taylor, H. M. (1967). On the distribution of stock
price differences. Operations research, 15(6), 1057-1062.
\item Manjunath, B. G., \& Wilhelm, S. (2012). Moments Calculation For the
Doubly Truncated Multivariate Normal Density. arXiv preprint arXiv:1206.5387.
\item Mathai, A. M., \& Saxena, R. K. (1973). Meijer's G-function. Generalized
Hypergeometric Functions with Applications in Statistics and Physical
Sciences, 1-40.
\item Matusita, K. (1955). Decision rules based on distance for problems
of fit, two samples and estimation. Annals of Mathematical Statistics,
26, 631–641.
\item McManus, H., \& Hastings, D. (2005, July). 3.4. 1 A Framework for
Understanding Uncertainty and its Mitigation and Exploitation in Complex
Systems. In INCOSE International Symposium (Vol. 15, No. 1, pp. 484-503).
\item Mehra, R., \& Prescott, E. C. (1985). The equity premium: A puzzle.
Journal of monetary Economics, 15(2), 145-161.
\item Mehra, R. (1988). On the existence and representation of equilibrium
in an economy with growth and non-stationary consumption. International
Economic Review, 131-135.
\item Mehra, R., \& Prescott, E. C. (1988). The equity risk premium: A solution?.
Journal of Monetary Economics, 22(1), 133-136.
\end{doublespace}
\item Meijer, C. S. (1936). Uber Whittakersche bezw. Besselsche funktionen
und deren produkte. Nieuw Archief voor Wiskunde, 18(2), 10-29.
\begin{doublespace}
\item Miller III, A. C., \& Rice, T. R. (1983). Discrete approximations
of probability distributions. Management science, 29(3), 352-362.
\item Miranda, M. J., \& Fackler, P. L. (2002). Applied Computational Economics
and Finance.
\item Nelson, J. (2010). Johnson-Lindenstrauss notes. Technical report,
MIT-CSAIL, Available at web. mit. edu/minilek/www/jl\_notes. pdf.
\item Osborne, M. F. (1959). Brownian motion in the stock market. Operations
research, 7(2), 145-173.
\item Osborne, M. F. M. (1962). Periodic structure in the Brownian motion
of stock prices. Operations Research, 10(3), 345-379.
\item Richardson, M., \& Smith, T. (1993). A test for multivariate normality
in stock returns. Journal of Business, 295-321.
\end{doublespace}
\item Rietz, T. A. (1988). The equity risk premium a solution. Journal of
monetary Economics, 22(1), 117-131.
\begin{doublespace}
\item Rubinstein, M. E. (1973). A comparative statics analysis of risk premiums.
The Journal of Business, 46(4), 605-615.
\item Rubinstein, M. (1976). The valuation of uncertain income streams and
the pricing of options. The Bell Journal of Economics, 407-425.
\item Sagan, C. (2006). Cosmos (Vol. 1). Edicions Universitat Barcelona.
\item Schweppe, F. C. (1967a). State space evaluation of the Bhattacharyya
distance between two Gaussian processes. Information and Control,
11(3), 352-372.
\item Schweppe, F. C. (1967b). On the Bhattacharyya distance and the divergence
between Gaussian processes. Information and Control, 11(4), 373-395.
\item Seijas-Macías, A., \& Oliveira, A. (2012). An approach to distribution
of the product of two normal variables. Discussiones Mathematicae
Probability and Statistics, 32(1-2), 87-99.
\item Siegel, A. F. (1993). A surprising covariance involving the minimum
of multivariate normal variables. Journal of the American Statistical
Association, 88(421), 77-80.
\item Simon, H. A. (1962). The Architecture of Complexity. Proceedings of
the American Philosophical Society, 106(6), 467-482.
\item Smith, J. E. (1993). Moment methods for decision analysis. Management
science, 39(3), 340-358.
\item Soranzo, A., \& Epure, E. (2014). Very simply explicitly invertible
approximations of normal cumulative and normal quantile function.
Applied Mathematical Sciences, 8(87), 4323-4341.
\item Sorzano, C. O. S., Vargas, J., \& Montano, A. P. (2014). A survey
of dimensionality reduction techniques. arXiv preprint arXiv:1403.2877.
\item Springer, M. D., \& Thompson, W. E. (1966). The distribution of products
of independent random variables. SIAM Journal on Applied Mathematics,
14(3), 511-526.
\item Springer, M. D., \& Thompson, W. E. (1970). The distribution of products
of beta, gamma and Gaussian random variables. SIAM Journal on Applied
Mathematics, 18(4), 721-737.
\item Stein, C. M. (1973). Estimation of the mean of a multivariate normal
distribution. Proceedings of the Prague Symposium of Asymptotic Statistics.
\item Stein, C. M. (1981). Estimation of the mean of a multivariate normal
distribution. The annals of Statistics, 1135-1151.
\item Sweeney, L. (2001). Information explosion. Confidentiality, disclosure,
and data access: Theory and practical applications for statistical
agencies, 43-74.
\item Székely, G. J., \& Rizzo, M. L. (2005). A new test for multivariate
normality. Journal of Multivariate Analysis, 93(1), 58-80.
\item Székely, G. J., Rizzo, M. L., \& Bakirov, N. K. (2007). Measuring
and testing dependence by correlation of distances. The Annals of
Statistics, 35(6), 2769-2794.
\item Székely, G. J., \& Rizzo, M. L. (2009). Brownian distance covariance.
The annals of applied statistics, 3(4), 1236-1265.
\item Székely, G. J., \& Rizzo, M. L. (2013). Energy statistics: A class
of statistics based on distances. Journal of statistical planning
and inference, 143(8), 1249-1272.
\item Tauchen, G. E., \& Pitts, M. (1983). The Price Variability-Volume
Relationship on Speculative Markets. Econometrica, 51(2), 485-505.
\item Teerapabolarn, K. (2013). Stein's identity for discrete distributions.
International Journal of Pure and Applied Mathematics, 83(4), 565.
\item Thorne, K. (2014). The science of Interstellar. WW Norton \& Company.
\item Venkatasubramanian, S., \& Wang, Q. (2011). The Johnson-Lindenstrauss
transform: an empirical study. In Proceedings of the Meeting on Algorithm
Engineering \& Experiments (pp. 164-173). Society for Industrial and
Applied Mathematics.
\item Vernic, R., Teodorescu, S., \& Pelican, E. (2009). Two lognormal models
for real data. Annals of Ovidius University, Series Mathematics, 17(3),
263-277.
\item Wald, A., \& Wolfowitz, J. (1946). Tolerance limits for a normal distribution.
The Annals of Mathematical Statistics, 208-215.
\item Wallace, D. L. (1958). Asymptotic approximations to distributions.
The Annals of Mathematical Statistics, 29(3), 635-654.
\item Ware, R., \& Lad, F. (2003). Approximating the distribution for sums
of products of normal variables. University of Canterbury, England,
Tech. Rep. UCDMS, 15, 2003.
\end{doublespace}
\item Wasserman, S., \& Faust, K. (1994). Social network analysis: Methods
and applications (Vol. 8). Cambridge university press.
\item Watts, D. J. (1999). Networks, dynamics, and the small-world phenomenon.
American Journal of sociology, 105(2), 493-527.
\begin{doublespace}
\item Wei, K. J., \& Lee, C. F. (1988). The generalized Stein/Rubinstein
covariance formula and its application to estimate real systematic
risk. Management science, 34(10), 1266-1270.
\item Weil, P. (1989). The equity premium puzzle and the risk-free rate
puzzle. Journal of Monetary Economics, 24(3), 401-421.
\item Weisstein, E. W. (2004). Newton-cotes formulas.
\item Weitzman, M. (2007). Subjective Expectations and Asset-Return Puzzles.
American Economic Review, 97(4), 1102-1130.
\item Wilhelm, S., \& Manjunath, B. G. (2010). tmvtnorm: A package for the
truncated multivariate normal distribution. sigma, 2, 2.
\item Wilhelm, S. (2015). Package ‘tmvtnorm’.
\item Yang, M. (2008). Normal log-normal mixture, leptokurtosis and skewness.
Applied Economics Letters, 15(9), 737-742.
\item Zogheib, B., \& Hlynka, M. (2009). Approximations of the Standard
Normal Distribution. University of Windsor, Department of Mathematics
and Statistics.
\end{doublespace}
\end{enumerate}

\section*{Appendix of Mathematical Proofs and Supplementary Material}

\section{\label{sec:Appendix:-Example-Four-Physical-Dimensions}Appendix:
Example for the Merits and Limits of Four Physical Dimensions}

\begin{doublespace}
Now, imagine an object in our world, say a soda can, or, a laptop.
Observing these articles, in three dimensions, can tell us a great
deal about these items, and help us understand their nature and properties.
But now, suppose that we put light, on these objects, from different
directions, and observe the shadow in a lower dimension, which is
a flat surface. Here, we are restricted to two dimensions and perhaps
we are also losing the color of the objects. But say, we could view
the shadow, with spectacles having different colored glasses (we could,
of course, drop this additional constraint and deal only with a grey
shadow, without any loss of intuition). We could capture, a great
deal of the properties of the three dimensional object, but we are
restricted by the characteristics, that could be captured in a lower
dimension. Another point to bear in mind, is that the shadow only
projects certain points onto the flat surface, as dictated by the
properties of physics, more specifically, perhaps, the laws of light;
whereas with the JL Lemma transformation, every point in the higher
dimension is mapped to a point in the lower dimension.

When we do the JL Lemma transformation, by multiplying the original
set of observations we have, (a probability distribution) by another
probability distribution; or, a conversion from a higher dimensional
object, to a lower dimension, can be compared to putting light, on
the object in the three dimensional world, from different angles and
capturing the different shadows. The shadow changes in the two dimensional
world depending on the angle from which the light is coming (or, which
particular probability distribution, is being used for the transformation),
but the original object, lives on, in the higher world, with no change
whatsoever. There is definitely some loss, but that loss is the reality
or the limitation of the dimension we are forced to work with. If
we have one object in a lower dimension and that is the best information,
we have about it; we cannot take it to the higher dimension, to compare
it with an object, in the higher dimension. The best we can do, is
bring the object from the higher dimension, to the lower dimension
and compare them, in the lower dimension.

From our example, this is essentially, comparing a shadow that we
have, to the shadows we can produce, by bringing objects from the
higher dimension to the lower dimension. To illustrate this further,
imagine that we have just observed the shadow of a thief, and that
is the best information we have captured (say from a camera footage).
The best way to go about tracking this thief, is trying to understand
,what sort of people can produce such shadows, and this might reduce
the range of possibilities, as we try to look for our burglar. For
someone, that would say that it is not too much information, surely
that is more data points, than having nothing and we should be thankful
and look into ways in which, we can use, the only lead we have (in
this case), in tracking down the shadowy suspect.
\end{doublespace}

\section{\label{sec:Appendix:-Uncertainty-and}Appendix: Uncertainty and Unintended
Consequences}

\begin{doublespace}
(Lawson 1985) argues that the Keynesian view on uncertainty (that
it is generally impossible, even in probabilistic terms, to evaluate
the future outcomes of all possible current actions; Keynes 1937;
1971; 1973), far from being innocuous or destructive of economic analysis
in general, can give rise to research programs incorporating, amongst
other things, a view of rational behavior under uncertainty, which
could be potentially fruitful. (McManus \& Hastings 2005) clarify
the wide range of uncertainties that affect complex engineering systems
and present a framework to understand the risks (and opportunities)
they create and the strategies system designers can use to mitigate
or take advantage of them. (Simon 1962) points out that any attempt
to seek properties common to many sorts of complex systems (physical,
biological or social), would lead to a theory of hierarchy since a
large proportion of complex systems observed in nature exhibit hierarchic
structure; that a complex system is composed of subsystems that, in
turn, have their own subsystems, and so on. These viewpoints hold
many lessons for policy designers in the social sciences and could
be instructive for researchers looking to create metrics to compare
complex systems, keeping in mind the caveats of dynamic social systems.

A hall mark of the social sciences is the lack of objectivity. Here
we assert that objectivity is with respect to comparisons done by
different participants and that a comparison is a precursor to a decision. 
\end{doublespace}
\begin{conjecture}
\begin{doublespace}
\textbf{Despite the several advances in the social sciences,} \textbf{we
have yet to discover an objective measuring stick for comparison,
a so called, True Comparison Theory, which can be an aid for arriving
at objective decisions}. 
\end{doublespace}
\end{conjecture}
\begin{doublespace}
The search for such a theory could again (Kashyap 2017) be compared,
to the medieval alchemists’ obsession with turning everything into
gold. For our present purposes, the lack of such an objective measure
means that the difference in comparisons, as assessed by different
participants, can effect different decisions under the same set of
circumstances. Hence, despite all the uncertainty in the social sciences,
the one thing we can be almost certain about is the subjectivity in
all decision making. Restricted to the particular sub-universe of
economic and financial theory, this translates to the lack of an objective
measuring stick of value, a so called, True Value Theory. This lack
of an objective measure of value, (hereafter, value will be synonymously
referred to as the price of an instrument), leads to dissimilar decisions
and actions by different participants, making prices react at varying
degrees and at varying speeds to the pull of diverse macro and micro
factors.
\end{doublespace}
\begin{doublespace}

\section{\label{sec:Appendix-of-Proofs}Appendix: Mathematical Proofs}
\end{doublespace}
\begin{doublespace}

\subsection{\label{subsec:Notation-and-Terminology}Notation and Terminology
for Key Results}
\end{doublespace}
\begin{itemize}
\begin{doublespace}
\item $D_{BC}\left(p_{i},p_{i}^{\prime}\right)$, the Bhattacharyya Distance
between two multinomial populations each consisting of $k$ categories
classes with associated probabilities $p_{1},p_{2},...,p_{k}$ and
$p_{1}^{\prime},p_{2}^{\prime},...,p_{k}^{\prime}$ respectively.
\item $\rho\left(p_{i},p_{i}^{\prime}\right)$, the Bhattacharyya Coefficient.
\item $d\left(p_{i},p_{i}^{\prime}\right)$, the modified Bhattacharyya
Metric.
\item $\chi^{2}\left(p_{i},p_{i}^{\prime}\right)$ Chi-Squared measure.
\item $D_{H-M}\left(p_{i},p_{i}^{\prime}\right)$ is the Hellinger or Matusita
Distance.
\item $D_{BC-N}(p,q)$ is the Bhattacharyya distance between $p$ and $q$
normal distributions or classes.
\item $D_{BC-MN}\left(p_{1},p_{2}\right)$ is the Bhattacharyya distance
between two multivariate normal distributions, $\boldsymbol{p_{1}},\boldsymbol{p_{2}}$
where $\boldsymbol{p_{i}}\sim\mathcal{N}(\boldsymbol{\mu}_{i},\,\boldsymbol{\Sigma}_{i})$.
\item $D_{BC-TN}(p,q)$ is the Bhattacharyya distance between $p$ and $q$
truncated normal distributions or classes.
\item $D_{BC-TMN}\left(p_{1},p_{2}\right)$ is the Bhattacharyya distance
between two truncated multivariate normal distributions, $\boldsymbol{p_{1}},\boldsymbol{p_{2}}$
where $\boldsymbol{p_{i}}\sim\mathcal{N}(\boldsymbol{\mu}_{i},\,\boldsymbol{\Sigma}_{i},\,\boldsymbol{a}_{i},\,\boldsymbol{b}_{i})$.
\item $D_{\mathrm{KL}}(P\|Q)$ is the Kullback-Leibler divergence of $Q$
from $P$.
\end{doublespace}
\end{itemize}
\begin{doublespace}

\subsection{\label{subsec:Proof-of-Proposition: Normal_Lognormal}Proof of Proposition
\ref{prop:Normal_Lognormal_The-density-function}}
\end{doublespace}
\begin{proof}
\begin{doublespace}
From the properties of a bivariate distribution (Bertsekas 2002) of
two normal random variables, $V,W$, the conditional expectation is
given by
\[
E\left[V\mid W\right]=E\left[V\right]+\frac{cov\left(V,W\right)}{\sigma_{W}^{2}}\left(W-E\left[W\right]\right)
\]
This is a linear function of $W$ and is a normal distribution. The
conditional distribution of $V$ given $W$ is normal with mean $E[V\mid W]$
and variance, $\sigma_{V\mid W}^{2}$ using the correlation coefficient
$\rho_{VW}$ between them is,
\[
Var\left(V\mid W\right)=\sigma_{V\mid W}^{2}=(1-\rho_{VW}^{2})\sigma_{V}^{2}
\]
This is easily established by setting, 
\[
\hat{V}=\rho_{VW}\frac{\sigma_{V}}{\sigma_{W}}\left(W-\mu_{W}\right)
\]
\[
\tilde{V}=V-\hat{V}
\]
\[
E\left[\tilde{V}\right]=E\left[V\right]-E\left[\rho_{VW}\frac{\sigma_{V}}{\sigma_{W}}\left(W-\mu_{W}\right)\right]=E\left[V\right]
\]
$W\text{ and }\tilde{V}$ are independent as are $\hat{V}\text{ and }\tilde{V}$,
since $\hat{V}$ is a scalar multiple of $W\text{ and }E\left[W\tilde{V}\right]=\mu_{W}\mu_{V}$
implying that $Cov\left(W,\tilde{V}\right)=E\left[W\tilde{V}\right]-\mu_{W}\mu_{V}=0$
as shown below,
\[
E\left[W\tilde{V}\right]=E\left[WV\right]-\rho_{VW}\frac{\sigma_{V}}{\sigma_{W}}E\left[W^{2}\right]+\rho_{VW}\frac{\sigma_{V}}{\sigma_{W}}E\left[W\right]\mu_{W}
\]
\[
=\rho_{VW}\sigma_{V}\sigma_{W}+\mu_{W}\mu_{V}-\rho_{VW}\frac{\sigma_{V}}{\sigma_{W}}\left\{ \sigma_{W}^{2}+\mu_{W}^{2}\right\} +\rho_{VW}\frac{\sigma_{V}}{\sigma_{W}}\mu_{W}^{2}=\mu_{W}\mu_{V}
\]
\[
E\left[V\mid W\right]=E\left[\tilde{V}+\hat{V}\mid W\right]=E\left[\tilde{V}\mid W\right]+E\left[\hat{V}\mid W\right]
\]
\[
=\mu_{V}+\rho_{VW}\frac{\sigma_{V}}{\sigma_{W}}\left(W-\mu_{W}\right)
\]
\[
Var\left(V\mid W\right)=Var\left(\hat{V}+\tilde{V}\mid W\right)
\]
\[
=Var\left(\tilde{V}\mid W\right)=Var\left(\tilde{V}\right)=Var\left(V-\rho_{VW}\frac{\sigma_{V}}{\sigma_{W}}W\right)
\]
\[
=Var\left(V\right)+\rho_{VW}^{2}\frac{\sigma_{V}^{2}}{\sigma_{W}^{2}}Var\left(W\right)-2\rho_{VW}\frac{\sigma_{V}}{\sigma_{W}}Cov\left(V,W\right)
\]
\[
=\sigma_{V}^{2}+\rho_{VW}^{2}\sigma_{V}^{2}-2\rho_{VW}^{2}\sigma_{V}^{2}
\]
\[
=\sigma_{V}^{2}\left(1-\rho_{VW}^{2}\right)
\]
Now consider the random variable, $U$, given by,
\[
U=Xe^{Y}
\]
Here, $X$ and $Y$ are random variables with correlation coefficient,
$\rho$ satisfying, 
\[
\left[\begin{array}{c}
X\\
Y
\end{array}\right]\sim N\left(\left[\begin{array}{c}
\mu_{X}\\
\mu_{Y}
\end{array}\right],\left[\begin{array}{cc}
\sigma_{X}^{2} & \rho\sigma_{X}\sigma_{Y}\\
\rho\sigma_{X}\sigma_{Y} & \sigma_{Y}^{2}
\end{array}\right]\right)
\]
The conditional distribution of $U$ given $Y$ is normal with mean,
variance and density given by,
\[
E\left[U\mid Y\right]=E\left[Xe^{Y}\mid Y\right]=e^{Y}E\left[X\mid Y\right]
\]
\[
=e^{Y}\left\{ \mu_{X}+\frac{cov\left(X,Y\right)}{\sigma_{Y}^{2}}\left(Y-\mu_{Y}\right)\right\} 
\]
\[
=e^{Y}\left\{ \mu_{X}+\frac{\rho\sigma_{X}\sigma_{Y}}{\sigma_{Y}^{2}}\left(Y-\mu_{Y}\right)\right\} 
\]
\[
=e^{Y}\left\{ \mu_{X}+\frac{\rho\sigma_{X}}{\sigma_{Y}}\left(Y-\mu_{Y}\right)\right\} 
\]
\[
\sigma_{U\mid Y}^{2}=Var\left(U\mid Y\right)=Var\left(Xe^{Y}\mid Y\right)=e^{2Y}Var\left(X\mid Y\right)
\]
\[
=e^{2Y}(1-\rho^{2})\sigma_{X}^{2}
\]
\[
f_{U\mid Y}\left(u\mid Y\right)\sim N\left(e^{Y}\left\{ \mu_{X}+\frac{\rho\sigma_{X}}{\sigma_{Y}}\left(Y-\mu_{Y}\right)\right\} ,e^{2Y}(1-\rho^{2})\sigma_{X}^{2}\right)
\]
We observe that this conditional distribution is undefined for $\sigma_{Y}=0$
when $Y$ degenerates to a constant. We then get the joint density
of $U,Y$ as, 
\[
f_{UY}\left(u,y\right)=f_{U\mid Y}\left(u\mid Y\right)\;\times\;f_{Y}\left(y\right)
\]
\[
=\left\{ 2\pi e^{2y}(1-\rho^{2})\sigma_{X}^{2}\right\} ^{-\frac{1}{2}}\;e^{-\frac{\left[u-e^{y}\left\{ \mu_{X}+\frac{\rho\sigma_{X}}{\sigma_{Y}}\left(y-\mu_{Y}\right)\right\} \right]^{2}}{2e^{2y}(1-\rho^{2})\sigma_{X}^{2}}}\;\times\;\left\{ 2\pi\sigma_{Y}^{2}\right\} ^{-\frac{1}{2}}\;e^{-\frac{\left[y-\mu_{Y}\right]^{2}}{2\sigma_{Y}^{2}}}
\]
The marginal density of $U$ is given by, 
\[
f_{U}\left(u\right)=\int_{-\infty}^{\infty}\left\{ 2\pi e^{2y}(1-\rho^{2})\sigma_{X}^{2}\right\} ^{-\frac{1}{2}}\;e^{-\frac{\left[u-e^{y}\left\{ \mu_{X}+\frac{\rho\sigma_{X}}{\sigma_{Y}}\left(y-\mu_{Y}\right)\right\} \right]^{2}}{2e^{2y}(1-\rho^{2})\sigma_{X}^{2}}}\;\times\;\left\{ 2\pi\sigma_{Y}^{2}\right\} ^{-\frac{1}{2}}\;e^{-\frac{\left[y-\mu_{Y}\right]^{2}}{2\sigma_{Y}^{2}}}\;dy
\]
In our case, we have, $\mu_{X}=0,\sigma_{X}^{2}=1/k\text{ and }\rho=0$
giving, 
\[
\left[\begin{array}{c}
X\\
Y
\end{array}\right]\sim N\left(\left[\begin{array}{c}
0\\
\mu_{Y}
\end{array}\right],\left[\begin{array}{cc}
\frac{1}{k} & 0\\
0 & \sigma_{Y}^{2}
\end{array}\right]\right)
\]
\[
f_{U}\left(u\right)=\int_{-\infty}^{\infty}\left\{ \frac{2\pi e^{2y}}{k}\right\} ^{-\frac{1}{2}}\;e^{-\frac{ku^{2}}{2e^{2y}}}\;\times\;\left\{ 2\pi\sigma_{Y}^{2}\right\} ^{-\frac{1}{2}}\;e^{-\frac{\left[y-\mu_{Y}\right]^{2}}{2\sigma_{Y}^{2}}}\;dy
\]
\[
f_{U}\left(u\right)=\frac{\sqrt{k}}{2\pi\sigma_{Y}}\int_{-\infty}^{\infty}\;e^{-y-\frac{ku^{2}}{2e^{2y}}-\frac{\left[y-\mu_{Y}\right]^{2}}{2\sigma_{Y}^{2}}}dy
\]
\end{doublespace}
\end{proof}
\begin{doublespace}

\subsection{\label{subsec:Proof-of-Proposition: Normal_Normal}Proof of Proposition
\ref{prop:Normal_Normal_The-density-function}}
\end{doublespace}
\begin{proof}
\begin{doublespace}
First we establish the density of the product of two independent random
variables. Let $W=XY$, a continuous random variable, product of two
independent continuous random variables $X$ and $Y$. The distribution
function of $W$ is, 
\[
F_{W}\left(w\right)=\int_{\left\{ \left(x,y\right)|xy\leq w\right\} }f_{XY}\left(x,y\right)dxdy
\]
Here, $\left\{ \left(x,y\right)|xy\leq w\right\} =\left\{ -\infty<x\leq0,\frac{w}{x}\leq y<\infty\right\} \cup\left\{ 0\leq x\leq\infty,-\infty<y\leq\frac{w}{x}\right\} $.
We can write the above as,
\[
F_{W}\left(w\right)=\int_{-\infty}^{0}\int_{\frac{w}{x}}^{\infty}f_{XY}\left(x,y\right)dydx+\int_{0}^{\infty}\int_{-\infty}^{\frac{w}{x}}f_{XY}\left(x,y\right)dydx
\]
Differentiating with respect to, $w$, using the Leibniz integral
rule gives the required density,
\[
f_{W}\left(w\right)=\int_{-\infty}^{0}\left(-\frac{1}{x}\right)f_{XY}\left(x,\frac{w}{x}\right)dx+\int_{0}^{\infty}\left(\frac{1}{x}\right)f_{XY}\left(x,\frac{w}{x}\right)dx
\]
\[
f_{W}\left(w\right)=\int_{-\infty}^{\infty}\left(\frac{1}{\left|x\right|}\right)f_{XY}\left(x,\frac{w}{x}\right)dx=\int_{-\infty}^{\infty}\left(\frac{1}{\left|x\right|}\right)f_{X}\left(x\right)f_{Y}\left(\frac{w}{x}\right)dx
\]
In our case, the density of the two independent normal variables,
where, $\mu_{X}$ is the mean and $\sigma_{X}^{2}$ is the variance
of the $X$ distribution, is given by,
\[
f_{X}\left(x\;|\;\mu_{X},\sigma_{X}^{2}\right)=\frac{1}{\sigma_{X}\sqrt{2\pi}}\;e^{-\frac{\left(x-\mu_{X}\right)^{2}}{2\sigma_{X}^{2}}}\;;\quad f_{Y}\left(y\;|\;0,\frac{1}{k}\right)=\sqrt{\frac{k}{2\pi}}\;e^{-\frac{k\left(y\right)^{2}}{2}}
\]
 
\[
f_{W}\left(w\right)=\int_{-\infty}^{\infty}\left(\frac{1}{\left|x\right|}\right)\frac{1}{\sigma_{X}\sqrt{2\pi}}\;e^{-\frac{\left(x-\mu_{X}\right)^{2}}{2\sigma_{X}^{2}}}\sqrt{\frac{k}{2\pi}}\;e^{-\frac{k\left(\frac{w}{x}\right)^{2}}{2}}dx
\]

Leibniz Integral Rule: Let $f(x,\theta)$ be a function such that
$f_{\theta}(x,\theta)$ exists, and is continuous. Then, 
\[
\frac{\mathrm{d}}{\mathrm{d}\theta}\left(\int_{a(\theta)}^{b(\theta)}f(x,\theta)\,\mathrm{d}x\right)=\int_{a(\theta)}^{b(\theta)}\partial_{\theta}f(x,\theta)\,\mathrm{d}x\,+\,f\big(b(\theta),\theta\big)\cdot b'(\theta)\,-\,f\big(a(\theta),\theta\big)\cdot a'(\theta)
\]
where the partial derivative of $f$ indicates that inside the integral
only the variation of $f(x,)$ with $\theta$ is considered in taking
the derivative. Extending the above to the case of double integrals
gives,

\begin{eqnarray*}
{\frac{\mathrm{d}}{\mathrm{d}t}}\left(\int_{c(t)}^{d(t)}\int_{a(t)}^{b(t)}f(x,y)\,\mathrm{d}x\,\mathrm{d}y\right) & = & \int_{c(t)}^{d(t)}\left\{ \,f{\big(}b(t),y{\big)}\cdot b'(t)\,-\,f{\big(}a(t),y{\big)}\cdot a'(t)\right\} \,\mathrm{d}y\,\\
 &  & +\int_{a(t)}^{b(t)}\left\{ \,f(d(t),x)\cdot d'(t)-\,f(c(t),x)\cdot c'(t)\right\} \,\mathrm{d}x\,
\end{eqnarray*}
\end{doublespace}
\end{proof}
\begin{doublespace}

\subsection{\label{subsec:Proof-of-Proposition: Truncated Normal}Proof of Proposition
\ref{prop:Truncated_N_The-Bhattacharyya-coefficient}}
\end{doublespace}
\begin{proof}
\begin{doublespace}
Suppose we have two truncated normal distributions $p,q$ with density
functions. 
\[
f_{p}\left(x\mid\mu_{p},\sigma_{p}^{2},a,b\right)=\begin{cases}
\frac{\frac{1}{\sigma_{p}}\phi\left(\frac{x-\mu_{p}}{\sigma_{p}}\right)}{\Phi\left(\frac{b-\mu_{p}}{\sigma_{p}}\right)-\Phi\left(\frac{a-\mu_{p}}{\sigma_{p}}\right)} & \quad;a\leq x\leq b\\
0 & ;\text{otherwise}
\end{cases}
\]
\[
f_{q}\left(x\mid\mu_{q},\sigma_{q}^{2},c,d\right)=\begin{cases}
\frac{\frac{1}{\sigma_{q}}\phi\left(\frac{x-\mu_{q}}{\sigma_{q}}\right)}{\Phi\left(\frac{d-\mu_{q}}{\sigma_{q}}\right)-\Phi\left(\frac{c-\mu_{q}}{\sigma_{q}}\right)} & \quad;c\leq x\leq d\\
0 & ;\text{otherwise}
\end{cases}
\]
Here, $\mu_{p}$ is the mean and $\sigma_{p}^{2}$ is the variance
of the $p$ distribution with bounds $a,b$. The Bhattacharyya coefficient
is given by the below when the distributions overlap and it is zero
otherwise,
\[
\rho\left(p,q\right)=\int_{l=\min\left(a,c\right)}^{u=\min\left(b,d\right)}\sqrt{f_{p}\left(x\mid\mu_{p},\sigma_{p}^{2},a,b\right)f_{q}\left(x\mid\mu_{q},\sigma_{q}^{2},c,d\right)}dx
\]
\[
=\int_{l}^{u}\sqrt{\frac{\frac{1}{\sigma_{p}}\phi\left(\frac{x-\mu_{p}}{\sigma_{p}}\right)}{\left[\Phi\left(\frac{b-\mu_{p}}{\sigma_{p}}\right)-\Phi\left(\frac{a-\mu_{p}}{\sigma_{p}}\right)\right]}\frac{\frac{1}{\sigma_{q}}\phi\left(\frac{x-\mu_{q}}{\sigma_{q}}\right)}{\left[\Phi\left(\frac{d-\mu_{q}}{\sigma_{q}}\right)-\Phi\left(\frac{c-\mu_{q}}{\sigma_{q}}\right)\right]}}dx
\]
\[
\int_{l}^{u}\sqrt{\phi\left(\frac{x-\mu_{p}}{\sigma_{p}}\right)\phi\left(\frac{x-\mu_{q}}{\sigma_{q}}\right)}dx=\frac{1}{\sqrt{2\pi}}\int_{l}^{u}\;\sqrt{\exp\left[-\frac{\left(x-\mu_{p}\right)^{2}}{2\sigma_{p}^{2}}\right]\exp\left[-\frac{\left(x-\mu_{q}\right)^{2}}{2\sigma_{q}^{2}}\right]}
\]
\[
=\frac{1}{\sqrt{2\pi}}\int_{l}^{u}\;\exp\left[-\frac{\left(\sigma_{p}^{2}+\sigma_{q}^{2}\right)}{4\sigma_{p}^{2}\sigma_{q}^{2}}\left\{ x^{2}-2x\frac{\left(\mu_{p}\sigma_{q}^{2}+\mu_{q}\sigma_{p}^{2}\right)}{\left(\sigma_{p}^{2}+\sigma_{q}^{2}\right)}+\frac{\left(\mu_{p}\sigma_{q}^{2}+\mu_{q}\sigma_{p}^{2}\right)^{2}}{\left(\sigma_{p}^{2}+\sigma_{q}^{2}\right)^{2}}-\frac{\left(\mu_{p}\sigma_{q}^{2}+\mu_{q}\sigma_{p}^{2}\right)^{2}}{\left(\sigma_{p}^{2}+\sigma_{q}^{2}\right)^{2}}+\frac{\left(\mu_{p}^{2}\sigma_{q}^{2}+\mu_{q}^{2}\sigma_{p}^{2}\right)}{\left(\sigma_{p}^{2}+\sigma_{q}^{2}\right)}\right\} \right]
\]
\[
=\frac{1}{\sqrt{2\pi}}\int_{l}^{u}\;\exp\left[-\frac{1}{2}\frac{\left(\sigma_{p}^{2}+\sigma_{q}^{2}\right)}{2\sigma_{p}^{2}\sigma_{q}^{2}}\left\{ x-\frac{\left(\mu_{p}\sigma_{q}^{2}+\mu_{q}\sigma_{p}^{2}\right)}{\left(\sigma_{p}^{2}+\sigma_{q}^{2}\right)}\right\} ^{2}\right]\exp\left[-\frac{1}{4}\left\{ \frac{\left(\mu_{p}-\mu_{q}\right)^{2}}{\sigma_{p}^{2}+\sigma_{q}^{2}}\right\} \right]
\]
\[
=\sqrt{\frac{2\sigma_{p}^{2}\sigma_{q}^{2}}{\left(\sigma_{p}^{2}+\sigma_{q}^{2}\right)}}\exp\left[-\frac{1}{4}\left\{ \frac{\left(\mu_{p}-\mu_{q}\right)^{2}}{\sigma_{p}^{2}+\sigma_{q}^{2}}\right\} \right]\;\left\{ \Phi\left[\frac{u-\frac{\left(\mu_{p}\sigma_{q}^{2}+\mu_{q}\sigma_{p}^{2}\right)}{\left(\sigma_{p}^{2}+\sigma_{q}^{2}\right)}}{\sqrt{\frac{2\sigma_{p}^{2}\sigma_{q}^{2}}{\left(\sigma_{p}^{2}+\sigma_{q}^{2}\right)}}}\right]-\Phi\left[\frac{l-\frac{\left(\mu_{p}\sigma_{q}^{2}+\mu_{q}\sigma_{p}^{2}\right)}{\left(\sigma_{p}^{2}+\sigma_{q}^{2}\right)}}{\sqrt{\frac{2\sigma_{p}^{2}\sigma_{q}^{2}}{\left(\sigma_{p}^{2}+\sigma_{q}^{2}\right)}}}\right]\right\} 
\]
\[
D_{BC-TN}(p,q)=-\ln\left[\rho\left(p,q\right)\right]
\]
\begin{eqnarray*}
D_{BC-TN}(p,q) & = & \frac{1}{4}\left(\frac{(\mu_{p}-\mu_{q})^{2}}{\sigma_{p}^{2}+\sigma_{q}^{2}}\right)+\frac{1}{4}\ln\left(\frac{1}{4}\left(\frac{\sigma_{p}^{2}}{\sigma_{q}^{2}}+\frac{\sigma_{q}^{2}}{\sigma_{p}^{2}}+2\right)\right)\\
 &  & +\frac{1}{2}\ln\left[\Phi\left(\frac{b-\mu_{p}}{\sigma_{p}}\right)-\Phi\left(\frac{a-\mu_{p}}{\sigma_{p}}\right)\right]+\frac{1}{2}\ln\left[\Phi\left(\frac{d-\mu_{q}}{\sigma_{q}}\right)-\Phi\left(\frac{c-\mu_{q}}{\sigma_{q}}\right)\right]\\
 &  & -\ln\left\{ \Phi\left[\frac{u-\frac{\left(\mu_{p}\sigma_{q}^{2}+\mu_{q}\sigma_{p}^{2}\right)}{\left(\sigma_{p}^{2}+\sigma_{q}^{2}\right)}}{\sqrt{\frac{2\sigma_{p}^{2}\sigma_{q}^{2}}{\left(\sigma_{p}^{2}+\sigma_{q}^{2}\right)}}}\right]-\Phi\left[\frac{l-\frac{\left(\mu_{p}\sigma_{q}^{2}+\mu_{q}\sigma_{p}^{2}\right)}{\left(\sigma_{p}^{2}+\sigma_{q}^{2}\right)}}{\sqrt{\frac{2\sigma_{p}^{2}\sigma_{q}^{2}}{\left(\sigma_{p}^{2}+\sigma_{q}^{2}\right)}}}\right]\right\} 
\end{eqnarray*}
Let, $\nu=\frac{\left(\mu_{p}\sigma_{q}^{2}+\mu_{q}\sigma_{p}^{2}\right)}{\left(\sigma_{p}^{2}+\sigma_{q}^{2}\right)}$
and $\varsigma=\sqrt{\frac{2\sigma_{p}^{2}\sigma_{q}^{2}}{\left(\sigma_{p}^{2}+\sigma_{q}^{2}\right)}}$.
Looking at conditions when $D_{BC-TN}(p,q)\geq D_{BC-N}(p,q)$ gives,
\[
\frac{1}{2}\ln\left[\Phi\left(\frac{b-\mu_{p}}{\sigma_{p}}\right)-\Phi\left(\frac{a-\mu_{p}}{\sigma_{p}}\right)\right]+\frac{1}{2}\ln\left[\Phi\left(\frac{d-\mu_{q}}{\sigma_{q}}\right)-\Phi\left(\frac{c-\mu_{q}}{\sigma_{q}}\right)\right]\geq\ln\left\{ \Phi\left[\frac{u-\nu}{\varsigma}\right]-\Phi\left[\frac{l-\nu}{\varsigma}\right]\right\} 
\]
\[
\sqrt{\left[\Phi\left(\frac{b-\mu_{p}}{\sigma_{p}}\right)-\Phi\left(\frac{a-\mu_{p}}{\sigma_{p}}\right)\right]\left[\Phi\left(\frac{d-\mu_{q}}{\sigma_{q}}\right)-\Phi\left(\frac{c-\mu_{q}}{\sigma_{q}}\right)\right]}\geq\left\{ \Phi\left[\frac{u-\nu}{\varsigma}\right]-\Phi\left[\frac{l-\nu}{\varsigma}\right]\right\} 
\]
For completeness, let us also consider two univariate normal distributions,
$p,q$ with density function, 
\[
f_{p}\left(x\;|\;\mu_{p},\sigma_{p}^{2}\right)=\frac{1}{\sigma_{p}\sqrt{2\pi}}\;e^{-\frac{\left(x-\mu_{p}\right)^{2}}{2\sigma_{p}^{2}}}
\]
Here, $\mu_{p}$ is the mean and $\sigma_{p}^{2}$ is the variance
of the $p$ distribution. The Bhattacharyya coefficient is given by,
\[
\rho\left(p,q\right)=\int\sqrt{f_{p}\left(x\;|\;\mu_{p},\sigma_{p}^{2}\right)f_{q}\left(x\;|\;\mu_{q},\sigma_{q}^{2}\right)}dx
\]
\[
=\int\sqrt{\frac{1}{\sigma_{p}\sqrt{2\pi}}\;e^{-\frac{\left(x^{2}+\mu_{p}^{2}-2x\mu_{p}\right)}{2\sigma_{p}^{2}}}\frac{1}{\sigma_{q}\sqrt{2\pi}}\;e^{-\frac{\left(x^{2}+\mu_{q}^{2}-2x\mu_{q}\right)}{2\sigma_{q}^{2}}}}dx
\]
\[
=\int\sqrt{\frac{1}{\sigma_{p}\sqrt{2\pi}}\;\frac{1}{\sigma_{q}\sqrt{2\pi}}\;e^{-\left\{ \frac{x^{2}\left(\sigma_{p}^{2}+\sigma_{q}^{2}\right)-2x\left(\mu_{p}\sigma_{q}^{2}+\mu_{q}\sigma_{p}^{2}\right)+\left(\mu_{p}^{2}\sigma_{q}^{2}+\mu_{q}^{2}\sigma_{p}^{2}\right)}{2\sigma_{p}^{2}\sigma_{q}^{2}}\right\} }}dx
\]
\[
=\int\sqrt{\frac{1}{\sigma_{p}\sqrt{2\pi}}\;\frac{1}{\sigma_{q}\sqrt{2\pi}}\;e^{-\frac{\left(\sigma_{p}^{2}+\sigma_{q}^{2}\right)}{2\sigma_{p}^{2}\sigma_{q}^{2}}\left\{ x^{2}-2x\frac{\left(\mu_{p}\sigma_{q}^{2}+\mu_{q}\sigma_{p}^{2}\right)}{\left(\sigma_{p}^{2}+\sigma_{q}^{2}\right)}+\frac{\left(\mu_{p}^{2}\sigma_{q}^{2}+\mu_{q}^{2}\sigma_{p}^{2}\right)}{\left(\sigma_{p}^{2}+\sigma_{q}^{2}\right)}\right\} }}dx
\]
\[
=\frac{1}{\sqrt{\sigma_{p}\sigma_{q}}}\;\frac{1}{\sqrt{2\pi}}\int\;e^{-\frac{\left(\sigma_{p}^{2}+\sigma_{q}^{2}\right)}{4\sigma_{p}^{2}\sigma_{q}^{2}}\left\{ x^{2}-2x\frac{\left(\mu_{p}\sigma_{q}^{2}+\mu_{q}\sigma_{p}^{2}\right)}{\left(\sigma_{p}^{2}+\sigma_{q}^{2}\right)}+\frac{\left(\mu_{p}\sigma_{q}^{2}+\mu_{q}\sigma_{p}^{2}\right)^{2}}{\left(\sigma_{p}^{2}+\sigma_{q}^{2}\right)^{2}}-\frac{\left(\mu_{p}\sigma_{q}^{2}+\mu_{q}\sigma_{p}^{2}\right)^{2}}{\left(\sigma_{p}^{2}+\sigma_{q}^{2}\right)^{2}}+\frac{\left(\mu_{p}^{2}\sigma_{q}^{2}+\mu_{q}^{2}\sigma_{p}^{2}\right)}{\left(\sigma_{p}^{2}+\sigma_{q}^{2}\right)}\right\} }dx
\]
\[
=\frac{1}{\sqrt{\sigma_{p}\sigma_{q}}}\sqrt{\frac{2\sigma_{p}^{2}\sigma_{q}^{2}}{\left(\sigma_{p}^{2}+\sigma_{q}^{2}\right)}}\;\sqrt{\frac{\left(\sigma_{p}^{2}+\sigma_{q}^{2}\right)}{2\sigma_{p}^{2}\sigma_{q}^{2}}}\frac{1}{\sqrt{2\pi}}\int\;e^{-\frac{1}{2}\frac{\left(\sigma_{p}^{2}+\sigma_{q}^{2}\right)}{2\sigma_{p}^{2}\sigma_{q}^{2}}\left\{ x-\frac{\left(\mu_{p}\sigma_{q}^{2}+\mu_{q}\sigma_{p}^{2}\right)}{\left(\sigma_{p}^{2}+\sigma_{q}^{2}\right)}\right\} ^{2}-\frac{\left(\sigma_{p}^{2}+\sigma_{q}^{2}\right)}{4\sigma_{p}^{2}\sigma_{q}^{2}}\left\{ -\frac{\left(\mu_{p}\sigma_{q}^{2}+\mu_{q}\sigma_{p}^{2}\right)^{2}}{\left(\sigma_{p}^{2}+\sigma_{q}^{2}\right)^{2}}+\frac{\left(\mu_{p}^{2}\sigma_{q}^{2}+\mu_{q}^{2}\sigma_{p}^{2}\right)}{\left(\sigma_{p}^{2}+\sigma_{q}^{2}\right)}\right\} }dx
\]
\[
=\sqrt{\frac{2\sigma_{p}\sigma_{q}}{\left(\sigma_{p}^{2}+\sigma_{q}^{2}\right)}}\;e^{-\frac{1}{4}\frac{\left(\mu_{p}-\mu_{q}\right)^{2}}{\sigma_{p}^{2}+\sigma_{q}^{2}}}
\]
\[
\left[\because-\frac{\left(\sigma_{p}^{2}+\sigma_{q}^{2}\right)}{4\sigma_{p}^{2}\sigma_{q}^{2}}\left\{ -\frac{\left(\mu_{p}\sigma_{q}^{2}+\mu_{q}\sigma_{p}^{2}\right)^{2}}{\left(\sigma_{p}^{2}+\sigma_{q}^{2}\right)^{2}}+\frac{\left(\mu_{p}^{2}\sigma_{q}^{2}+\mu_{q}^{2}\sigma_{p}^{2}\right)}{\left(\sigma_{p}^{2}+\sigma_{q}^{2}\right)}\right\} \right.=\frac{1}{4\sigma_{p}^{2}\sigma_{q}^{2}}\left\{ \frac{\left(\mu_{p}\sigma_{q}^{2}+\mu_{q}\sigma_{p}^{2}\right)^{2}-\left(\sigma_{p}^{2}+\sigma_{q}^{2}\right)\left(\mu_{p}^{2}\sigma_{q}^{2}+\mu_{q}^{2}\sigma_{p}^{2}\right)}{\left(\sigma_{p}^{2}+\sigma_{q}^{2}\right)}\right\} 
\]
\[
=\frac{1}{4\sigma_{p}^{2}\sigma_{q}^{2}}\left\{ \frac{\mu_{p}^{2}\sigma_{q}^{4}+\mu_{q}^{2}\sigma_{p}^{4}+2\mu_{p}\sigma_{q}^{2}\mu_{q}\sigma_{p}^{2}-\sigma_{p}^{2}\mu_{p}^{2}\sigma_{q}^{2}-\sigma_{q}^{4}\mu_{p}^{2}-\sigma_{p}^{4}\mu_{q}^{2}-\sigma_{q}^{2}\mu_{q}^{2}\sigma_{p}^{2}}{\left(\sigma_{p}^{2}+\sigma_{q}^{2}\right)}\right\} 
\]
\[
=\frac{1}{4\sigma_{p}^{2}\sigma_{q}^{2}}\left\{ \frac{2\mu_{p}\sigma_{q}^{2}\mu_{q}\sigma_{p}^{2}-\sigma_{p}^{2}\mu_{p}^{2}\sigma_{q}^{2}-\sigma_{q}^{2}\mu_{q}^{2}\sigma_{p}^{2}}{\left(\sigma_{p}^{2}+\sigma_{q}^{2}\right)}\right\} =\left.-\frac{1}{4}\left\{ \frac{\left(\mu_{p}-\mu_{q}\right)^{2}}{\sigma_{p}^{2}+\sigma_{q}^{2}}\right\} \right]
\]
The Bhattacharyya distance then becomes,

\[
D_{BC-N}(p,q)=-\ln\left[\rho\left(p,q\right)\right]=-\ln\sqrt{\left\{ \frac{\left(\sigma_{p}^{2}+\sigma_{q}^{2}\right)^{2}}{4\sigma_{p}^{2}\sigma_{q}^{2}}\right\} ^{-\frac{1}{2}}}-\;\ln\left\{ e^{-\frac{1}{4}\frac{\left(\mu_{p}-\mu_{q}\right)^{2}}{\sigma_{p}^{2}+\sigma_{q}^{2}}}\right\} 
\]
\[
=\frac{1}{4}\ln\left(\frac{1}{4}\left(\frac{\sigma_{p}^{2}}{\sigma_{q}^{2}}+\frac{\sigma_{q}^{2}}{\sigma_{p}^{2}}+2\right)\right)+\frac{1}{4}\left(\frac{(\mu_{p}-\mu_{q})^{2}}{\sigma_{p}^{2}+\sigma_{q}^{2}}\right)
\]
\end{doublespace}
\end{proof}
\begin{doublespace}

\subsection{\label{subsec:Proof-of-Proposition: Truncated Multivariate Normal}Proof
of Proposition \ref{prop:Truncated_MN_The-Bhattacharyya-distance}}
\end{doublespace}
\begin{proof}
\begin{doublespace}
Suppose we have two truncated normal distributions $\boldsymbol{p},\boldsymbol{q}$
with density functions, 
\[
f_{\mathbf{p}}\left(x_{1},\ldots,x_{k}\mid\boldsymbol{\mu_{p}},\,\boldsymbol{\Sigma_{p}},\,\boldsymbol{a},\,\boldsymbol{b}\right)=\frac{\exp\left(-\frac{1}{2}\left({\mathbf{x}}-{\boldsymbol{\mu_{p}}}\right)^{\mathrm{T}}{\boldsymbol{\Sigma_{p}}}^{-1}\left({\mathbf{x}}-{\boldsymbol{\mu_{p}}}\right)\right)}{\int_{\boldsymbol{a}}^{\boldsymbol{b}}\exp\left(-\frac{1}{2}\left({\mathbf{x}}-{\boldsymbol{\mu_{p}}}\right)^{\mathrm{T}}{\boldsymbol{\Sigma_{p}}}^{-1}\left({\mathbf{x}}-{\boldsymbol{\mu_{p}}}\right)\right)d\boldsymbol{x};\;\boldsymbol{x}\in\boldsymbol{R}_{\boldsymbol{a}\leq\boldsymbol{x}\leq\boldsymbol{b}}^{k}}
\]
\[
f_{\mathbf{\boldsymbol{q}}}\left(x_{1},\ldots,x_{k}\mid\boldsymbol{\mu_{\boldsymbol{q}}},\,\boldsymbol{\Sigma_{\boldsymbol{q}}},\,\boldsymbol{c},\,\boldsymbol{d}\right)=\frac{\exp\left(-\frac{1}{2}\left({\mathbf{x}}-{\boldsymbol{\mu_{p}}}\right)^{\mathrm{T}}{\boldsymbol{\Sigma_{\boldsymbol{q}}}}^{-1}\left({\mathbf{x}}-{\boldsymbol{\mu_{\boldsymbol{q}}}}\right)\right)}{\int_{\boldsymbol{c}}^{\boldsymbol{d}}\exp\left(-\frac{1}{2}\left({\mathbf{x}}-{\boldsymbol{\mu_{p}}}\right)^{\mathrm{T}}{\boldsymbol{\Sigma_{\boldsymbol{q}}}}^{-1}\left({\mathbf{x}}-{\boldsymbol{\mu_{\boldsymbol{q}}}}\right)\right)d\boldsymbol{x};\;\boldsymbol{x}\in\boldsymbol{R}_{\boldsymbol{\boldsymbol{c}}\leq\boldsymbol{x}\leq\boldsymbol{d}}^{k}}
\]
Here, $\boldsymbol{\mu_{p}}$ is the mean vector and $\boldsymbol{\Sigma_{p}}$
is the symmetric positive definite covariance matrix of the $\boldsymbol{p}$
distribution and the integrals are $k$ dimensional integrals with
lower and upper bounds given by the vectors $\left(\boldsymbol{a},\boldsymbol{b}\right)$
and $\boldsymbol{\boldsymbol{x}\in\boldsymbol{R}_{\boldsymbol{a}\leq\boldsymbol{x}\leq\boldsymbol{b}}^{k}}$
. The Bhattacharyya coefficient is given by,
\[
\rho\left(\boldsymbol{p},\boldsymbol{q}\right)=\int_{\boldsymbol{l=\min\left(a,\boldsymbol{c}\right)}}^{\boldsymbol{u=\min\left(b,d\right)}}\sqrt{f_{\mathbf{p}}\left(x_{1},\ldots,x_{k}\mid\boldsymbol{\mu_{p}},\,\boldsymbol{\Sigma_{p}},\,\boldsymbol{a},\,\boldsymbol{b}\right)f_{\mathbf{\boldsymbol{q}}}\left(x_{1},\ldots,x_{k}\mid\boldsymbol{\mu_{\boldsymbol{q}}},\,\boldsymbol{\Sigma_{\boldsymbol{q}}},\,\boldsymbol{c},\,\boldsymbol{d}\right)}d\boldsymbol{x};\;\boldsymbol{x}\in\boldsymbol{R}_{\boldsymbol{\min\left(a,\boldsymbol{c}\right)}\leq\boldsymbol{x}\leq\boldsymbol{\min\left(b,d\right)}}^{k}
\]
\begin{eqnarray*}
\rho\left(\boldsymbol{p},\boldsymbol{q}\right) & = & \left[\int_{\boldsymbol{a}}^{\boldsymbol{b}}\exp\left(-\frac{1}{2}\left({\mathbf{x}}-{\boldsymbol{\mu_{p}}}\right)^{\mathrm{T}}{\boldsymbol{\Sigma_{p}}}^{-1}\left({\mathbf{x}}-{\boldsymbol{\mu_{p}}}\right)\right)d\boldsymbol{x};\;\boldsymbol{x}\in\boldsymbol{R}_{\boldsymbol{a}\leq\boldsymbol{x}\leq\boldsymbol{b}}^{k}\right]^{-\frac{1}{2}}\\
 &  & \left[\int_{\boldsymbol{c}}^{\boldsymbol{d}}\exp\left(-\frac{1}{2}\left({\mathbf{x}}-{\boldsymbol{\mu_{p}}}\right)^{\mathrm{T}}{\boldsymbol{\Sigma_{\boldsymbol{q}}}}^{-1}\left({\mathbf{x}}-{\boldsymbol{\mu_{\boldsymbol{q}}}}\right)\right)d\boldsymbol{x};\;\boldsymbol{x}\in\boldsymbol{R}_{\boldsymbol{\boldsymbol{c}}\leq\boldsymbol{x}\leq\boldsymbol{d}}^{k}\right]^{-\frac{1}{2}}\\
 &  & \int_{\boldsymbol{l}}^{\boldsymbol{u}}\exp\left(-\frac{1}{4}\left\{ \left({\mathbf{x}-\mathbf{m}}\right)^{\mathrm{T}}\left({\boldsymbol{S}}^{-1}\right)\left({\mathbf{x}-\mathbf{m}}\right)+{\boldsymbol{M}}\right\} \right)d\boldsymbol{x};\;\boldsymbol{x}\in\boldsymbol{R}_{\boldsymbol{\min\left(a,\boldsymbol{c}\right)}\leq\boldsymbol{x}\leq\boldsymbol{\min\left(b,d\right)}}^{k}
\end{eqnarray*}
Here, 
\[
{\boldsymbol{S}}=\left({\boldsymbol{\Sigma_{p}}}^{-1}+{\boldsymbol{\Sigma_{q}}}^{-1}\right)^{-1}={\boldsymbol{\Sigma_{p}}}\left[{\boldsymbol{\Sigma_{q}}}+{\boldsymbol{\Sigma_{p}}}\right]^{-1}{\boldsymbol{\Sigma_{q}}}
\]
\[
{\mathbf{m}}=\left[\left({\boldsymbol{\mu_{p}}}^{\mathrm{T}}{\boldsymbol{\Sigma_{p}}}^{-1}+{\boldsymbol{\mu_{q}}}^{\mathrm{T}}{\boldsymbol{\Sigma_{q}}}^{-1}\right)\left({\boldsymbol{\Sigma_{p}}}^{-1}+{\boldsymbol{\Sigma_{q}}}^{-1}\right)^{-1}\right]^{\mathrm{T}}
\]
\begin{eqnarray*}
{\boldsymbol{M}} & = & \left({\boldsymbol{\mu_{p}}}-{\boldsymbol{\mu_{q}}}\right)^{\mathrm{T}}\left[{\boldsymbol{\Sigma_{p}}}+{\boldsymbol{\Sigma_{q}}}\right]^{-1}\left({\boldsymbol{\mu_{p}}}-{\boldsymbol{\mu_{q}}}\right)
\end{eqnarray*}
\begin{eqnarray*}
\rho\left(\boldsymbol{p},\boldsymbol{q}\right) & = & \left[\frac{1}{\sqrt{(2\pi)^{k}\left(|\boldsymbol{\Sigma_{p}}|\right)}}\int_{\boldsymbol{a}}^{\boldsymbol{b}}\exp\left(-\frac{1}{2}\left({\mathbf{x}}-{\boldsymbol{\mu_{p}}}\right)^{\mathrm{T}}{\boldsymbol{\Sigma_{p}}}^{-1}\left({\mathbf{x}}-{\boldsymbol{\mu_{p}}}\right)\right)d\boldsymbol{x};\;\boldsymbol{x}\in\boldsymbol{R}_{\boldsymbol{a}\leq\boldsymbol{x}\leq\boldsymbol{b}}^{k}\right]^{-\frac{1}{2}}\\
 &  & \left[\frac{1}{\sqrt{(2\pi)^{k}\left(|\boldsymbol{\Sigma_{q}}|\right)}}\int_{\boldsymbol{c}}^{\boldsymbol{d}}\exp\left(-\frac{1}{2}\left({\mathbf{x}}-{\boldsymbol{\mu_{p}}}\right)^{\mathrm{T}}{\boldsymbol{\Sigma_{\boldsymbol{q}}}}^{-1}\left({\mathbf{x}}-{\boldsymbol{\mu_{\boldsymbol{q}}}}\right)\right)d\boldsymbol{x};\;\boldsymbol{x}\in\boldsymbol{R}_{\boldsymbol{\boldsymbol{c}}\leq\boldsymbol{x}\leq\boldsymbol{d}}^{k}\right]^{-\frac{1}{2}}\\
 &  & \left[\exp\left(-\frac{1}{4}{\boldsymbol{M}}\right)\frac{\sqrt{\det\left({\boldsymbol{\Sigma_{p}}}{\boldsymbol{\Sigma}}^{-1}{\boldsymbol{\Sigma_{q}}}\right)}}{\left(|\boldsymbol{\Sigma_{p}}||\boldsymbol{\Sigma_{q}}|\right)^{\frac{1}{4}}}\right]\left[\frac{1}{\sqrt{(2\pi)^{k}\det\left({\boldsymbol{\Sigma_{p}}}{\boldsymbol{\Sigma}}^{-1}{\boldsymbol{\Sigma_{q}}}\right)}}\right.\\
 &  & \left.\int_{\boldsymbol{l}}^{\boldsymbol{u}}\exp\left(-\frac{1}{2}\left\{ \left({\mathbf{x}-\mathbf{m}}\right)^{\mathrm{T}}\left({\boldsymbol{\Sigma_{q}}}^{-1}\left[{\boldsymbol{\Sigma}}\right]{\boldsymbol{\Sigma_{p}}}^{-1}\right)\left({\mathbf{x}-\mathbf{m}}\right)\right\} \right)d\boldsymbol{x};\;\boldsymbol{x}\in\boldsymbol{R}_{\boldsymbol{\min\left(a,\boldsymbol{c}\right)}\leq\boldsymbol{x}\leq\boldsymbol{\min\left(b,d\right)}}^{k}\vphantom{\frac{\sqrt{(2\pi)^{k}\det\left({\boldsymbol{\Sigma_{p}}}^{-1}\right)}}{\sqrt{(2\pi)^{k}\det\left({\boldsymbol{\Sigma_{p}}}^{-1}\right)}}}\right]
\end{eqnarray*}
Here,
\[
\boldsymbol{\Sigma}=\frac{\boldsymbol{\Sigma_{p}}+\boldsymbol{\Sigma_{q}}}{2}
\]
\[
D_{BC-TMN}\left(\boldsymbol{p},\boldsymbol{q}\right)=-\ln\left[\rho\left(\boldsymbol{p},\boldsymbol{q}\right)\right]
\]
\begin{eqnarray*}
D_{BC-TMN}\left(\boldsymbol{p},\boldsymbol{q}\right) & = & \frac{1}{8}(\boldsymbol{\mu_{p}}-\boldsymbol{\mu_{q}})^{T}\boldsymbol{\Sigma}^{-1}(\boldsymbol{\mu_{p}}-\boldsymbol{\mu_{q}})+\frac{1}{2}\ln\,\left(\frac{\det\boldsymbol{\Sigma}}{\sqrt{\det\boldsymbol{\Sigma_{p}}\,\det\boldsymbol{\Sigma_{q}}}}\right)\\
 &  & +\frac{1}{2}\ln\left[\frac{1}{\sqrt{(2\pi)^{k}\left(|\boldsymbol{\Sigma_{p}}|\right)}}\int_{\boldsymbol{a}}^{\boldsymbol{b}}\exp\left(-\frac{1}{2}\left({\mathbf{x}}-{\boldsymbol{\mu_{p}}}\right)^{\mathrm{T}}{\boldsymbol{\Sigma_{p}}}^{-1}\left({\mathbf{x}}-{\boldsymbol{\mu_{p}}}\right)\right)d\boldsymbol{x};\;\boldsymbol{x}\in\boldsymbol{R}_{\boldsymbol{a}\leq\boldsymbol{x}\leq\boldsymbol{b}}^{k}\right]\\
 &  & +\frac{1}{2}\ln\left[\frac{1}{\sqrt{(2\pi)^{k}\left(|\boldsymbol{\Sigma_{q}}|\right)}}\int_{\boldsymbol{c}}^{\boldsymbol{d}}\exp\left(-\frac{1}{2}\left({\mathbf{x}}-{\boldsymbol{\mu_{p}}}\right)^{\mathrm{T}}{\boldsymbol{\Sigma_{\boldsymbol{q}}}}^{-1}\left({\mathbf{x}}-{\boldsymbol{\mu_{\boldsymbol{q}}}}\right)\right)d\boldsymbol{x};\;\boldsymbol{x}\in\boldsymbol{R}_{\boldsymbol{\boldsymbol{c}}\leq\boldsymbol{x}\leq\boldsymbol{d}}^{k}\right]\\
 &  & -\ln\left[\frac{1}{\sqrt{(2\pi)^{k}\det\left({\boldsymbol{\Sigma_{p}}}{\boldsymbol{\Sigma}}^{-1}{\boldsymbol{\Sigma_{q}}}\right)}}\right.\\
 &  & \left.\int_{\boldsymbol{l}}^{\boldsymbol{u}}\exp\left(-\frac{1}{2}\left\{ \left({\mathbf{x}-\mathbf{m}}\right)^{\mathrm{T}}\left({\boldsymbol{\Sigma_{q}}}^{-1}\left[{\boldsymbol{\Sigma}}\right]{\boldsymbol{\Sigma_{p}}}^{-1}\right)\left({\mathbf{x}-\mathbf{m}}\right)\right\} \right)d\boldsymbol{x};\;\boldsymbol{x}\in\boldsymbol{R}_{\boldsymbol{\min\left(a,\boldsymbol{c}\right)}\leq\boldsymbol{x}\leq\boldsymbol{\min\left(b,d\right)}}^{k}\vphantom{\frac{\sqrt{(2\pi)^{k}\det\left({\boldsymbol{\Sigma_{p}}}^{-1}\right)}}{\sqrt{(2\pi)^{k}\det\left({\boldsymbol{\Sigma_{p}}}^{-1}\right)}}}\right]
\end{eqnarray*}
Looking at conditions when $D_{BC-TMN}(\boldsymbol{p},\boldsymbol{q})\geq D_{BC-MN}(\boldsymbol{p},\boldsymbol{q})$
gives,
\begin{eqnarray*}
 &  & \frac{1}{2}\ln\left[\frac{1}{\sqrt{(2\pi)^{k}\left(|\boldsymbol{\Sigma_{p}}|\right)}}\int_{\boldsymbol{a}}^{\boldsymbol{b}}\exp\left(-\frac{1}{2}\left({\mathbf{x}}-{\boldsymbol{\mu_{p}}}\right)^{\mathrm{T}}{\boldsymbol{\Sigma_{p}}}^{-1}\left({\mathbf{x}}-{\boldsymbol{\mu_{p}}}\right)\right)d\boldsymbol{x};\;\boldsymbol{x}\in\boldsymbol{R}_{\boldsymbol{a}\leq\boldsymbol{x}\leq\boldsymbol{b}}^{k}\right]\\
 &  & +\frac{1}{2}\ln\left[\frac{1}{\sqrt{(2\pi)^{k}\left(|\boldsymbol{\Sigma_{q}}|\right)}}\int_{\boldsymbol{c}}^{\boldsymbol{d}}\exp\left(-\frac{1}{2}\left({\mathbf{x}}-{\boldsymbol{\mu_{p}}}\right)^{\mathrm{T}}{\boldsymbol{\Sigma_{\boldsymbol{q}}}}^{-1}\left({\mathbf{x}}-{\boldsymbol{\mu_{\boldsymbol{q}}}}\right)\right)d\boldsymbol{x};\;\boldsymbol{x}\in\boldsymbol{R}_{\boldsymbol{\boldsymbol{c}}\leq\boldsymbol{x}\leq\boldsymbol{d}}^{k}\right]\\
 & \geq & -\ln\left[\frac{1}{\sqrt{(2\pi)^{k}\det\left({\boldsymbol{\Sigma_{p}}}{\boldsymbol{\Sigma}}^{-1}{\boldsymbol{\Sigma_{q}}}\right)}}\right.\\
 &  & \left.\int_{\boldsymbol{l}}^{\boldsymbol{u}}\exp\left(-\frac{1}{2}\left\{ \left({\mathbf{x}-\mathbf{m}}\right)^{\mathrm{T}}\left({\boldsymbol{\Sigma_{q}}}^{-1}\left[{\boldsymbol{\Sigma}}\right]{\boldsymbol{\Sigma_{p}}}^{-1}\right)\left({\mathbf{x}-\mathbf{m}}\right)\right\} \right)d\boldsymbol{x};\;\boldsymbol{x}\in\boldsymbol{R}_{\boldsymbol{\min\left(a,\boldsymbol{c}\right)}\leq\boldsymbol{x}\leq\boldsymbol{\min\left(b,d\right)}}^{k}\vphantom{\frac{\sqrt{(2\pi)^{k}\det\left({\boldsymbol{\Sigma_{p}}}^{-1}\right)}}{\sqrt{(2\pi)^{k}\det\left({\boldsymbol{\Sigma_{p}}}^{-1}\right)}}}\right]
\end{eqnarray*}
For completeness, let us also consider two multivariate normal distributions,
$\boldsymbol{p},\boldsymbol{q}$ ($k$ dimensional random vectors
in our case) where, $\boldsymbol{p}\sim\mathcal{N}(\boldsymbol{\mu_{p}},\,\boldsymbol{\Sigma_{p}})$,
$\boldsymbol{q}\sim\mathcal{N}(\boldsymbol{\mu_{q}},\,\boldsymbol{\Sigma_{q}})$
with density function,
\[
f_{\mathbf{p}}\left(x_{1},\ldots,x_{k}\mid\boldsymbol{\mu_{p}},\,\boldsymbol{\Sigma_{p}}\right)=\frac{1}{\sqrt{(2\pi)^{k}|\boldsymbol{\Sigma_{p}}|}}\exp\left(-\frac{1}{2}\left({\mathbf{x}}-{\boldsymbol{\mu_{p}}}\right)^{\mathrm{T}}{\boldsymbol{\Sigma_{p}}}^{-1}\left({\mathbf{x}}-{\boldsymbol{\mu_{p}}}\right)\right)
\]
Here, $\boldsymbol{\mu_{p}}$ is the mean vector and $\boldsymbol{\Sigma_{p}}$
is the symmetric positive definite covariance matrix of the $\boldsymbol{p}$
distribution. The Bhattacharyya coefficient is given by,
\[
\rho\left(\boldsymbol{p},\boldsymbol{q}\right)=\int\cdots\int\sqrt{f_{\mathbf{p}}\left(x_{1},\ldots,x_{k}\mid\boldsymbol{\mu_{p}},\,\boldsymbol{\Sigma_{p}}\right)f_{\mathbf{q}}\left(x_{1},\ldots,x_{k}\mid\boldsymbol{\mu_{q}},\,\boldsymbol{\Sigma_{q}}\right)}dx_{1}\cdots dx_{k}
\]
\[
\rho\left(\boldsymbol{p},\boldsymbol{q}\right)=\int\cdots\int\sqrt{\frac{1}{(2\pi)^{k}\sqrt{|\boldsymbol{\Sigma_{p}}||\boldsymbol{\Sigma_{q}}|}}\exp\left(-\frac{1}{2}\left\{ \left({\mathbf{x}}-{\boldsymbol{\mu_{p}}}\right)^{\mathrm{T}}{\boldsymbol{\Sigma_{p}}}^{-1}\left({\mathbf{x}}-{\boldsymbol{\mu_{p}}}\right)+\left({\mathbf{x}}-{\boldsymbol{\mu_{q}}}\right)^{\mathrm{T}}{\boldsymbol{\Sigma_{q}}}^{-1}\left({\mathbf{x}}-{\boldsymbol{\mu_{q}}}\right)\right\} \right)}dx_{1}\cdots dx_{k}
\]
\begin{eqnarray*}
\rho\left(\boldsymbol{p},\boldsymbol{q}\right) & = & \int\cdots\int\frac{1}{(2\pi)^{\frac{k}{2}}\left(|\boldsymbol{\Sigma_{p}}||\boldsymbol{\Sigma_{q}}|\right)^{\frac{1}{4}}}\exp\left(-\frac{1}{4}\left\{ \left({\mathbf{x}-\mathbf{m}}\right)^{\mathrm{T}}\left({\boldsymbol{S}}^{-1}\right)\left({\mathbf{x}-\mathbf{m}}\right)+{\boldsymbol{M}}\right\} \right)dx_{1}\cdots dx_{k}
\end{eqnarray*}
\begin{eqnarray*}
\left({\mathbf{x}}-{\boldsymbol{\mu_{p}}}\right)^{\mathrm{T}}{\boldsymbol{\Sigma_{p}}}^{-1}\left({\mathbf{x}}-{\boldsymbol{\mu_{p}}}\right)+\left({\mathbf{x}}-{\boldsymbol{\mu_{q}}}\right)^{\mathrm{T}}{\boldsymbol{\Sigma_{q}}}^{-1}\left({\mathbf{x}}-{\boldsymbol{\mu_{q}}}\right) & =\\
{\mathbf{x}}^{\mathrm{T}}{\boldsymbol{\Sigma_{p}}}^{-1}{\mathbf{x}}-{\boldsymbol{\mu_{p}}}^{\mathrm{T}}{\boldsymbol{\Sigma_{p}}}^{-1}{\mathbf{x}}-{\mathbf{x}}^{\mathrm{T}}{\boldsymbol{\Sigma_{p}}}^{-1}{\boldsymbol{\mu_{p}}}+{\boldsymbol{\mu_{p}}}^{\mathrm{T}}{\boldsymbol{\Sigma_{p}}}^{-1}{\boldsymbol{\mu_{p}}}\\
+{\mathbf{x}}^{\mathrm{T}}{\boldsymbol{\Sigma_{q}}}^{-1}{\mathbf{x}}-{\boldsymbol{\mu_{q}}}^{\mathrm{T}}{\boldsymbol{\Sigma_{q}}}^{-1}{\mathbf{x}}-{\mathbf{x}}^{\mathrm{T}}{\boldsymbol{\Sigma_{q}}}^{-1}{\boldsymbol{\mu_{q}}}+{\boldsymbol{\mu_{q}}}^{\mathrm{T}}{\boldsymbol{\Sigma_{q}}}^{-1}{\boldsymbol{\mu_{q}}}
\end{eqnarray*}
\[
={\mathbf{x}}^{\mathrm{T}}\left({\boldsymbol{\Sigma_{p}}}^{-1}+{\boldsymbol{\Sigma_{q}}}^{-1}\right){\mathbf{x}}-2\left({\boldsymbol{\mu_{p}}}^{\mathrm{T}}{\boldsymbol{\Sigma_{p}}}^{-1}+{\boldsymbol{\mu_{q}}}^{\mathrm{T}}{\boldsymbol{\Sigma_{q}}}^{-1}\right){\mathbf{x}}+{\boldsymbol{\mu_{p}}}^{\mathrm{T}}{\boldsymbol{\Sigma_{p}}}^{-1}{\boldsymbol{\mu_{p}}}+{\boldsymbol{\mu_{q}}}^{\mathrm{T}}{\boldsymbol{\Sigma_{q}}}^{-1}{\boldsymbol{\mu_{q}}}
\]
We want this to be of the form,
\[
\left({\mathbf{x}-\mathbf{m}}\right)^{\mathrm{T}}\left({\boldsymbol{S}}^{-1}\right)\left({\mathbf{x}-\mathbf{m}}\right)+{\boldsymbol{M}}={\mathbf{x}}^{\mathrm{T}}{\boldsymbol{S}}^{-1}{\mathbf{x}}-{\mathbf{m}}^{\mathrm{T}}{\boldsymbol{S}}^{-1}{\mathbf{x}}-{\mathbf{x}}^{\mathrm{T}}{\boldsymbol{S}}^{-1}{\mathbf{m}}+{\mathbf{m}}^{\mathrm{T}}{\boldsymbol{S}}^{-1}{\mathbf{m}}+{\boldsymbol{M}}
\]
Comparing the coefficients of ${\mathbf{x}}$,
\[
{\boldsymbol{S}}^{-1}=\left({\boldsymbol{\Sigma_{p}}}^{-1}+{\boldsymbol{\Sigma_{q}}}^{-1}\right)
\]
\[
{\mathbf{m}}^{\mathrm{T}}{\boldsymbol{S}}^{-1}=\left({\boldsymbol{\mu_{p}}}^{\mathrm{T}}{\boldsymbol{\Sigma_{p}}}^{-1}+{\boldsymbol{\mu_{q}}}^{\mathrm{T}}{\boldsymbol{\Sigma_{q}}}^{-1}\right)\Rightarrow{\mathbf{m}}=\left[\left({\boldsymbol{\mu_{p}}}^{\mathrm{T}}{\boldsymbol{\Sigma_{p}}}^{-1}+{\boldsymbol{\mu_{q}}}^{\mathrm{T}}{\boldsymbol{\Sigma_{q}}}^{-1}\right)\left({\boldsymbol{\Sigma_{p}}}^{-1}+{\boldsymbol{\Sigma_{q}}}^{-1}\right)^{-1}\right]^{\mathrm{T}}
\]
\[
{\boldsymbol{M}}={\boldsymbol{\mu_{p}}}^{\mathrm{T}}{\boldsymbol{\Sigma_{p}}}^{-1}{\boldsymbol{\mu_{p}}}+{\boldsymbol{\mu_{q}}}^{\mathrm{T}}{\boldsymbol{\Sigma_{q}}}^{-1}{\boldsymbol{\mu_{q}}}-{\mathbf{m}}^{\mathrm{T}}{\boldsymbol{S}}^{-1}{\mathbf{m}}
\]
\[
={\boldsymbol{\mu_{p}}}^{\mathrm{T}}{\boldsymbol{\Sigma_{p}}}^{-1}{\boldsymbol{\mu_{p}}}+{\boldsymbol{\mu_{q}}}^{\mathrm{T}}{\boldsymbol{\Sigma_{q}}}^{-1}{\boldsymbol{\mu_{q}}}-\left({\boldsymbol{\mu_{p}}}^{\mathrm{T}}{\boldsymbol{\Sigma_{p}}}^{-1}+{\boldsymbol{\mu_{q}}}^{\mathrm{T}}{\boldsymbol{\Sigma_{q}}}^{-1}\right)\left[{\boldsymbol{S}}\right]^{\mathrm{T}}\left({\boldsymbol{\mu_{p}}}^{\mathrm{T}}{\boldsymbol{\Sigma_{p}}}^{-1}+{\boldsymbol{\mu_{q}}}^{\mathrm{T}}{\boldsymbol{\Sigma_{q}}}^{-1}\right)^{\mathrm{T}}
\]
\[
={\boldsymbol{\mu_{p}}}^{\mathrm{T}}{\boldsymbol{\Sigma_{p}}}^{-1}{\boldsymbol{\mu_{p}}}+{\boldsymbol{\mu_{q}}}^{\mathrm{T}}{\boldsymbol{\Sigma_{q}}}^{-1}{\boldsymbol{\mu_{q}}}-\left({\boldsymbol{\mu_{p}}}^{\mathrm{T}}{\boldsymbol{\Sigma_{p}}}^{-1}+{\boldsymbol{\mu_{q}}}^{\mathrm{T}}{\boldsymbol{\Sigma_{q}}}^{-1}\right)\left[{\boldsymbol{S}}\right]\left({\boldsymbol{\Sigma_{p}}}^{-1}{\boldsymbol{\mu_{p}}}+{\boldsymbol{\Sigma_{q}}}^{-1}{\boldsymbol{\mu_{q}}}\right)
\]
\begin{eqnarray*}
 & = & {\boldsymbol{\mu_{p}}}^{\mathrm{T}}{\boldsymbol{\Sigma_{p}}}^{-1}{\boldsymbol{\mu_{p}}}+{\boldsymbol{\mu_{q}}}^{\mathrm{T}}{\boldsymbol{\Sigma_{q}}}^{-1}{\boldsymbol{\mu_{q}}}-{\boldsymbol{\mu_{p}}}^{\mathrm{T}}{\boldsymbol{\Sigma_{p}}}^{-1}{\boldsymbol{S}}{\boldsymbol{\Sigma_{p}}}^{-1}{\boldsymbol{\mu_{p}}}\\
 &  & -{\boldsymbol{\mu_{q}}}^{\mathrm{T}}{\boldsymbol{\Sigma_{q}}}^{-1}{\boldsymbol{S}}{\boldsymbol{\Sigma_{p}}}^{-1}{\boldsymbol{\mu_{p}}}-{\boldsymbol{\mu_{p}}}^{\mathrm{T}}{\boldsymbol{\Sigma_{p}}}^{-1}{\boldsymbol{S}}{\boldsymbol{\Sigma_{q}}}^{-1}{\boldsymbol{\mu_{q}}}-{\boldsymbol{\mu_{q}}}^{\mathrm{T}}{\boldsymbol{\Sigma_{q}}}^{-1}{\boldsymbol{S}}{\boldsymbol{\Sigma_{q}}}^{-1}{\boldsymbol{\mu_{q}}}
\end{eqnarray*}
Using the result in (Henderson \& Searle 1981, eq-22),
\[
{\boldsymbol{S}}=\left({\boldsymbol{\Sigma_{p}}}^{-1}+{\boldsymbol{\Sigma_{q}}}^{-1}\right)^{-1}={\boldsymbol{\Sigma_{p}}}-{\boldsymbol{\Sigma_{p}}}\left({\boldsymbol{I}}+{\boldsymbol{\Sigma_{q}}}^{-1}{\boldsymbol{\Sigma_{p}}}\right)^{-1}{\boldsymbol{\Sigma_{q}}}^{-1}{\boldsymbol{\Sigma_{p}}}
\]
\[
={\boldsymbol{\Sigma_{p}}}-{\boldsymbol{\Sigma_{p}}}\left({\boldsymbol{\Sigma_{p}}}^{-1}{\boldsymbol{\Sigma_{p}}}+{\boldsymbol{\Sigma_{q}}}^{-1}{\boldsymbol{\Sigma_{p}}}\right)^{-1}{\boldsymbol{\Sigma_{q}}}^{-1}{\boldsymbol{\Sigma_{p}}}
\]
\[
={\boldsymbol{\Sigma_{p}}}-{\boldsymbol{\Sigma_{p}}}{\boldsymbol{\Sigma_{p}}}^{-1}\left({\boldsymbol{\Sigma_{p}}}^{-1}+{\boldsymbol{\Sigma_{q}}}^{-1}\right)^{-1}{\boldsymbol{\Sigma_{q}}}^{-1}{\boldsymbol{\Sigma_{p}}}
\]
\[
\Rightarrow\left({\boldsymbol{\Sigma_{p}}}^{-1}+{\boldsymbol{\Sigma_{q}}}^{-1}\right)^{-1}\left[{\boldsymbol{I}}+{\boldsymbol{\Sigma_{q}}}^{-1}{\boldsymbol{\Sigma_{p}}}\right]={\boldsymbol{\Sigma_{p}}}
\]
\[
\Rightarrow\left({\boldsymbol{\Sigma_{p}}}^{-1}+{\boldsymbol{\Sigma_{q}}}^{-1}\right)^{-1}={\boldsymbol{\Sigma_{p}}}\left[{\boldsymbol{\Sigma_{q}}}^{-1}{\boldsymbol{\Sigma_{q}}}+{\boldsymbol{\Sigma_{q}}}^{-1}{\boldsymbol{\Sigma_{p}}}\right]^{-1}
\]
\[
\Rightarrow\left({\boldsymbol{\Sigma_{p}}}^{-1}+{\boldsymbol{\Sigma_{q}}}^{-1}\right)^{-1}={\boldsymbol{\Sigma_{p}}}\left[{\boldsymbol{\Sigma_{q}}}+{\boldsymbol{\Sigma_{p}}}\right]^{-1}{\boldsymbol{\Sigma_{q}}}
\]
From symmetry we can write,
\[
{\boldsymbol{S}}=\left({\boldsymbol{\Sigma_{p}}}^{-1}+{\boldsymbol{\Sigma_{q}}}^{-1}\right)^{-1}={\boldsymbol{\Sigma_{q}}}-{\boldsymbol{\Sigma_{q}}}\left({\boldsymbol{I}}+{\boldsymbol{\Sigma_{p}}}^{-1}{\boldsymbol{\Sigma_{q}}}\right)^{-1}{\boldsymbol{\Sigma_{p}}}^{-1}{\boldsymbol{\Sigma_{q}}}
\]
\[
={\boldsymbol{\Sigma_{q}}}-{\boldsymbol{\Sigma_{q}}}\left({\boldsymbol{\Sigma_{q}}}^{-1}{\boldsymbol{\Sigma_{q}}}+{\boldsymbol{\Sigma_{p}}}^{-1}{\boldsymbol{\Sigma_{q}}}\right)^{-1}{\boldsymbol{\Sigma_{p}}}^{-1}{\boldsymbol{\Sigma_{q}}}
\]
\[
={\boldsymbol{\Sigma_{q}}}-{\boldsymbol{\Sigma_{q}}}{\boldsymbol{\Sigma_{q}}}^{-1}\left({\boldsymbol{\Sigma_{q}}}^{-1}+{\boldsymbol{\Sigma_{p}}}^{-1}\right)^{-1}{\boldsymbol{\Sigma_{p}}}^{-1}{\boldsymbol{\Sigma_{q}}}
\]
\[
\Rightarrow\left({\boldsymbol{\Sigma_{p}}}^{-1}+{\boldsymbol{\Sigma_{q}}}^{-1}\right)^{-1}\left[{\boldsymbol{I}}+{\boldsymbol{\Sigma_{p}}}^{-1}{\boldsymbol{\Sigma_{q}}}\right]={\boldsymbol{\Sigma_{q}}}
\]
\[
\Rightarrow\left({\boldsymbol{\Sigma_{p}}}^{-1}+{\boldsymbol{\Sigma_{q}}}^{-1}\right)^{-1}={\boldsymbol{\Sigma_{q}}}\left[{\boldsymbol{\Sigma_{p}}}^{-1}{\boldsymbol{\Sigma_{p}}}+{\boldsymbol{\Sigma_{p}}}^{-1}{\boldsymbol{\Sigma_{q}}}\right]^{-1}
\]
\[
\Rightarrow\left({\boldsymbol{\Sigma_{p}}}^{-1}+{\boldsymbol{\Sigma_{q}}}^{-1}\right)^{-1}={\boldsymbol{\Sigma_{q}}}\left[{\boldsymbol{\Sigma_{p}}}+{\boldsymbol{\Sigma_{q}}}\right]^{-1}{\boldsymbol{\Sigma_{p}}}
\]
Using this in the result for ${\boldsymbol{M}}$, 
\begin{eqnarray*}
{\boldsymbol{M}} & = & {\boldsymbol{\mu_{p}}}^{\mathrm{T}}{\boldsymbol{\Sigma_{p}}}^{-1}{\boldsymbol{\mu_{p}}}-{\boldsymbol{\mu_{p}}}^{\mathrm{T}}{\boldsymbol{\Sigma_{p}}}^{-1}{\boldsymbol{S}}{\boldsymbol{\Sigma_{p}}}^{-1}{\boldsymbol{\mu_{p}}}-{\boldsymbol{\mu_{p}}}^{\mathrm{T}}{\boldsymbol{\Sigma_{p}}}^{-1}{\boldsymbol{\Sigma_{p}}}\left[{\boldsymbol{\Sigma_{p}}}+{\boldsymbol{\Sigma_{q}}}\right]^{-1}{\boldsymbol{\Sigma_{q}}}{\boldsymbol{\Sigma_{q}}}^{-1}{\boldsymbol{\mu_{q}}}\\
 &  & +{\boldsymbol{\mu_{q}}}^{\mathrm{T}}{\boldsymbol{\Sigma_{q}}}^{-1}{\boldsymbol{\mu_{q}}}-{\boldsymbol{\mu_{q}}}^{\mathrm{T}}{\boldsymbol{\Sigma_{q}}}^{-1}{\boldsymbol{S}}{\boldsymbol{\Sigma_{q}}}^{-1}{\boldsymbol{\mu_{q}}}-{\boldsymbol{\mu_{q}}}^{\mathrm{T}}{\boldsymbol{\Sigma_{q}}}^{-1}{\boldsymbol{\Sigma_{q}}}\left[{\boldsymbol{\Sigma_{p}}}+{\boldsymbol{\Sigma_{q}}}\right]^{-1}{\boldsymbol{\Sigma_{p}}}{\boldsymbol{\Sigma_{p}}}^{-1}{\boldsymbol{\mu_{p}}}
\end{eqnarray*}
\begin{eqnarray*}
{\boldsymbol{M}} & = & {\boldsymbol{\mu_{p}}}^{\mathrm{T}}{\boldsymbol{\Sigma_{p}}}^{-1}\left[{\boldsymbol{S}}{\boldsymbol{S}}^{-1}-{\boldsymbol{S}}{\boldsymbol{\Sigma_{p}}}^{-1}\right]{\boldsymbol{\mu_{p}}}-{\boldsymbol{\mu_{p}}}^{\mathrm{T}}\left[{\boldsymbol{\Sigma_{p}}}+{\boldsymbol{\Sigma_{q}}}\right]^{-1}{\boldsymbol{\mu_{q}}}\\
 &  & +{\boldsymbol{\mu_{q}}}^{\mathrm{T}}{\boldsymbol{\Sigma_{q}}}^{-1}\left[{\boldsymbol{S}}{\boldsymbol{S}}^{-1}-{\boldsymbol{S}}{\boldsymbol{\Sigma_{q}}}^{-1}\right]{\boldsymbol{\mu_{q}}}-{\boldsymbol{\mu_{q}}}^{\mathrm{T}}\left[{\boldsymbol{\Sigma_{p}}}+{\boldsymbol{\Sigma_{q}}}\right]^{-1}{\boldsymbol{\mu_{p}}}
\end{eqnarray*}
\begin{eqnarray*}
{\boldsymbol{M}} & = & {\boldsymbol{\mu_{p}}}^{\mathrm{T}}{\boldsymbol{\Sigma_{p}}}^{-1}\left[{\boldsymbol{S}}\left\{ {\boldsymbol{S}}^{-1}-{\boldsymbol{\Sigma_{p}}}^{-1}\right\} \right]{\boldsymbol{\mu_{p}}}-{\boldsymbol{\mu_{p}}}^{\mathrm{T}}\left[{\boldsymbol{\Sigma_{p}}}+{\boldsymbol{\Sigma_{q}}}\right]^{-1}{\boldsymbol{\mu_{q}}}\\
 &  & +{\boldsymbol{\mu_{q}}}^{\mathrm{T}}{\boldsymbol{\Sigma_{q}}}^{-1}\left[{\boldsymbol{S}}\left\{ {\boldsymbol{S}}^{-1}-{\boldsymbol{\Sigma_{q}}}^{-1}\right\} \right]{\boldsymbol{\mu_{q}}}-{\boldsymbol{\mu_{q}}}^{\mathrm{T}}\left[{\boldsymbol{\Sigma_{p}}}+{\boldsymbol{\Sigma_{q}}}\right]^{-1}{\boldsymbol{\mu_{p}}}
\end{eqnarray*}
\begin{eqnarray*}
{\boldsymbol{M}} & = & {\boldsymbol{\mu_{p}}}^{\mathrm{T}}{\boldsymbol{\Sigma_{p}}}^{-1}\left[{\boldsymbol{\Sigma_{p}}}\left[{\boldsymbol{\Sigma_{q}}}+{\boldsymbol{\Sigma_{p}}}\right]^{-1}{\boldsymbol{\Sigma_{q}}}{\boldsymbol{\Sigma_{q}}}^{-1}\right]{\boldsymbol{\mu_{p}}}-{\boldsymbol{\mu_{p}}}^{\mathrm{T}}\left[{\boldsymbol{\Sigma_{p}}}+{\boldsymbol{\Sigma_{q}}}\right]^{-1}{\boldsymbol{\mu_{q}}}\\
 &  & +{\boldsymbol{\mu_{q}}}^{\mathrm{T}}{\boldsymbol{\Sigma_{q}}}^{-1}\left[{\boldsymbol{\Sigma_{q}}}\left[{\boldsymbol{\Sigma_{p}}}+{\boldsymbol{\Sigma_{q}}}\right]^{-1}{\boldsymbol{\Sigma_{p}}}{\boldsymbol{\Sigma_{p}}}^{-1}\right]{\boldsymbol{\mu_{q}}}-{\boldsymbol{\mu_{q}}}^{\mathrm{T}}\left[{\boldsymbol{\Sigma_{p}}}+{\boldsymbol{\Sigma_{q}}}\right]^{-1}{\boldsymbol{\mu_{p}}}
\end{eqnarray*}
\begin{eqnarray*}
{\boldsymbol{M}} & = & {\boldsymbol{\mu_{p}}}^{\mathrm{T}}\left[{\boldsymbol{\Sigma_{p}}}+{\boldsymbol{\Sigma_{q}}}\right]^{-1}{\boldsymbol{\mu_{p}}}-{\boldsymbol{\mu_{p}}}^{\mathrm{T}}\left[{\boldsymbol{\Sigma_{p}}}+{\boldsymbol{\Sigma_{q}}}\right]^{-1}{\boldsymbol{\mu_{q}}}\\
 &  & +{\boldsymbol{\mu_{q}}}^{\mathrm{T}}\left[{\boldsymbol{\Sigma_{p}}}+{\boldsymbol{\Sigma_{q}}}\right]^{-1}{\boldsymbol{\mu_{q}}}-{\boldsymbol{\mu_{q}}}^{\mathrm{T}}\left[{\boldsymbol{\Sigma_{p}}}+{\boldsymbol{\Sigma_{q}}}\right]^{-1}{\boldsymbol{\mu_{p}}}
\end{eqnarray*}
\begin{eqnarray*}
{\boldsymbol{M}} & = & \left({\boldsymbol{\mu_{p}}}-{\boldsymbol{\mu_{q}}}\right)^{\mathrm{T}}\left[{\boldsymbol{\Sigma_{p}}}+{\boldsymbol{\Sigma_{q}}}\right]^{-1}\left({\boldsymbol{\mu_{p}}}-{\boldsymbol{\mu_{q}}}\right)
\end{eqnarray*}
Let $\boldsymbol{\Sigma}=\frac{\boldsymbol{\Sigma_{p}}+\boldsymbol{\Sigma_{q}}}{2}$,
\begin{eqnarray*}
\rho\left(\boldsymbol{p},\boldsymbol{q}\right) & = & \int\cdots\int\left[\frac{\sqrt{\det\left({\boldsymbol{\Sigma_{p}}}{\boldsymbol{\Sigma}}^{-1}{\boldsymbol{\Sigma_{q}}}\right)}}{\left(|\boldsymbol{\Sigma_{p}}||\boldsymbol{\Sigma_{q}}|\right)^{\frac{1}{4}}\sqrt{(2\pi)^{k}\det\left({\boldsymbol{\Sigma_{p}}}{\boldsymbol{\Sigma}}^{-1}{\boldsymbol{\Sigma_{q}}}\right)}}\right.\\
 &  & \left.\exp\left(-\frac{1}{4}{\boldsymbol{M}}\right)\exp\left(-\frac{1}{2}\left\{ \left({\mathbf{x}-\mathbf{m}}\right)^{\mathrm{T}}\left({\boldsymbol{\Sigma_{q}}}^{-1}\left[{\boldsymbol{\Sigma}}\right]{\boldsymbol{\Sigma_{p}}}^{-1}\right)\left({\mathbf{x}-\mathbf{m}}\right)\right\} \right)dx_{1}\cdots dx_{k}\vphantom{\frac{\sqrt{\det\left({\boldsymbol{\Sigma_{p}}}\right)}}{\sqrt{\det\left({\boldsymbol{\Sigma_{p}}}\right)}}}\right]
\end{eqnarray*}
\begin{eqnarray*}
\rho\left(\boldsymbol{p},\boldsymbol{q}\right) & = & \frac{\left(|{\boldsymbol{\Sigma_{p}}}||{\boldsymbol{\Sigma_{q}}}|\right)^{\frac{1}{4}}}{\left(|{\boldsymbol{\Sigma}}|\right)^{\frac{1}{2}}}\exp\left(-\frac{1}{4}{\boldsymbol{M}}\right)
\end{eqnarray*}
\begin{eqnarray*}
D_{BC-MN}\left(\boldsymbol{p},\boldsymbol{q}\right) & = & -\ln\left[\rho\left(\boldsymbol{p},\boldsymbol{q}\right)\right]=\frac{1}{8}(\boldsymbol{\mu_{p}}-\boldsymbol{\mu_{q}})^{T}\boldsymbol{\Sigma}^{-1}(\boldsymbol{\mu_{p}}-\boldsymbol{\mu_{q}})+\frac{1}{2}\ln\,\left(\frac{\det\boldsymbol{\Sigma}}{\sqrt{\det\boldsymbol{\Sigma_{p}}\,\det\boldsymbol{\Sigma_{q}}}}\right)
\end{eqnarray*}
\end{doublespace}
\end{proof}
\begin{doublespace}

\subsection{\label{subsec:Proof-of-Proposition: Stein Lemma Generic}Proof of
Proposition \ref{prop:Stein_Lemma_Generic}}
\end{doublespace}
\begin{proof}
\begin{doublespace}
Since the following differential equation is satisfied,
\[
\frac{f_{XY}'\left(t,u\right)}{f_{XY}\left(t,u\right)}=-\frac{g'\left(t,u\right)}{g\left(t,u\right)}+\frac{\left[\mu_{Y}-h\left(u\right)\right]}{g\left(t,u\right)}\quad,\quad t,u\in(a,b)
\]
\[
\Rightarrow\quad f_{XY}'\left(t,u\right)g\left(t,u\right)+g'\left(t,u\right)f_{XY}\left(t,u\right)=\left[\mu_{Y}-h\left(u\right)\right]f_{XY}\left(t,u\right)
\]
\[
\Rightarrow\quad\frac{\partial f_{XY}\left(t,u\right)g\left(t,u\right)}{\partial t}=\left[\mu_{Y}-h\left(u\right)\right]f_{XY}\left(t,u\right)
\]
Integrating with respect to $t$ from $r$ to $b$ and assuming $\underset{t\rightarrow b}{\lim}\;g\left(t,u\right)f_{XY}\left(t,u\right)=0$
shows that for a given $h\left(u\right)$ the value of $g\left(t,u\right)$
uniquely determines the joint distribution of $X$ and $Y$.
\[
\left|f_{XY}\left(t,u\right)g\left(t,u\right)\right|_{r}^{b}=\int_{r}^{b}\left[\mu_{Y}-h\left(u\right)\right]\:f_{XY}\left(t,u\right)\:dt
\]
\[
f_{XY}\left(r,u\right)\:g\left(r,u\right)=\int_{r}^{b}\left[h\left(u\right)-\mu_{Y}\right]\:f_{XY}\left(t,u\right)\:dt
\]
\[
\int_{a}^{b}f_{XY}\left(r,u\right)\:g\left(r,u\right)\:du=\int_{a}^{b}\int_{r}^{b}\left[h\left(u\right)-\mu_{Y}\right]\:f_{XY}\left(t,u\right)\:dt\:du
\]
Similarly, integrating with respect to $t$ from $a$ to $r$ and
assuming $\underset{t\rightarrow a}{\lim}\;g\left(t,u\right)f_{XY}\left(t,u\right)=0$
gives,
\[
f_{XY}\left(r,u\right)\:g\left(r,u\right)=\int_{a}^{r}\left[\mu_{Y}-h\left(u\right)\right]\:f_{XY}\left(t,u\right)\:dt
\]
\[
\int_{a}^{b}f_{XY}\left(r,u\right)\:g\left(r,u\right)\:du=\int_{a}^{b}\int_{a}^{r}\left[\mu_{Y}-h\left(u\right)\right]\:f_{XY}\left(t,u\right)\:dt\:du
\]
Now consider,
\[
\text{Cov}\left[c\left(X\right),h\left(Y\right)\right]=\int_{a}^{b}\int_{a}^{b}\:\left[c\left(t\right)-E\left\{ c\left(X\right)\right\} \right]\:\left[h\left(u\right)-E\left\{ h\left(Y\right)\right\} \right]\:f_{XY}\left(t,u\right)\:dt\:du
\]
\[
=\int_{a}^{b}\int_{a}^{b}\:c\left(t\right)\:\left[h\left(u\right)-\mu_{Y}\right]\:f_{XY}\left(t,u\right)\:dt\:du-\int_{a}^{b}\int_{a}^{b}E\left\{ c\left(X\right)\right\} \:\left(h\left(u\right)-\mu_{Y}\right)\:f_{XY}\left(t,u\right)\:dt\:du
\]
\[
=\int_{a}^{b}\int_{a}^{b}\:c\left(t\right)\:\left[h\left(u\right)-\mu_{Y}\right]\:f_{XY}\left(t,u\right)\:dt\:du\quad\left\{ \because\;\int_{a}^{b}\int_{a}^{b}\:E\left\{ c\left(X\right)\right\} \:\left(h\left(u\right)-\mu_{Y}\right)f_{XY}\left(t,u\right)\:dt\:du=0\right\} 
\]
\[
=\int_{a}^{b}\int_{a}^{b}\:\left[c\left(t\right)-c\left(a\right)\right]\:\left[h\left(u\right)-\mu_{Y}\right]\:f_{XY}\left(t,u\right)\:dt\:du
\]
\[
=\int_{a}^{b}\int_{a}^{b}\:\left[\int_{a}^{t}\:c'\left(r\right)\:dr\right]\:\left[h\left(u\right)-\mu_{Y}\right]\:f_{XY}\left(t,u\right)\:dt\:du
\]
Using Fubini's theorem and interchanging the order of integration,
\[
\text{Cov}\left[c\left(X\right),h\left(Y\right)\right]=\int_{a}^{b}\int_{a}^{b}\:c'\left(r\right)\:\left\{ \int_{r}^{b}\:\left[h\left(u\right)-\mu_{Y}\right]\:f_{XY}\left(t,u\right)\:dt\right\} \:dr\:du
\]
\[
=\int_{a}^{b}\int_{a}^{b}\:c'\left(r\right)\:f_{XY}\left(r,u\right)\:g\left(r,u\right)\:dr\:du
\]
\[
=E\left[c'\left(r\right)\:g\left(r,u\right)\right]=E\left[c'\left(X\right)\:g\left(X,Y\right)\right]
\]
\end{doublespace}
\end{proof}
\begin{doublespace}

\subsection{\label{subsec:Proof-of-Proposition: Distance_Covariance_Relationship}Proof
of Proposition \ref{prop:Distance-Covariance-Relationship}}
\end{doublespace}
\begin{proof}
\begin{doublespace}
We start with the definition of covariance as below,
\[
\text{Cov}\left[c\left(X\right),Y\right]=\int\int\left[c\left(t\right)-E\left\{ c\left(X\right)\right\} \right]\left(u-\mu_{Y}\right)f_{XY}\left(t,u\right)\:dt\:du
\]
\[
=\int\int\left[c\left(t\right)u-E\left\{ c\left(X\right)\right\} u-c\left(t\right)\mu_{Y}+E\left\{ c\left(X\right)\right\} \mu_{Y}\right]f_{XY}\left(t,u\right)\:dt\:du
\]
\[
=\int\int c\left(t\right)u\:f_{XY}\left(t,u\right)\:dt\:du-E\left\{ c\left(X\right)\right\} \mu_{Y}
\]
\[
=\int\int c\left(t\right)\:u\:f_{\left(Y\mid X\right)}\left(u\mid t\right)\:f_{X}\left(t\right)\:dt\:du-E\left\{ c\left(X\right)\right\} \mu_{Y}
\]
\[
=\int\int c\left(t\right)\:u\:f_{\left(X\mid Y\right)}\left(t\mid u\right)\:f_{Y}\left(u\right)\:dt\:du-E\left\{ c\left(X\right)\right\} \mu_{Y}
\]
Set,
\[
c\left(t\right)=t-\sqrt{\frac{f_{Y}\left(t\right)}{f_{X}\left(t\right)}}
\]
\[
\text{Cov}\left[c\left(X\right),Y\right]=\int\int\left\{ t-\sqrt{\frac{f_{Y}\left(t\right)}{f_{X}\left(t\right)}}\right\} \:u\:f_{XY}\left(t,u\right)\:dt\:du-E\left[\left\{ t-\sqrt{\frac{f_{Y}\left(t\right)}{f_{X}\left(t\right)}}\right\} \right]\mu_{Y}
\]
\[
=\text{Cov}\left(X,Y\right)-\int\int\sqrt{\frac{f_{Y}\left(t\right)}{f_{X}\left(t\right)}}\:u\:f_{\left(Y\mid X\right)}\left(u\mid t\right)\:f_{X}\left(t\right)\:dt\:du+\mu_{Y}\rho\left(f_{X},f_{Y}\right)
\]
\[
=\text{Cov}\left(X,Y\right)-\int\sqrt{\frac{f_{Y}\left(t\right)}{f_{X}\left(t\right)}}\:E\left[Y\mid X=t\right]\:f_{X}\left(t\right)\:dt+\mu_{Y}\rho\left(f_{X},f_{Y}\right)
\]
\[
=\text{Cov}\left(X,Y\right)-E\left[\sqrt{\frac{f_{Y}\left(t\right)}{f_{X}\left(t\right)}}Y\right]+\mu_{Y}\rho\left(f_{X},f_{Y}\right)
\]
\[
=\text{Cov}\left(X,Y\right)-\int\int\sqrt{\frac{f_{Y}\left(t\right)}{f_{X}\left(t\right)}}\:u\:f_{\left(X\mid Y\right)}\left(t\mid u\right)\:f_{Y}\left(u\right)\:dt\:du+\mu_{Y}\rho\left(f_{X},f_{Y}\right)
\]
\end{doublespace}
\end{proof}
\begin{doublespace}

\subsection{\label{subsec:Proof-of-Proposition: Asset Pricing}Proof of Proposition
\ref{prop:The-asset-pricing}}
\end{doublespace}
\begin{proof}
\begin{doublespace}
Proved within the main body of the article.
\end{doublespace}
\end{proof}
\begin{doublespace}

\section{\label{sec:R-Code-Snippets} Appendix: R Code Snippets}
\end{doublespace}

\begin{doublespace}
\inputencoding{latin9}\begin{lstlisting}
#Johnson Linderstrauss Lemma
minimumDimension = function(numberOfPoints, errorTolerance){
	result = (4 * log(numberOfPoints)) / 
				(((errorTolerance^2)/2) - ((errorTolerance^3)/3));
	result = ceiling(result);
	return(result);
};

#Johnson Linderstrauss Lemma
reduceDimension = function(inputTable, errorTolerance,minDimension){
	nRows = dim(inputTable)[1];
	nColumns = dim(inputTable)[2];
	randomMatrix = matrix(rnorm(nColumns*minDimension,
		mean=0,sd=1/sqrt(minDimension)), nColumns, minDimension);
	cat(paste("\t\tRandom Dim Reduction Matrix : ", 
		dim(randomMatrix)[1],"*",dim(randomMatrix)[2],"\n")); 
	inputTable = (inputTable%*%randomMatrix);
	return(inputTable);
};

#Overflow Bug Fix to the bhattacharya.dist() function available in library("fps")
bhattacharyyaDistanceEigenValues = function (mu1, mu2, Sigma1, Sigma2)  {   
	aggregatesigma <- (Sigma1 + Sigma2)/2;   
	d1 <- mahalanobis(mu1, mu2, aggregatesigma)/8;   
	#d2 <- log(det(as.matrix(aggregatesigma))/sqrt(det(as.matrix(Sigma1)) *    
	#det(as.matrix(Sigma2))))/2;   
	eigenAggregate=log(eigen(as.matrix(aggregatesigma))$values);   
	eigenOne=(log(eigen(as.matrix(Sigma1))$values)/2);   
	eigenTwo=(log(eigen(as.matrix(Sigma2))$values)/2);
	d2 <- sum(eigenAggregate)-sum(eigenOne)-sum(eigenTwo);   
	out <- d1 + (d2/2);   
	return(out); 
};

#Perform PCA Rotation with different number of components
pcaRotation = function(inputTable,significantDigits=3,numberOfCOmponents=-1,
	returnTruncated=FALSE,transposeIfNeeded=TRUE){
	numRows = dim(inputTable)[1];
	numColumns = dim(inputTable)[2];
	transposeDone = FALSE;
	if(transposeIfNeeded){
		if(numRows<numColumns){
		inputTable = (t(inputTable));
		numRows = dim(inputTable)[1];
		numColumns = dim(inputTable)[2];
		transposeDone = TRUE;
		}
	}
	pcaResult = prcomp(inputTable);
	#Plot is extremely useful. This graphic is highly insightful
	plot(cumsum(pcaResult$sdev^2/sum(pcaResult$sdev^2)));

	if(numberOfCOmponents==-1){
		cummulativeContribution = cumsum(pcaResult$sdev^2/sum(pcaResult$sdev^2));
		cummulativeContribution = round(diff(cummulativeContribution),
			significantDigits);
		pcaSignifcantDimension = sum(cummulativeContribution>0)+1;
		cat(paste("\t\tCalculating Contributions and Retaining: ", 
		pcaSignifcantDimension,"out of: ",numColumns," total Dimensions\n"));
		if(!returnTruncated) {
			inputTable = (pcaResult$x[,1:pcaSignifcantDimension]) %*% 
				(t(pcaResult$rotation[,1:pcaSignifcantDimension]));
			#and add the center (and re-scale) back to data
		if(pcaResult$scale != FALSE){
			inputTable <- scale(inputTable, center = FALSE , 
						scale=1/pcaResult$scale)
		}
		if(pcaResult$center != FALSE){
			inputTable <- scale(inputTable, center = -1 * pcaResult$center, 
						scale=FALSE)
		}
		}else {
			inputTable = (pcaResult$x[,1:pcaSignifcantDimension]);
		}
		} else {
			cat(paste("\t\tComponent Count Received and Retaining: ", 
				numberOfCOmponents,"out of: ",numColumns,
												" total Dimensions\n"));
			numberOfCOmponents = min(numberOfCOmponents,numColumns);
			cat(paste("\t\tMinimum dimension: ", 
				numberOfCOmponents,"out of: ",numColumns,
												" total Dimensions\n"));
			inputTable = (pcaResult$x[,1:numberOfCOmponents]);
		}
	if(transposeDone){
		inputTable = (t(inputTable));
	}
	return (inputTable);
};
\end{lstlisting}
\inputencoding{utf8}\end{doublespace}

\end{document}